\documentclass[12pt]{article}
\usepackage[utf8]{inputenc}
\usepackage{graphicx}
\usepackage{amsthm, amssymb}
\usepackage{hyperref}
\usepackage{textcomp} 
\usepackage{setspace}
\usepackage{mathtools}
\usepackage[round]{natbib}
\usepackage{multirow} 
\usepackage{authblk}
\usepackage[usenames,dvipsnames]{xcolor}
\usepackage{pdflscape}

\numberwithin{equation}{section}

\newtheorem{theorem}{Theorem}[section]
\newtheorem{lemma}{Lemma}[section] 
\newtheorem{assumption}{Assumption}[section] 
\newtheorem{definition}{Definition}[section] 
\newtheorem{corollary}{Corollary}[section] 

\newtheorem{remark}{Remark}[section]

\newcommand{\C}{\mathbb{C}}

\newcommand{\N}{\mathbb{N}}
\newcommand{\R}{\mathbb{R}}

\newcommand{\Z}{\mathbb{Z}}

\newcommand{\Dcon}{\stackrel{\mathcal{D}}{\rightarrow}}
\newcommand{\Pcon}{\stackrel{\mathcal{P}}{\rightarrow}}
\newcommand{\Lcon}{\stackrel{L_2}{\rightarrow}}

\newcommand{\ubb}{\boldsymbol{u}}
\newcommand{\vbb}{\boldsymbol{v}}
\newcommand{\jb}{\boldsymbol{j}}

\newcommand{\ob}{\boldsymbol{\omega}}
\newcommand{\blambda}{\boldsymbol{\lambda}}
\newcommand{\btheta}{\boldsymbol{\theta}}

\newcommand{\betaa}{\boldsymbol{\eta}}
\newcommand{\kb}{\boldsymbol{k}}
\newcommand{\sbb}{\boldsymbol{s}}
\newcommand{\xb}{\boldsymbol{x}}
\newcommand{\yb}{\boldsymbol{y}}
\newcommand{\zb}{\boldsymbol{z}}
\newcommand{\tb}{\boldsymbol{t}}

\newcommand{\bbb}{\boldsymbol{b}}

\newcommand{\ellb}{\boldsymbol{\ell}}
\newcommand{\aB}{\boldsymbol{A}}

\newcommand{\cov}{\mathrm{cov}}
\newcommand{\var}{\mathrm{var}}
\newcommand{\cum}{\mathrm{cum}}
\newcommand{\Ex}{\mathbb{E}}

\title{Fourier analysis of spatial point processes}

\author[1]{Junho Yang  \thanks{\textit{Email}: \url{junhoyang@stat.sinica.edu.tw}}
}
\author[2]{Yongtao Guan \thanks{\textit{Email}: \url{guanyongtao@cuhk.edu.cn}}
}

\affil[1]{Institute of Statistical Science, Academia Sinica}
\affil[2]{Shenzhen Research Institute of Big Data, School of Data Science, The Chinese University of Hong Kong, Shenzhen (CUHK-Shenzhen)}

\date{\today}

\begin{document}

\maketitle

\begin{abstract}
In this article, we develop comprehensive frequency domain methods for estimating and inferring the second-order structure of spatial point processes. The main element here is on utilizing the discrete Fourier transform (DFT) of the point pattern and its tapered counterpart. Under second-order stationarity, we show that both the DFTs and the tapered DFTs are asymptotically jointly independent Gaussian even when the DFTs share the same limiting frequencies. Based on these results, we establish an $\alpha$-mixing central limit theorem for a statistic formulated as a quadratic form of the tapered DFT. As applications, we derive the asymptotic distribution of the kernel spectral density estimator and establish a frequency domain inferential method for parametric stationary point processes. For the latter, the resulting model parameter estimator is computationally tractable and yields meaningful interpretations even in the case of model misspecification. We investigate the finite sample performance of our estimator through simulations, considering scenarios of both correctly specified and misspecified models. 
\vspace{0.5em}

\noindent{\it Keywords and phrases:} Bartlett's spectrum, data tapering, discrete Fourier transform, inhomogeneous point processes, stationary point processes, Whittle likelihood. 

\end{abstract}

\section{Introduction} \label{sec:intro}

Spatial point patterns, which are collections of events in space, are increasingly common across various disciplines, including seismology (\cite{p:zhu-04}), epidemiology (\cite{p:gab-09}), ecology (\cite{p:war-10}), and network analysis (\cite{p:dan-22}). A common assumption when analyzing such point patterns is that the underlying spatial point process is second-order stationary or second-order intensity reweighted stationary (\cite{p:bad-00}). Under these assumptions, a majority of estimation procedures for the second-order structure of a spatial point process can be conducted through well-established tools in the spatial domain, such as the pair correlation function and $K$-function (\cite{b:ill-08,p:waa-09}). For a comprehensive review of spatial domain approaches, we refer the readers to \cite{b:mol-04}, Chapter 4.

However, considerably less attention has been devoted to estimations in the frequency domain. In his pioneering work, \cite{p:bar-64} defined the spectral density function of a second-order stationary point process in two-dimensional space and proposed using the periodogram, a squared modulus of the discrete Fourier transform (DFT), as an estimator of the spectral density. For practical implementations, \cite{p:mug-96} provided a guide to using periodograms with illustrative examples. \cite{p:raj-23} derived detailed calculations for the first- and second-order moments of the DFTs and periodograms for fixed frequencies. Despite these advances, theoretical properties of DFTs and periodograms remain largely unexplored. For example, fundamental properties for the DFTs of time series, such as asymptotic uncorrelatedness and asymptotic joint normality of the DFTs (cf. \cite{b:bri-81}, Chapters 4.3 and 4.4), are yet to be rigorously investigated in the spatial point process setting. 

One inherent challenge in conducting spectral analysis of spatial point processes is that spatial point patterns are irregularly scattered. As such, theoretical tools designed for time series data or spatial gridded data cannot be readily extended to the spatial point process setting. One potential solution is to discretize the (spatial or temporal) point pattern using regular bins. This approach allows the application of classical spectral methods from the ``regular'' time series or random fields to the aggregated count data. For example, \cite{p:che-lan-22} developed a frequency domain parameter estimation method for the one-dimentional stationary binned Hawkes process (refer to Section \ref{sec:hawkes} below for details on the Hawkes process). See also \cite{p:shl-22} for the use of the binned Hawkes process to estimate the parameters in the spatial domain. However, aggregating events may introduce additional errors, and there is no theoretical result for binned count processes beyond the stationary Hawkes process case. 

In this article, instead of focusing on the discretized count data, we aim to present a new frequency domain approach for spatial point processes utilizing the ``complete'' information in the process. In Section \ref{sec:spect}, we cover relevant terminologies (Section \ref{sec:prelim}), review the concept of the DFT and periodogram incorporating data tapering (Section \ref{sec:spec-per}), and provide features contained in the spectral density functions and periodograms with illustrations (Section \ref{sec:spec-char}). Building on these concepts, we show in Section \ref{sec:DFT} that under an increasing domain framework and an $\alpha$-mixing condition, the DFTs are asymptotically jointly Gaussian, even when the DFTs share the same limiting frequencies. Therefore, our asymptotic results extend those of \cite{p:raj-23}, who considered only fixed frequencies, enabling us to quantify statistics written in terms of the integrated periodogram (we will elaborate on this below).

A crucial aspect of showing asymptotic joint normality of the DFTs is utilizing spectral analysis tools for irregularly spaced spatial data. \cite{p:mat-09} proposed a novel framework to define the DFT for irregularly spaced spatial data, where the observation locations are generated from a probability distribution on the observation domain. Intriguingly, the DFTs for spatial point processes and those for irregularly spaced spatial data exhibit similar structures. Therefore, tools developed for irregularly spaced spatial data, such as those by \cite{p:ban-09} and \cite{p:sub-18}, are also useful in spatial point process setting.  

Despite the aforementioned similarities, it is important to note that the stochastic mechanisms generating spatial point patterns and irregularly spaced spatial data are very different. The former considers the number of events in a fixed area being random, while the latter determines a (deterministic) number of sampling locations at which the random field is observed. 
Moreover, from a technical point of view, unlike the random field case, the spectral density function $f(\ob)$ of the spatial point process, as given in (\ref{eq:spectral1}) below, is not absolutely integrable. Therefore, the interchange of summations in the expansions of covariances and cumulants of the DFTs is not straightforward.
To reconcile the differences between the spatial point process and irregularly spaced spatial data settings, we introduce in Sections \ref{sec:DFT} and \ref{sec:spec-mean} several new assumptions tailored to the spatial point process setting. In Section \ref{sec:pp-model}, we verify these assumptions for four widely used point process models, namely the Hawkes process, Neyman-Scott point process, log-Gaussian Cox process, and determinantal point process.

Expanding on the theoretical properties of the DFTs for spatial point processes, we also consider parameter estimations. Our main interest lies in parameters expressed in terms of the spectral mean of the form $\int_{D} \phi(\ob) f(\ob) d\ob$, where $D \subset \R^d$ is a prespecified compact region and $\phi(\cdot)$ is a continuous function on $D$. To estimate the spectral mean, we employ the integrated periodogram as defined in (\ref{eq:Aphi}) below. Parameters and estimators in this nature were first considered in \cite{p:par-57} and have since garnered great attention in the time series literature, given that both the kernel spectral density and autocovariance estimator take this general form. In Section \ref{sec:spec-mean}, we derive the central limit theorem (CLT) for the integrated periodogram under an $\alpha$-mixing condition. 
We note that since the integrated periodogram is written as a quadratic form of the DFTs, one cannot directly use the standard techniques to show the CLT for $\alpha$-mixed point processes, as reviewed in \cite{p:bis-19}, Section 1. Instead, we use a series of approximation techniques to prove the CLT for the integrated periodogram. See Appendices \ref{sec:A1-3} and \ref{sec:mixing} in the Supplementary Material for details. As a direct application, in Theorem \ref{thm:KSD-2}, we derive the asymptotic distribution of the kernel spectral density estimator.

Another major application of the integrated periodogram is the model parameter estimation. \cite{p:whi-53} introduced the periodogram-based approximation of the Gaussian likelihood for stationary time series. Subsequently, the concept of Whittle likelihood was extended to lattice (\cite{p:guy-82, p:dah-87}) and irregularly spaced spatial data (\cite{p:mat-09, p:sub-18}). In Section \ref{sec:Whittle}, we develop a procedure to fit parametric spatial point process models based on the Whittle-type likelihood (hereafter, just Whittle likelihood) and obtain sampling properties of the resulting estimator. A noteworthy aspect of our estimator is that it not only estimates the true parameter when the model is correctly specified but also estimates the best fitting parameter when the model is misspecified, where ``best" is defined in terms of the spectral divergence criterion. 
While misspecified first-order intensity models have been considered (e.g., \cite{p:cho-21}), as far as we aware, our result is the first attempt that studies both the first- and second-order model misspecifications for (stationary) spatial point processes. In Section \ref{sec:sim}, we compare the performances of our estimator and two existing estimation methods in the spatial domain through simulations. 

Lastly, proofs, auxiliary results, and additional simulations can be found in the Supplementary Material, \cite{p:yg-24-supp} (hereafter, just Appendix).

\section{Spectral density functions for the second-order stationary point processes} \label{sec:spect}

\subsection{Preliminaries} \label{sec:prelim}
In this section, we introduce the notation used throughout the article and review terminologies related to the mathematical presentation of spatial point processes.

Let $d \in \N$ and let $\R$ and $\C$ be the real and complex fields, respectively. For a set $A$, $n(A)$ denotes the cardinality of $A$ and $A^{n,\neq}$ ($n\in \N$) denotes a set containing all $n$-tuples of pairwise disjoint points in $A$.
 For a vector $\vbb = (v_1, \dots, v_d)^\top \in \C^{d}$, $|\vbb| = \sum_{j=1}^{d} |v_j|$,
$\|\vbb\| = \{\sum_{j=1}^{d} |v_j|^2\}^{1/2}$, and $\|\vbb\|_{\infty} = \max_{1\leq j \leq d} |v_j|$ denote the $\ell_1$ norm, Euclidean norm, and maximum norm, respectively. For vectors $\ubb = (u_1, \dots, u_d)^\top$ and $\vbb= (v_1, \dots, v_d)^\top$ in $\R^d$, $\ubb \cdot \vbb = (u_1v_1, \dots, u_d v_d)^\top$ and $\ubb/\vbb = (u_1/v_1, \dots, u_d/v_d)^\top$, provided $v_1, \dots, v_d \neq 0$. Now we define functional spaces. For $p \in [1,\infty)$ and $k \in \N$, $L^p(\R^{k})$ denotes the set of all measurable functions $g: \R^{k} \rightarrow \C$ such that $\int_{\R^{k}} |g|^p <\infty$. For $g$ in either $L^1(\R^k)$ or $L^2(\R^k)$, the Fourier transform and the inverse Fourier transform are respectively defined as
\begin{equation*}
\mathcal{F}(g)(\cdot) = \int_{\R^k} g(\xb) \exp(i \xb^\top \cdot) d\xb
\quad \text{and} \quad
\mathcal{F}^{-1}(g)(\cdot) = (2\pi)^{-k} \int_{\R^k} g(\xb) \exp(-i \xb^\top \cdot) d\xb.
\end{equation*}

Throughout this article, let $X$ be a simple spatial point process defined on $\R^d$. 
Then, the $n$th-order intensity function (also known as the product density function) of $X$,
denoted as $\lambda_n: \R^{nd} \rightarrow [0,\infty)$, satisfies the following identity
\begin{equation} \label{eq:lambda}
\Ex \bigg[ \sum_{(\xb_1, \dots, \xb_n) \in X^{n} \cap (\R^d)^{n,\neq}}
q(\xb_1, \dots, \xb_n)
\bigg] = 
\int_{} q(\xb_1, \dots, \xb_n)
\lambda_n(\xb_1, \dots, \xb_{n}) \prod_{j=1}^{n} d\xb_j
\end{equation} for any positive measurable function $q:\R^{nd} \rightarrow [0,\infty)$. 
Next, we define the cumulant intensity functions. For $n\in \N$, let $S_n$ be the set of all partitions of $\{1, \dots, n\}$
and for $B = \{i_1, \dots, i_m\} \subseteq \{1, \dots, n\}$ ($m \leq n$), 
let $\lambda_{n(B)}(\xb_{B}) = \lambda_{m}(\xb_{i_1}, \dots, \xb_{i_m})$. Then, the $n$th-order cumulant intensity function (cf. \cite{b:bri-81}, Chapter 2.3) of $X$ is defined as
\begin{equation}\label{eq:gamma}
\gamma_n(\xb_1, \dots, \xb_{n})  = \sum_{\pi \in S_n} (n(\pi) - 1)! (-1)^{n(\pi)-1} \prod_{B \in \pi} \lambda_{n(B)}(\xb_{B}), \quad \xb_1, \dots, \xb_n \in \R^d.
\end{equation} 

\subsection{Spectral density function and its estimator} \label{sec:spec-per}

From now onwards, we assume that $X$ is a $k$th-order stationary ($k \geq 2$) point process. An extension to the nonstationary point process case will be discussed in Section \ref{sec:IRS}.
Under the $k$th-order stationarity, we can define the $n$th-order reduced intensity functions as follows:
\begin{equation} \label{eq:lambda-reduced}
\lambda_n(\xb_1, \dots, \xb_{n}) = \lambda_{n, \text{red}} (\xb_1-\xb_n, \dots, \xb_{n-1}-\xb_n), \quad 
n \in \{1, \dots, k\}.
\end{equation} The $n$th-order reduced cumulant intensity function, denoted as $\gamma_{n,\text{red}}$, is defined similarly, but replacing $\lambda_n$ with $\gamma_n$ in (\ref{eq:lambda-reduced}). In particular, when $n=1$, we use the common notation $\lambda_{1,\text{red}} = \gamma_{1,\text{red}} = \lambda$ and refer to it as the (constant) first-order intensity.

Next, the complete covariance function of $X$ at two locations $\xb_1, \xb_2 \in \R^d$ (which are not necessarily distinct) in the sense of \cite{p:bar-64} is defined as
\begin{equation} \label{eq:cov}
C(\xb_1-\xb_2) = \lambda \delta(\xb_1-\xb_2) +\gamma_{2,\text{red}}(\xb_1-\xb_2),
\end{equation} where $\delta(\cdot)$ is the Dirac-delta function. Heuristically, $C(\xb_1-\xb_2)d\xb_1 d\xb_2$ is the covariance density of $N_X(d\xb_1)$ and $N_{X}(d\xb_2)$, where $N_{X}(\cdot)$ is the counting measure induced by $X$ and $d\xb$ is an infinitesimal region in $\R^d$ that contains $\xb$.
Provided that $\gamma_{2,\text{red}} \in L^1(\R^d)$,
we can define the non-negative valued spectral density function of $X$ by the inverse Fourier transform of $C(\cdot)$ as
\begin{equation}
f(\ob) = (2\pi)^{-d} \int_{\R^d}  C(\sbb) \exp(-i \xb^\top \ob)d\xb 
= (2\pi)^{-d} \lambda  +\mathcal{F}^{-1}(\gamma_{2,\text{red}})(\ob), \quad \ob \in \R^d. 
\label{eq:spectral1}
\end{equation} Here, we use (\ref{eq:cov}) in the second identity. See \cite{b:dal-03}, Sections 8.2 for the mathematical construction of Bartlett's spectral density function.

To estimate the spectral density function, we assume that the point process $X$ is observed within a compact domain (window) $D_n \subset \R^d$ of the form
\begin{equation} \label{eq:Dn}
D_n = [-A_1/2,A_1/2] \times \dots  [-A_d/2,A_d/2], \quad n \in \N,
\end{equation}
 where for $i \in \{1, \dots, d\}$, $\{A_{i}  = A_i(n)\}_{n=1}^{\infty}$ is an increasing sequence of positive numbers. 
Now, we define the DFT of the observed point pattern that incorporates data tapering---a commonly used approach to mitigate the bias inherent in the periodogram (\cite{p:tuk-67}). Let $h(\cdot)$ be a non-negative data taper on $\R^d$ with compact support $[-1/2, 1/2]^d$. For a domain $D_n$ of form (\ref{eq:Dn}), let
\begin{equation}  \label{eq:Hkn}
H_{h,k}^{(n)}(\ob) = \int_{D_n} h(\xb/\aB)^k \exp(-i\xb^\top \ob) d\xb,
\quad k \in \N, \quad \ob \in \R^d,
\end{equation} where $h(\xb / \aB) = h(x_1/A_1, \dots, x_d/A_d)$. 
Let $H_{h,k} =  \int_{[-1/2,1/2]^d} h(\xb)^k d\xb$, $k\in \N$. Throughout the article, we assume $H_{h,k} > 0$, $k \in \N$. Using these notation, the DFT incorporating data taper $h$ is defined as
\begin{equation} \label{eq:mathcalDFT-h}
\mathcal{J}_{h,n}(\ob) = (2\pi)^{-d/2} H_{h,2}^{-1/2} |D_n|^{-1/2} \sum_{\xb \in X \cap D_n}
h(\xb / \aB) \exp(-i \xb^\top \ob), \quad  \ob \in \R^d,
\end{equation} where $|D_n|$ denotes the volume of $D_n$.
We note that by setting $h(\xb) = 1$ on $[-1/2,1/2]^d$, the tapered DFT above encompasses the non-tapered DFT 
\begin{equation} \label{eq:mathcalDFT}
\mathcal{J}_n(\ob) = (2\pi)^{-d/2}|D_n|^{-1/2} \sum_{\xb \in X \cap D_n} \exp(-i\xb^\top \ob), \quad \ob \in \R^d.
\end{equation}
Unless otherwise specified, we will use the term ``DFT'' to indicate the tapered DFT defined as in (\ref{eq:mathcalDFT-h}). Unlike the classical setting in time series or random fields, the DFT is not centered.
By applying (\ref{eq:lambda}), it can be easily seen that $\Ex[\mathcal{J}_{h,n}(\ob)] = \lambda c_{h,n}(\ob)$, where
\begin{equation}  \label{eq:Cn-h}
c_{h,n}(\ob) = (2\pi)^{-d/2} H_{h,2}^{-1/2} |D_n|^{-1/2} H_{h,1}^{(n)} (\ob), \quad \ob \in \R^d,
\end{equation}
 is the bias factor.  
Therefore, the centered DFT is defined as
\begin{equation} \label{eq:Jn-h}
J_{h,n}(\ob) = \mathcal{J}_{h,n}(\ob) - \Ex[\mathcal{J}_{h,n}(\ob)] = \mathcal{J}_{h,n}(\ob) -\lambda c_{h,n}(\ob), \quad \ob \in \R^d.
\end{equation} 
To estimate the unknown first-order intensity, the feasible criterion of $J_{h,n}(\cdot)$ becomes
\begin{equation} \label{eq:Jnmathbb-h}
\widehat{J}_{h,n}(\ob) = \mathcal{J}_{h,n}(\ob) -\widehat{\lambda}_{h,n} c_{h,n}(\ob), \quad \ob \in \R^d,
\end{equation} where $\widehat{\lambda}_{h,n} =  H_{h,1}^{-1}  |D_n|^{-1}\sum_{\xb \in X\cap D_n} h(\xb/\aB)$ ($n \in \N$) is an unbiased estimator of $\lambda$.

Finally, we define the periodogram and its feasible criterion respectively by
\begin{eqnarray} 
I_{h,n}(\ob) = |J_{h,n}(\ob)|^2  \quad \hbox{and} \quad \widehat{I}_{h,n}(\ob) = |\widehat{J}_{h,n}(\ob)|^2,  \quad \ob \in \R^d.
\label{eq:In-feas-h}
\end{eqnarray} 

\subsection{Features of the spectral density functions and their estimators} \label{sec:spec-char}
To motivate the spectral approaches for spatial point processes, the top panel of Figure \ref{fig:motiv} display four spatial point patterns on the observation domain $[-20,20]^2$. These patterns are generated from four different stationary isotropic point process models, exhibiting clustering behaviors in realizations A and B but repulsive behaviors in realizations C and D. All four models share the same first-order intensity, set at 0.5. In the middle panel of Figure \ref{fig:motiv}, we plot the pair correlation function (PCF; middle left) and spectral density function (middle right) for each process.

\begin{figure}[ht!]
\centering
\textbf{Realizations}

\includegraphics[width=0.9\textwidth]{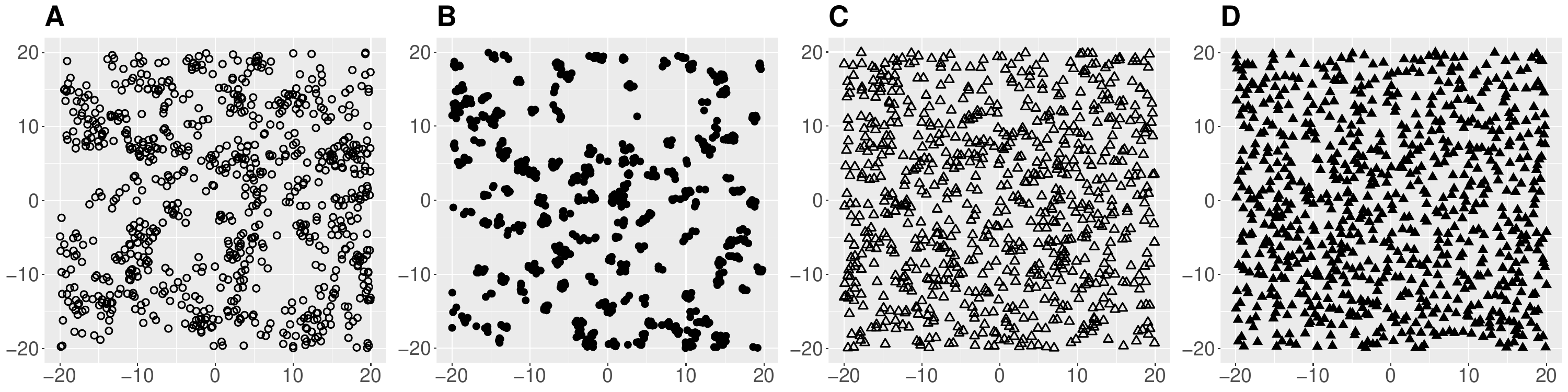}

\textbf{Pair correlation functions and spectral density functions}

\includegraphics[width=0.45\textwidth]{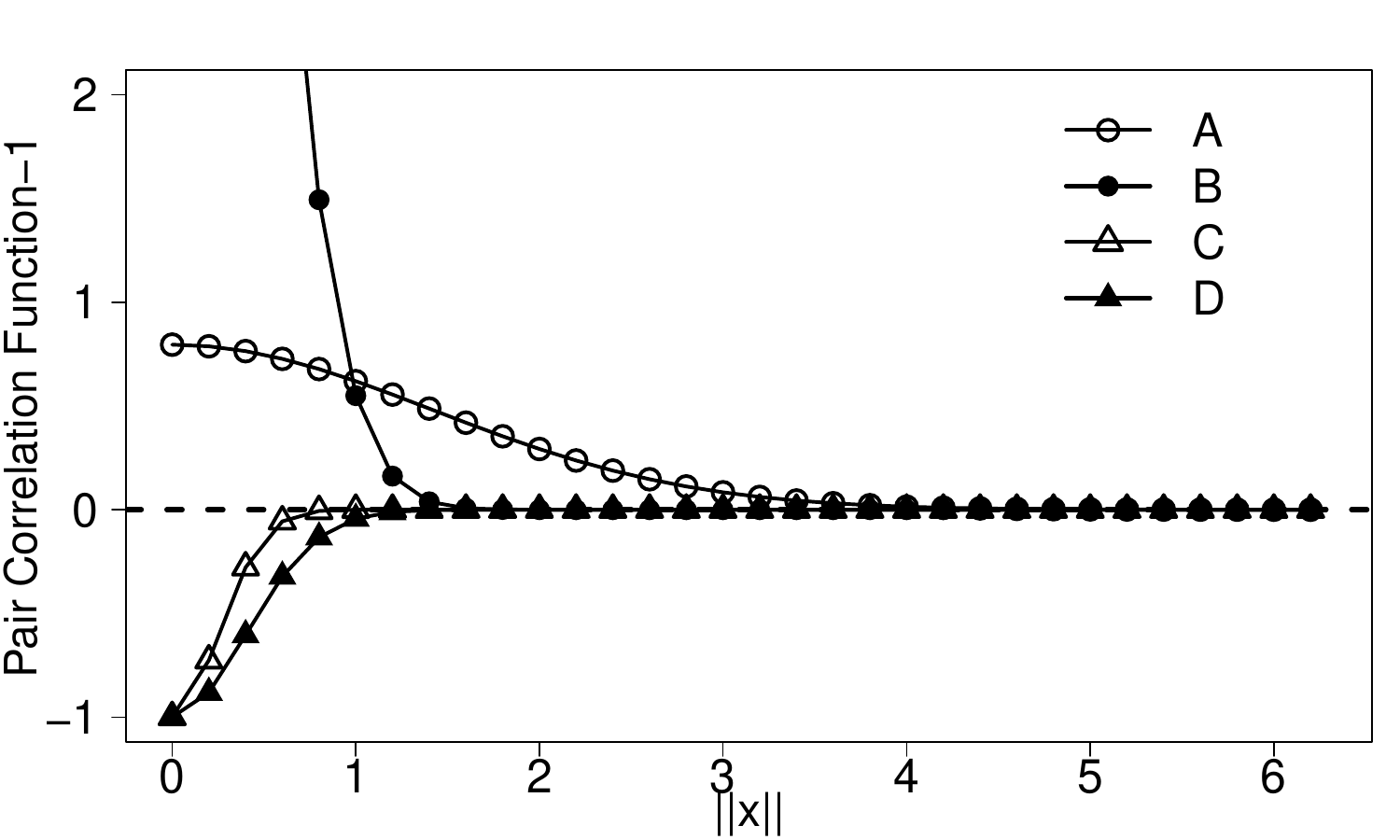}
\includegraphics[width=0.45\textwidth]{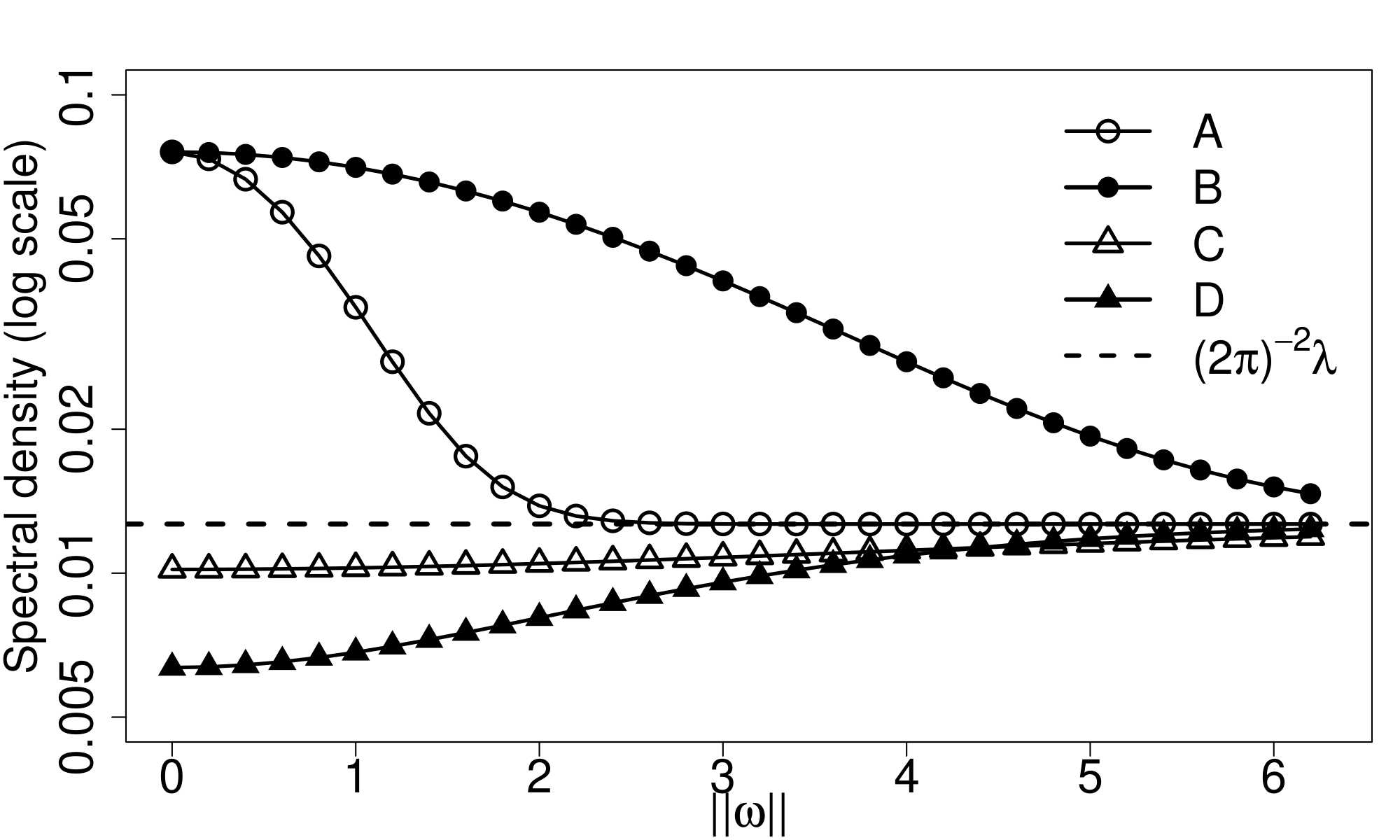}

\textbf{Periodograms}

\includegraphics[width=0.9\textwidth]{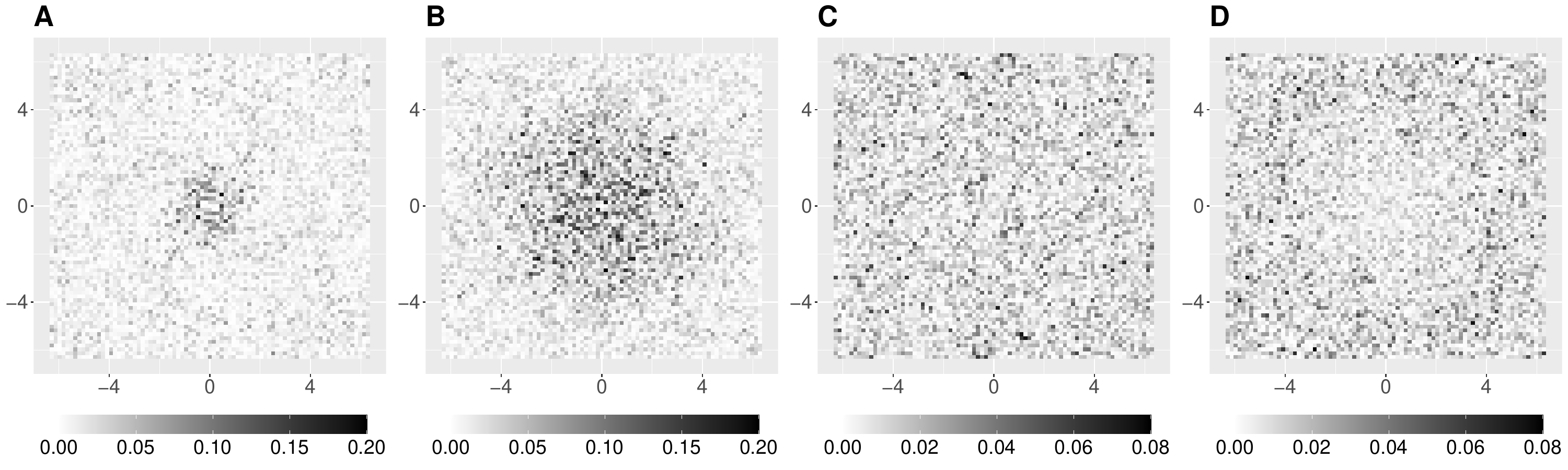}

\caption{\textit{ Top: Realizations of the four different stationary isotropic spatial point processes on the observation domain $[-20,20]^2$.
Middle left: Plot of the pair correlation function $g(\xb)-1$ against $\|\xb\| \in [0,\infty)$ for each model.
Middle right: Plot of the spectral density function $f(\ob)$ in log-scale against $\|\ob\| \in [0,\infty)$ for each model.
Bottom: Plot of the periodogram $\widehat{I}_{h,n}(\ob)$. 
}}
\label{fig:motiv}
\end{figure}

We now investigate how the features of the spatial point patterns are reflected in the spectral density functions. Since the PCF $g(\xb) = \gamma_{2,\text{red}}(\xb)/\lambda^2 +1$, by using (\ref{eq:spectral1}), we have
\begin{equation} \label{eq:PCF-spec}
f(\ob) - (2\pi)^{-d} \lambda = \lambda^2 \mathcal{F}^{-1} (g-1) (\ob), \quad \ob \in \R^2.
\end{equation}

Given the uniqueness of the Fourier transform, the information contained in the spectral density function can be fully recovered from the first-order intensity and the PCF, and vice versa. However, while the PCF $g(\xb)$ captures only "local" information of the point process at a certain lag $\xb$, the spectral density function encapsulates "global" information of the point process, including the first-order intensity (at high frequencies) and overall clustering/repulsive behavior (at low frequencies).

\vspace{0.3em}

\noindent \textit{High frequency information}. Assuming $g-1 \in L^1(\R^2)$ (which is equivalent to the Assumption \ref{assum:C} for $\ell=2$ below), (\ref{eq:PCF-spec}) implies $\lim_{\|\ob\| \rightarrow \infty} f(\ob) = (2\pi)^{-2} \lambda$. Therefore, at high frequencies, the spectral density function contains information about the first-order intensity. 

\vspace{0.3em}

\noindent \textit{Low frequency information}. Note, $g(\xb)-1 >0$ (resp. $<0$) implies the clustering (resp. repulsive) behavior at the fixed lag $\xb \in \R^2$. In the frequency domain, by using (\ref{eq:PCF-spec}), we have $f(\ob) - (2\pi)^{-2} \lambda \approx f(\textbf{0}) - (2\pi)^{-2} \lambda = (2\pi)^{-2} \lambda^2 \int_{\R^2} (g(\xb)-1) d\xb$ for $\|\ob\| \approx \textbf{0}$. Therefore, the spectral density function evaluated at low frequencies above (resp. below) the asymptote indicates the ``overall'' clustering (resp. repulsive) behavior of the point process. 

\vspace{0.3em}

\noindent \textit{Rate of convergence}. Comparing the clustered or repulsive realizations, realization B is more clustered than A while realization D is more repulsive than C. Reflected from the PCFs of the associated models, the PCF of model B is larger at small lag distances but drops more rapidly as the lag distance increases when compared to that of model A, while the PCF of model D is smaller than that of model C. In the frequency domain, one can also extract information on the quality of the clustering and repulsive behaviors. Since the decaying rate of the Fourier transform is related to the smoothness of the original function (cf. \cite{b:fol-99}, Theorem 8.22), the faster (resp. slower) convergence of the spectral density to the asymptote implies a smoother (resp. rougher) PCF. 

\vspace{0.3em}

\noindent \textit{Properties of the periodogram}.
Lastly, we discuss properties of the periodogram as a raw estimator of the spectral density function. The bottom panel of Figure \ref{fig:motiv} plots the periodograms $\widehat{I}_{h,n}(\ob)$ for realizations A--D. Using a computational method described in Appendix \ref{eq:per-compute}, it takes less than 0.06 seconds to evaluate the periodograms on a grid of frequencies $\{(2\pi k_1/40, 2\pi k_2/40): k_1, k_2 \in \{-40, \dots, 40\}\}$ for each model. We observe that for all realizations, the periodograms follow the trend of the corresponding spectral density functions. However, the periodograms are very noisy, indicating that uncorrelated noise fluctuations are added to the trend. We rigorously investigate theoretical properties of the periodograms in Theorem \ref{thm:In-bias} below. To obtain a consistent estimator of the spectral density function, one can locally smooth the periodogram. Detailed computations of the smoothed periodogram with illustrations can be found in Appendix \ref{sec:add-sim} and the theoretical results for the smoothed periodogram are presented in Sections \ref{sec:KSD} and \ref{sec:spec-mean}. %

\section{Asymptotic properties of the DFT and periodogram} \label{sec:asym-periodogram}

\subsection{Asymptotic results} \label{sec:DFT}

In this section, we investigate asymptotic properties of the DFT and periodogram. To do so, we require the following sets of assumptions.


 The first assumption is on the increasing-domain asymptotic framework.
\begin{assumption} \label{assum:A}
Let $D_n$ ($n\in \N$) be a sequence of increasing windows of form (\ref{eq:Dn}) with $\lim_{n \rightarrow \infty} |D_n| = \infty$. Moreover, $D_n$ grows with the same speed in all coordinates of $\R^d$:
\begin{equation} \label{eq:Aij}
A_{i}/A_{j} = O(1), \quad n \rightarrow \infty, \quad i,j \in \{1, \dots, d\}.
\end{equation}
\end{assumption} 

\vspace{0.5em}

The next assumption is on the higher-order cumulants of $X$.

\begin{assumption} \label{assum:C}
Let $\ell \in \{2,3, \dots\}$ be fixed. For $n \in \{1, \dots, \ell\}$,
the cumulant density function $\gamma_{n}$ in (\ref{eq:gamma}) is well-defined and 
\begin{equation} \label{eq:brillinger-mixing}
\sup_{\xb_n \in \R^d} \int_{\R^{d(n-1)}} 
\left|\gamma_{n} (\xb_1, \dots, \xb_{n})\right| d\xb_1 \cdots d\xb_{n-1}
<\infty, \quad n \in \{2, \dots, \ell\}.
\end{equation}
\end{assumption}
For an $\ell$th-order stationary process, Assumption \ref{eq:brillinger-mixing} can be equivalently expressed as $\gamma_{n,\text{red}} \in L^1(\R^{d(n-1)})$, $n \in \{2, \dots, \ell\}$. 

\vspace{0.5em}

The next assumption concerns the $\alpha$-mixing coefficient of $X$ that was first introduced by \cite{p:ros-56}.
For compact and convex subsets $E_i, E_j \subset \mathbb{R}^d$, 
let $d(E_{i}, E_{j})=\inf_{} \{ \|\xb_i - \xb_j\|_{\infty}: \xb_i\in E_{i}, \xb_j \in E_{j}\}$. Then, for $p, q, k \in (0,\infty)$,
the $\alpha$-mixing coefficient of $X$ is defined as
\begin{equation}
\begin{aligned}
\alpha_{p,q}(k) &= \sup_{A_{i}, A_j, E_i, E_j}  \Bigl\{\left| P(A_{i} \cap A_{j}) - P(A_{i}) P(A_{j}) \right|: 
 A_{i} \in \mathcal{F}_{}(E_{i}), A_{j} \in \mathcal{F}_{}(E_{j}), \\
& \qquad \qquad \quad |E_{i}| \leq p, |E_{j}| \leq q, d(E_{i}, E_{j}) \geq k \Bigr\},
\end{aligned}
 \label{eq:alpha}
\end{equation} 
where $\mathcal{F}_{}(E)$ denotes the $\sigma$-field generated by $X$ in $E_{} \subset \mathbb{R}^d$.

\begin{assumption} \label{assum:D}
Let $\alpha_{p,q}(k)$ be the $\alpha$-mixing coefficient of $X$ defined in (\ref{eq:alpha}). We assume one of the following two conditions.
\begin{itemize}
\item[(i)] There exists $\varepsilon>0$ such that $\sup_{p \in (0,\infty)} \alpha_{p,p}(k) / \max(p,1) = O(k^{-d-\varepsilon})$ as $k \rightarrow \infty$.
\item[(ii)] There exists $\varepsilon>2d$ such that $\sup_{p \in (0,\infty)} \alpha_{p,p}(k) / \max(p,1) = O(k^{-d-\varepsilon})$ as $k \rightarrow \infty$.
\end{itemize}
\end{assumption}

The last set of assumptions is on the data taper. 

\begin{assumption} \label{assum:E}
The data taper $h(\xb)$, $\xb \in \R^d$, is non-negative and has a compact support on $[-1/2,1/2]^d$. Moreover,
$h$ satisfies one of the following two conditions below.
\begin{itemize}
\item[(i)] $h$ is continuous on $[-1/2,1/2]^d$.
\item[(ii)] Let $m \in \N$ be fixed. For $\boldsymbol{\alpha}\in \{0,1, \dots\}^d$ with $1 \leq |\boldsymbol{\alpha}| \leq m$, $\partial^{\boldsymbol{\alpha}} h$ exists and is continuous on $\R^d$.
\end{itemize}
\end{assumption}
Assumption \ref{assum:E}(i) encompasses the non-taper case, i.e., $h(\xb) = 1$ on $[-1/2,1/2]^d$. Assumption \ref{assum:E}(ii) for $m=d+1$ is used to show the $|D_n|^{1/2}$-asymptotic normality of the integrated periodogram. To be more precise, it is required to show the Fourier transform of $h$ is absolutely integrable on $\R^d$. As an alternative condition, one can use a slightly different condition:
\begin{equation} \label{eq:h-summable}
\centering
\text{
$\partial^{\boldsymbol{\alpha}} h$ exists
for any $\boldsymbol{\alpha} = (\alpha_1, \dots, \alpha_d) \in \{0,1, 2\}^d$.
}
\end{equation}
Our theoretical results remain unchanged when substituting Assumption \ref{assum:E}(ii) for $m=d+1$ with (\ref{eq:h-summable}). An example of a data taper that satisfies (\ref{eq:h-summable}) is provided in (\ref{eq:taper-ex}). According to our simulation results in Section \ref{sec:sim}, the choice of data taper does not seem to affect the performance of the periodogram-based estimator.


Using the aformentioned sets of assumptions, we now establish the asymptotic joint distribution of the feasible criteria of the DFTs and periodograms.
Recall (\ref{eq:Cn-h}) and (\ref{eq:Jnmathbb-h}). It is easily seen that $\widehat{J}_{h,n}(\textbf{0}) = 0$. Consequently, we exclude the frequency at the origin. Next, we introduce the concept of asymptotically distant frequencies by \cite{p:ban-09}. For two sequences of frequencies $\{\ob_{1,n}\}_{n=1}^{\infty}$ and $\{\ob_{2,n}\}_{n=1}^{\infty}$ on $\R^d$,
we say $\{\ob_{1,n}\}$ and $\{\ob_{2,n}\}$ are asymptotically distant if 
\begin{equation*} 
\lim_{n\rightarrow \infty} |D_n|^{1/d} \|\ob_{1,n} - \ob_{2,n}\| = \infty.
\end{equation*} 
Now, we compute the limit of $\cov(\widehat{J}_{h,n}(\ob_{1,n}), \widehat{J}_{h,n}(\ob_{2,n}))$ and 
$\cov(\widehat{I}_{h,n}(\ob_{1,n}), \widehat{I}_{h,n}(\ob_{2,n}))$
 for two asymptotically distant frequencies.

\begin{theorem}[Asymptotic uncorrelatedness of the DFT and periodogram] \label{thm:In-bias}
Let $X$ be a second-order stationary point process on $\R^d$.
Suppose that Assumptions \ref{assum:A}, \ref{assum:C} (for $\ell=2$), and \ref{assum:E}(i) hold. 
Let $\{\ob_{1,n}\}$ and $\{\ob_{2,n}\}$ be sequences on $\R^d$ such that $\{\ob_{1,n}\}$, $\{\ob_{2,n}\}$, and $\{\textbf{0}\}$ 
are pairwise asymptotically distant. Moreover, let $\{\ob_{n}\}$ be a sequence that is asymptotically distant from $\{\textbf{0}\}$ and converges to the fixed frequency $\ob \in \R^d$.
Then,
\begin{eqnarray}
&& \lim_{n\rightarrow \infty} \cov(\widehat{J}_{h,n}(\ob_{1,n}), \widehat{J}_{h,n}(\ob_{2,n})) = 0.  \label{eq:lim-DFT} \\
&& \text{and} \quad \lim_{n\rightarrow \infty} \var(\widehat{J}_{h,n}(\ob_{n})) = f(\ob).
\label{eq:lim-DFT2}
\end{eqnarray}
If we further assume Assumption \ref{assum:C} for $\ell=4$ holds and $\{\ob_{1,n}\}$ and $\{-\ob_{2,n}\}$ are asymptotically distant, then
\begin{equation} \label{eq:lim-DFT3}
\lim_{n\rightarrow \infty} \cov(\widehat{I}_{h,n}(\ob_{1,n}), \widehat{I}_{h,n}(\ob_{2,n})) = 0 \quad \text{and} \quad
\lim_{n\rightarrow \infty} \var(\widehat{I}_{h,n}(\ob_{n}))
= f(\ob)^2.
\end{equation}
\end{theorem}
\begin{proof} See Appendix \ref{sec:proof0}.
\end{proof}

By using the aforementioned moment properties in conjunction with the $\alpha$-mixing condition, we derive the asymptotic joint distribution of the DFTs and periodograms.

\begin{theorem}[Asymptotic joint distribution of the DFTs and periodograms] \label{thm:asymp-DFT}
Let $X$ be a second-order stationary point process on $\R^d$.
Suppose that Assumptions \ref{assum:A}, \ref{assum:C} (for $\ell=4$), \ref{assum:D}(i), and \ref{assum:E}(i) hold.
For a fixed $r\in \N$, $\{\ob_{1,n}\}$, \dots, $\{\ob_{r,n}\}$ denote
$r$ sequences on $\R^d$ that satisfy the following three conditions: for $(i,j) \in \{1, \dots, r\}^{2,\neq}$,
\begin{itemize}
\item[(1)] $\lim_{n\rightarrow \infty} \ob_{i,n} = \ob_{i} \in \R^d$.
\item[(2)] $\{\ob_{i,n}\}$ is asymptotically distant from $\{\textbf{0}\}$.
\item[(3)] $\{\ob_{i,n} + \ob_{j,n}\}$ and $\{\ob_{i,n} - \ob_{j,n}\}$ are asymptotically distant from $\{\textbf{0}\}$.
\end{itemize}
 Then, we have
\begin{equation*}
\left( \frac{\widehat{J}_{h,n}(\ob_{1,n})}{ (\frac{1}{2}f(\ob_1))^{1/2}}, \dots, \frac{\widehat{J}_{h,n}(\ob_{r,n})}{(\frac{1}{2}f(\ob_r))^{1/2}} \right)
 \Dcon (Z_1, \dots, Z_r), \quad n \rightarrow \infty,
\end{equation*} where $\Dcon$ denotes weak convergence and
$\{Z_k\}_{k=1}^{r}$ are independent standard normal random variables on $\C$. Therefore, by using the continunous mapping theorem, we have
\begin{equation*}
\left(\frac{\widehat{I}_{h,n}(\ob_{1,n})}{\frac{1}{2}f(\ob_1)}, \dots,
\frac{\widehat{I}_{h,n}(\ob_{r,n})}{\frac{1}{2}f(\ob_r)}
 \right)
 \Dcon (\chi^2_1, \dots, \chi^2_r), \quad n \rightarrow \infty,
\end{equation*} where $\{\chi^2_k\}_{k=1}^{r}$ are independent chi-squared random variables with degrees of freedom two. 
\end{theorem}
\begin{proof} See Appendix \ref{sec:proof1}. 
\end{proof}

\begin{remark}
The limit frequencies $\ob_1, \dots, \ob_r \in \R^d$ need not be distinct nor nonzero, as long as the sequences $\{\ob_{1,n}\}, \dots, \{\ob_{r,n}\}$ satisfy the conditions (1)--(3) in Theorem \ref{thm:asymp-DFT}. 
\end{remark}


\subsection{Nonparametric kernel spectral density estimator} \label{sec:KSD}
We observe from Theorem \ref{thm:In-bias} that $\lim_{n\rightarrow \infty} \var (\widehat{I}_{h,n}(\ob_{})) = f(\ob)^2>0$, $\ob \in \R^d \backslash \{\textbf{0}\}$. Therefore, the periodogram is an inconsistent estimator of the spectral density function. In this section, we obtain a consistent estimator of the spectral density function via periodogram smoothing. 

Let $W: \R^d \rightarrow \R$ be a positive continuous and symmetric kernel function with compact support on $[-1/2,1/2]^d$, satisfying $\int_{\R^d} W(\xb) d\xb = 1$ and $\int_{\R^d} W(\xb)^2 d\xb <\infty$. For a bandwidth $\bbb = (b_1, \dots, b_d)^\top \in (0,\infty)^d$, let $W_{\bbb}(\xb) = (b_1 \cdots b_d)^{-1} W(\xb/\bbb)$, $\xb \in \R^d$. For ease of presentation, we set $b_1 = \cdots = b_d = b \in (0,\infty)$. Thus, we write $W_{b}(\xb) = W_{\bbb}(\xb) = b^{-d}  W(b^{-1}\xb)$, $\xb \in \R^d$. Now, we define the kernel spectral density estimator $\widehat{f}_{n,b}(\ob)$ by
\begin{equation} \label{eq:KSD}
\widehat{f}_{n,b}(\ob)
= \int_{\R^d} W_{b}(\ob - \xb) \widehat{I}_{h,n}(\xb)  d\xb
, \quad n \in \N, \quad \ob \in \R^d.
\end{equation}
Below, we show that $\widehat{f}_{n,b}$ consistently estimates the spectral density function.
\begin{theorem} \label{thm:KSD}
Let $X$ be a second-order stationary point process on $\R^d$. Suppose that Assumptions \ref{assum:A}, \ref{assum:C} (for $\ell=4$), and \ref{assum:E}(i) hold. 
Moreover, the bandwidth $b = b(n)$ is such that $\lim_{n\rightarrow \infty} b(n) + |D_n|^{-1}b(n)^{-d} = 0$. 
Then, for $\ob \in \R^d$,
\begin{equation*}
 \widehat{f}_{n,b}(\ob)  \Pcon f(\ob), \quad n \rightarrow \infty,
\end{equation*}
where $\Pcon$ denotes convergence in probability.
\end{theorem}
\begin{proof} See Appendix \ref{sec:proofKSD}. 
\end{proof}

\section{Estimation of the spectral mean statistics} \label{sec:spec-mean}
 Let $D$ be a prespecified compact region on $\R^d$ that does not depend on the index $n \in \N$. For a real continuous function $\phi(\cdot)$ on $D$, our goal is to estimate parameter written in terms of the the spectral mean on $D$:
\begin{equation} \label{eq:Aphi}
A(\phi) = \int_{D} \phi(\ob) f(\ob)  d\ob.
\end{equation} 
A natural estimator of $A(\phi)$ is the integrated periodogram 
\begin{equation} \label{eq:Anphi}
\widehat{A}_{h,n}(\phi) = \int_{D} \phi(\ob) \widehat{I}_{h,n}(\ob) d\ob, \quad n\in\N.
\end{equation}

We briefly mention two examples of estimation problems for spatial point processes that fall under the above framework. First, let $\phi(\cdot) = \phi_{b}(\cdot) = W_{b}(\ob - \cdot)$, where $W_{b}$ is a kernel function with bandwidth $b \in (0,\infty)$ as in Section \ref{sec:KSD}. Since $f(\ob)$ is locally constant in a small neighborhood of $\ob \in \R^d$, as $b = b(n) \rightarrow 0$, we have $A(\phi_b) \approx f(\ob)$. Thus, $\widehat{f}_{n,b}(\ob) = \widehat{A}_{h,n}(\phi_b)$, as in (\ref{eq:KSD}), is our non-parametric estimator of the spectral density. Second, let $\phi(\cdot) = \phi(\cdot; \btheta) = f_{\btheta}^{-1}(\cdot)$, where $\{f_{\btheta}\}$ is a family of spectral density functions with parameter $\btheta \in \Theta$. Then, $A(\phi(\cdot; \btheta)) + \int_{D} \log f_{\btheta}(\ob) d\ob$ denotes the spectral divergence between $f$ and $f_{\btheta}$. An estimator of the spectral divergence is given by $\widehat{A}_{h,n}(\phi(\cdot; \btheta)) + \int_{D} \log f_{\btheta}(\ob) d\ob$, which we refer to as the Whittle likelihood. Please see Section \ref{sec:Whittle} for further details on the Whittle likelihood and the resulting model parameter estimation.

To derive the asymptotic properties of the integrated periodogram, we note that the variance expression of $\widehat{A}_{h,n}(\phi)$ involves the fourth-order cumulant term of $X$. To obtain a simple limiting variance expression of $\widehat{A}_{h,n}(\phi)$, we assume that $X$ is fourth-order stationary, thus both $\lambda_{n,\text{red}}$ and $\gamma_{n,\text{red}}$ in (\ref{eq:lambda-reduced}) are well-defined for $n \in \{1,2,3,4\}$. Following an argument similar to \cite{p:bar-64}, we introduce the complete fourth-order reduced cumulant density function, denoted as $\kappa_{4,\text{red}}(\tb_1-\tb_4, \tb_2-\tb_4,\tb_3-\tb_4)$. Heuristically, this function is defined as a cumulant density function of $N_{X}(d\tb_1)$, $N_{X}(d\tb_2)$, $N_{X}(d\tb_3)$, and $N_{X}(d\tb_4)$, where $\tb_1, \dots, \tb_4 \in \R^d$ may not necessarily be distinct. Explicitly, $\kappa_{4,\text{red}}(\cdot, \cdot,\cdot)$ can be written as a sum of reduced cumulant intensity functions of orders up to four. See (\ref{eq:4th-cum}) in the Appendix for a precise expression. Therefore, under Assumption \ref{assum:C} for $\ell=4$, it can be easily seen that $\kappa_{4,\text{red}} \in L^1(\R^{3d})$, in turn, the fourth-order spectral density of $X$ can be defined as an inverse Fourier transform of $\kappa_{4,\text{red}}$
\begin{equation} \label{eq:4th-spec}
f_4(\ob_1, \ob_2, \ob_3) = (2\pi)^{-3d} \int_{\R^{3d}}\exp \left(-i \sum_{i=1}^{3} \ob_i^\top \xb_i \right) \kappa_{4,\text{red}} (\xb_1,\xb_2,\xb_3) d\xb_1 
d\xb_2 d\xb_3
\end{equation} for $\ob_1, \ob_2, \ob_3 \in \R^d$.
We now introduce the following assumptions on $f$ and $f_4$. For $\boldsymbol{\alpha} = (\alpha_1, \dots, \alpha_d) \in \{0,1, \dots\}^d$, let $\partial^{\boldsymbol{\alpha}} = \left( \partial / \partial \omega_1\right)^{\alpha_1} \cdots  \left( \partial / \partial \omega_d\right)^{\alpha_d}$ be the $\boldsymbol{\alpha}$th-order partial derivative.

\begin{assumption} \label{assum:B}
Suppose the spectral density function $f(\ob)$ of $X$ is well-defined for all $\ob \in \R^d$. Moreover, $f$ satisfies the following:
(i) $f(\ob) - (2\pi)^{-d}\lambda \in L^1(\R^d)$ and (ii) for $\boldsymbol{\alpha}  \in \{0,1,2\}^d$ with $|\boldsymbol{\alpha}|=2$, 
 $\partial^{\boldsymbol{\alpha}} f (\ob)$ exists for $\ob \in \R^d$ and $\sup_{\ob} \left|\partial^{\boldsymbol{\alpha}} f (\ob) \right| <\infty$.
\end{assumption} 
We make two remarks on the above assumptions. Firstly, we observe from (\ref{eq:spectral1}) that $f$ is not absolutely integrable. Instead, given Assumption \ref{assum:B}(i), the spectral density function, when appropriately ``shifted'', admits the Fourier transformation
\begin{equation} \label{eq:cov-spectral}
\gamma_{2,\text{red}}(\xb) = 
\mathcal{F}(f(\ob) - (2\pi)^{-d}  \lambda ) = 
\int_{\R^d}  \left\{ f(\ob) - (2\pi)^{-d} \lambda \right\} \exp(i\xb^\top \ob)d\ob.
\end{equation}
Secondly, for some parametric models (e.g., the log-Gaussian Cox processes in Section \ref{sec:LGCP} below), a closed-form expression for $f$ is not available, while $\gamma_{2,\text{red}}$ has an analytic form. In this case, a sufficient condition for Assumption \ref{assum:B}(i) to hold is that $\gamma_{2,\text{red}}$ has continuous partial derivatives up to order $(d+1)$, as per \cite{b:fol-99}, Theorem 8.22 (see also page 257 of the same reference). Moreover, since $(\partial^2 f/ \partial \omega_i \partial \omega_j)(\ob) = -(2\pi)^{-d}\int_{\R^d} x_i x_j \gamma_{2,\text{red}}(\xb) \exp(-i \ob^\top \xb) d\xb$, Assumption \ref{assum:B}(ii) holds if $|\xb|^{2} \gamma_{2,\text{red}}(\xb) \in L^{1}(\R^{d})$.

\begin{assumption} \label{assum:F}
Suppose that $f_4(\ob_1, \ob_2, \ob_3)$, is well-defined for all $\ob_1, \ob_2, \ob_3 \in \R^d$. Moreover, $f_4$ satisfies the following: 
(i) $f_4 - (2\pi)^{-3d} \lambda \in L^1(\R^{3d})$ and (ii) $f_4(\ob_1, \ob_2, \ob_3)$ is twice partial differentiable with respect to $\ob_2 $ and the second-order partial derivative is bounded above.
\end{assumption} 
By using the same auguments above, sufficient conditions for Assumption \ref{assum:F} to hold in terms of the differentiability and integrability of $\gamma_{n,\text{red}}$ ($n \in \{2,3,4\}$) also can be easily derived.

Now, we are ready to state our main theorem addressing the asymptotic normality of $\widehat{A}_{h,n}(\phi)$.

\begin{theorem}[Asymptotic distribution of the integrated periodogram] \label{thm:A1}
Let $X$ be a fourth-order stationary point process on $\R^d$, that is, (\ref{eq:lambda-reduced}) is satisfied for $k=4$.
Then, the following three assertions hold.
\begin{itemize}
\item[(i)] Suppose that Assumptions \ref{assum:A}, \ref{assum:C} (for $\ell=2$), \ref{assum:E}(ii) (for $m=1$), and \ref{assum:B}(ii) hold. Then,
\begin{equation*}
\Ex [ \widehat{A}_{h,n}(\phi) ] = A(\phi) + O(|D_n|^{-2/d}), \quad n\rightarrow \infty.
\end{equation*}

\item[(ii)] Suppose that Assumptions \ref{assum:A}, \ref{assum:C} (for $\ell=4$), \ref{assum:B}(i), and \ref{assum:F} hold. Furthermore, the data taper $h$ is constant on $[-1/2,1/2]^d$ or satisfies Assumption \ref{assum:E}(ii) for $m=d+1$. Then,
\begin{equation*}
\lim_{n \rightarrow \infty} |D_n| \var ( \widehat{A}_{h,n}(\phi)) =  (2\pi)^{d} (H_{h,4} / H_{h,2}^2) (\Omega_1 + \Omega_2),
\end{equation*} where
\begin{equation}
\begin{aligned}
\Omega_1 &= \int_{D} \phi(\ob) \left( \phi(\ob) + \phi(-\ob) \right) f(\ob)^2 d\ob \\
\text{and} \quad 
\Omega_2 &=   \int_{D^2} \phi(\ob_1) \phi(\ob_2) f_4(\ob_1, -\ob_1, \ob_2) d\ob_1 d\ob_2.
\end{aligned}
 \label{eq:Omegas}
\end{equation}

\item[(iii)]
Now, let $d \in \{1,2,3\}$. Suppose that Assumptions \ref{assum:A}, \ref{assum:C} (for $\ell=8$), \ref{assum:D}(ii), 
\ref{assum:E}(ii) (for $m=d+1$), \ref{assum:B}, and \ref{assum:F} hold. Then,
\begin{equation*}
|D_n|^{1/2} ( \widehat{A}_{h,n}(\phi) - A(\phi)) \Dcon \mathcal{N}\left(0, (2\pi)^{d} (H_{h,4} / H_{h,2}^2) (\Omega_1 + \Omega_2) \right)
, \quad n \rightarrow \infty.
\end{equation*}
\end{itemize}

\end{theorem}
\begin{proof} See Appendix \ref{sec:proofA1}. 
\end{proof}

\begin{remark}[Estimation of the asymptotic variance] \label{rmk:SS}
Since $\Omega_1$ and $\Omega_2$ above are unknown functions of the spectral density and fourth-order spectral density function, the asymptotic variance of $\widehat{A}_{h,n}(\phi)$ needs to be estimated. In Appendix \ref{sec:subsampling}, we provide details on the estimation of $(2\pi)^{d} (H_{h,4} / H_{h,2}^2) (\Omega_1 + \Omega_2)$ using the subsampling method.
\end{remark}


As a direct application of the above theorem, we derive the asymptotic distribution of the kernel spectral density estimator. Recall (\ref{eq:KSD}).

\begin{theorem} \label{thm:KSD-2}
Let $X$ be a fourth-order stationary point process on $\R^d$, that is, (\ref{eq:lambda-reduced}) is satisfied for $k=4$.
Suppose that Assumptions \ref{assum:A}, \ref{assum:C} (for $\ell=8$), \ref{assum:D}(ii),  \ref{assum:E}(ii) (for $m=d+1$), \ref{assum:B}, and \ref{assum:F} hold. Moreover, the bandwidth $b = b(n)$ is such that 
\begin{equation*}
\lim_{n\rightarrow \infty} b(n) + |D_n|^{-1}b(n)^{-d} = 0 \text{~~and~~} \lim_{n\rightarrow \infty} |D_n|^{1/2} b(n)^{d/2} (|D_n|^{-2/d} + b(n)^{2}) = 0.
\end{equation*}
Let $W_2 = \int_{\R^d} W(\xb)^2 d\xb$. Then, for $\ob \in \R^d \backslash \{\textbf{0}\}$,
\begin{equation*}
\sqrt{|D_n|b^d} \big( \widehat{f}_{n,b}(\ob) - f(\ob) \big) \Dcon \mathcal{N}\left(0, (2\pi)^{d} (H_{h,4} / H_{h,2}^2)W_2  f(\ob)^2 \right),
\quad n\rightarrow \infty
\end{equation*} and for $\ob = \textbf{0}$,
\begin{equation*}
\sqrt{|D_n|b^d} \big( \widehat{f}_{n,b}(\ob) - f(\ob) \big) \Dcon \mathcal{N}\left(0, 2(2\pi)^{d} (H_{h,4} / H_{h,2}^2) W_2 f(\ob)^2 \right),
\quad n\rightarrow \infty.
\end{equation*}
\end{theorem}
\begin{proof}See Appendix \ref{sec:proofKSD-2}. 
\end{proof}

\section{Examples of spatial point process models} \label{sec:pp-model}

In this section, we provide examples of four widely used stationary spatial point process models and specify the conditions under which each model satisfies Assumptions \ref{assum:C}, \ref{assum:D}, and \ref{assum:B}, which are required to establish the asymptotic results in Sections \ref{sec:asym-periodogram} and \ref{sec:spec-mean}.

\subsection{Example I: Hawkes processes} \label{sec:hawkes}
The Hawkes process is a doubly stochastic process on $\R$ characterized by self-exciting and clustering properties (\cite{p:haw-71, p:haw-71a}). A stationary Hawkes process is described by a conditional intensity function of the form
$\lambda(t) = \nu + \int_{-\infty}^{t} \eta(t-u) N_{X}(du)$, $t \in \R$. Here, $\nu > 0$ is the immigration intensity and $\eta: [0,\infty) \rightarrow [0,\infty)$ is a measurable function satisfying $\int_{\R} \eta(u) du < 1$, referred to as the reproduction function. 
The spectral density function of the Hawkes process, as stated in \cite{b:dal-03}, Example 8.2(e), is given by
$f(\omega) =(2\pi)^{-1}\lambda \left| 1-\mathcal{F}(\eta)(\omega) \right|^{-2}$, $\omega \in \R$.
Here, $\lambda = \Ex[\lambda(t)] = \nu / \{1-\int_{\R}\eta(u)du \} \in (0,\infty)$ denotes the first-order intensity of the corresponding Hawkes process.

To verify Assumption \ref{assum:C} for Hawkes processes, \cite{p:jov-15} provides explicit expressions of the cumulant intensity functions, thus one can check (\ref{eq:brillinger-mixing}) for the specified cumulant intensity functions. To assess the $\alpha$-mixing conditions, \cite{p:che-lan-22}, Theorem 1, states that if $\int_{0}^{\infty} u^{1+\delta}\eta(u)du <\infty$, for some $\delta > 0$, then $\sup_{p,q \in (0,\infty)} \alpha_{p,q}(k) = O(k^{-\delta})$, as $k \rightarrow \infty$. Therefore, if $\eta(\cdot)$ has a $(2+\delta)$th-moment (resp. $(4+\delta)$th-moment) for some $\delta >0$, then the corresponding Hawkes process satisfies Assumption \ref{assum:D}(i) (resp. \ref{assum:D}(ii)). 
Lastly, to check for Assumption \ref{assum:B}, one can employ the expression of $f(\ob)$, provided the Fourier transform of $\eta(\cdot)$ has a closed form expression. For example, if the reproduction function has a form $\eta(\cdot) = \alpha \exp(-\beta \cdot)$ for some $0<\alpha<\beta$, then $f(\omega) - (2\pi)^{-1} \lambda =\alpha(2\beta-\alpha)/ \{ (\beta-\alpha)^2 +\omega^2 \}$, $\omega \in \R$. Therefore, the defined $f$ satisfies Assumption \ref{assum:B}.


\subsection{Example II: Neyman-Scott point processes} \label{sec:ex2-NS}
A Neyman-Scott (N-S) process (\cite{p:ney-sco-58}) is a special class of the Cox process, where the random latent intensity field is given by
$\Lambda(\ubb) = \sum_{\xb \in \Phi} \alpha k(\ubb-\xb)$, $\ubb \in \R^d$. Here, $\Phi$ is a homogeneous Poisson process with intensity $\kappa >0$ and $k(\cdot)$ is a probability density (kernel) function on $\R^d$. Common choices for $k$ include $k_1(\cdot) = (2\pi \sigma^2)^{-d/2} \exp(-\|\cdot\|^2/ (2\sigma^2))$ and $k_2(\cdot) = (r^d s_d)^{-1} I(\|\cdot\| \leq r)$, where $s_d$ is the volume of the unit ball on $\R^d$. The corresponding N-S process for the kernel functions $k_1$ and $k_2$ are known as the Thomas cluster process and Mat\'{e}rn cluster process. 
Using, for example, \cite{p:cha-97}, Equation (37) (see also \cite{b:dal-03}, Exercise 8.2.9(c)), the spectral density function for the N-S process is given by
$f(\ob) = (2\pi)^{-d} \lambda \left( 1 + \alpha |\mathcal{F}(k)(\ob)|^2\right)$, $\ob \in \R^d$,
where $\lambda = \kappa \alpha$ denotes the first-order intensity of the corresponding N-S process.
In particular, the spectral density function of the Thomas cluster process is given by
\begin{equation} \label{eq:spectral-thomas}
f^{(TCP)}(\xb) = (2\pi)^{-d} \lambda \left\{ 1+ \alpha \exp(-\sigma^2 \|\xb\|^2)\right\}, \quad \ob \in \R^d.
\end{equation}

\cite{p:pro-13} (page 398) showed that the N-S process satisfies Assumption \ref{assum:C} for all $\ell \in \N$. Moreover, according to Lemma 1 of the same reference, if $k(\xb) = O(\|\xb\|^{-2d-\varepsilon})$ as $\|\xb\|\rightarrow \infty$ for some $\varepsilon > 0$ (resp. $\varepsilon >2d$), then the corresponding N-S process satisfies Assumption \ref{assum:D}(i) (resp. \ref{assum:D}(ii)). Assumption \ref{assum:B} can be verified for specific kernel functions with known forms of their Fourier transform. In particular, spectral density functions associated with Thomas cluster processes and Mat\'{e}rn cluster processes satisfy Assumption \ref{assum:B}. 

\subsection{Example III: Log-Gaussian Cox Processes} \label{sec:LGCP}
The log-Gaussian Cox process (LGCP; \cite{p:mol-98}) is a special class of the Cox process where the logarithm of the intensify field is a Gaussian random field. Let $X$ be a stationary LGCP driven by the intensity field $\Lambda(\cdot)$ and let $R(\xb) = \cov(\log \Lambda(\xb), \log \Lambda(\textbf{0}))$, $\xb \in \R^d$, be the autocovariance function of the log intensity field. Then, by using \cite{p:mol-98}, Equation (4), the second-order reduced cumulant is given by $\gamma_{2,\text{red}}(\xb) = \exp\{R(\xb)\}-1$, $\xb \in \R^d$.

Assumption \ref{assum:C} holds for arbitrary $\ell \in \N$, provided $R(\cdot) \in L^1(\R^d)$. See \cite{p:zhu-23}, Lemma D.1. Concerning Assumption \ref{assum:D}, \cite{b:dou-94}, page 59, stated that for a stationary random field on $\Z^d$, if $R(\xb) = O(\|\xb\|^{-2d-\varepsilon})$ as $\|\xb\|\rightarrow \infty$ for some $\varepsilon > 0$ (resp. $\varepsilon >2d$), then the corresponding LGCP on $\Z^d$ satisfies Assumption \ref{assum:D}(i) (resp. \ref{assum:D}(ii)). However, a general $\alpha$-mixing condition for stationary Gaussian random fields on $\R^d$ is not readily available. Lastly, since the second-order reduced cumulant function has a closed-form expression in terms of $R(\cdot)$, Assumption \ref{assum:B} can be easily verified when $R(\cdot)$ is specified (see the remarks after Assumption \ref{assum:B}).

\subsection{Example IV: Determinantal point processes} \label{sec:ex4-DPP}
The Determinantal point process (DPP), first introduced by \cite{p:mac-75}, has the intensity function characterized by the determinant of some function. To be specific, a stationary determinantal point process induced by a kernel function $K:\R^d \rightarrow \infty$, denoted as DPP($K$), has the reduced intensity function
$\lambda_{n,\text{red}}(\xb_1-\xb_n, \dots, \xb_{n-1}-\xb_{n})
= \det ( K(\xb_i-\xb_j)_{1\leq i, j \leq n})$, $n \in \N$, $\xb_1, \dots, \xb_n \in \R^d$. Here, we assume that the kernel function $K$ is symmetric, continuous, and belongs to $L^{2}(\R^d)$ satisfying $\mathcal{F}(K)(\ob) \in [0,1]$, $\ob \in \R^d$. Then, the second-order reduced cumulant function and the spectral density function of DPP($K$) are respectively given by
$\gamma_{2,\text{red}}(\xb) = - |K(\xb)|^2$ and
$f(\ob) = (2\pi)^{-d} K(0) - \mathcal{F}^{-1} (|K|^2)(\ob)$. Therefore, the DPPs exhibit a repulsive behavior. For example, choosing the Gaussian kernel $K^{(G)}(\xb) = \lambda \exp(-\|\xb\|^2/\rho^2)$ with parameter restriction $ 0< \rho  \leq 1/(\sqrt{\pi} \lambda^{1/d})$, the spectral density function corresponds to $DPP(K^{(G)})$ is given by
\begin{equation} \label{eq:DPP-G}
f^{(GDPP)} (\ob) = (2\pi)^{-d} \left\{ \lambda - \lambda^2 (\pi \rho^2/2)^{d/2} \exp(-\rho^2 \|\ob\|^2 / 8)  \right\}, \quad \ob \in \R^d.
\end{equation}

\cite{p:bis-16} showed that DPP($K$) satisfies Assumption \ref{assum:C} for any $\ell \in \N$. Moreover, \cite{p:poi-19} showed $\alpha_{p,p}(k)/p \leq C \int_{k}^{\infty} t^{d-1} \sup_{\|\xb\| = t}|\gamma_{2,\text{red}}(\xb)| dt$, $p \in (0,\infty)$. Therefore, if there exists $\varepsilon >0$ (resp. $\varepsilon>2d$) such that
$\sup_{\|\xb\| = t} |K(\xb)| = O(t^{-d - (\varepsilon/2)})$ as $t\rightarrow \infty$, then DPP($K$) satisfies Assumption \ref{assum:D}(i) (resp. \ref{assum:D}(ii)). Therefore, DPP with the Gaussian kernel $K^{(G)}$ satisfies Assumption \ref{assum:D}(ii). Lastly, Assumption \ref{assum:B} can be verified provided $|K|^2$ has a known Fourier transform, including the case of $f^{(GDPP)}(\cdot)$ in (\ref{eq:DPP-G}).

\section{Frequency domain parameter estimation under possible model misspecification} \label{sec:Whittle}

As an application of the CLT results in Section \ref{sec:spec-mean}, we turn our attention to inferences for spatial point processes in the frequency domain through the Whittle likelihood.
Let $\{X_{\btheta}\}$ be a family of second-order stationary spatial point processes with parameter $\btheta \in \Theta \subset \R^p$. The associated spectral density function of $X_{\btheta}$ is denoted as $f_{\btheta}$, $\btheta \in \Theta$. Then, we fit the model with spectral density $f_{\btheta}$ using the pseudo-likelihood given by
\begin{equation} \label{eq:Whittle}
L_n(\btheta) = \int_{D} \left( \frac{\widehat{I}_{h,n}(\ob)}{f_{\btheta}(\ob)} + \log f_{\btheta} (\ob) \right) d\ob, \quad n \in \N, \quad \btheta \in \Theta \subset \R^p.
\end{equation} Here, $D$ is a prespecified compact and symmetric region on $\R^d$. 
Let
\begin{equation} \label{eq:thetahat}
\widehat{\btheta}_n = \arg\min_{\btheta \in \Theta} L_n(\btheta), \qquad n \in \N,
\end{equation} be our proposed model parameter estimator. Here, we do not necessarily assuming the existence of $\btheta_{0} \in \Theta$ such that the true spectral density function $f = f_{\btheta_0}$. Since the periodogram is an unbiased estimator of the ``true" spectral density, the best fitting parameter could be
\begin{equation} \label{eq:mathcalL}
\btheta_0 = \arg\min_{\btheta \in \Theta} \mathcal{L}(\btheta), \quad \text{where} \quad
\mathcal{L}(\btheta)=\int_{D} \left( \frac{f(\ob)}{f_{\btheta}(\ob)} + \log f_{\btheta}(\ob) \right) d\ob.
\end{equation}
The best fitting parameter $\btheta_0$ has a clear interpretation in terms of the spectral divergence criterion, as $\mathcal{L}(\btheta)$ above computes the spectral (information) divergence between the true and conjectured spectral densities. Investigations into the properties of the Whittle estimator under model misspecification for time series and random fields can be found in \cite{p:dah-96, p:dah-00, p:sub-21}. We mention that, in the case where the mapping $\btheta \rightarrow f_{\btheta}$ is injective and the model is correctly specified, it can be easily seen that $\btheta_0$ is uniquely determined and satisfies $f = f_{\btheta_0}$. 

Now, we assume the following for the parameter space.

\begin{assumption} \label{assum:G}
The parameter space $\Theta$ is a compact subset of $\R^p$, $p \in \N$. The parametric family of spectral density functions $\{f_{\btheta}\}$ is uniformly bounded above and bounded below from zero. $f_{\btheta}(\ob)$ is twice differentiable with respect to $\btheta$ and its first and second derivatives are continuous on $ \Theta \times D$. 
$\btheta_0$ in (\ref{eq:mathcalL}) is uniquely determined and lies in the interior of $\Theta$.
Lastly, $\widehat{\btheta}_n$ in (\ref{eq:thetahat}) exists for all $n\in \N$ and lies in the interior of $\Theta$.
\end{assumption}

To obtain the asymptotic variance of $\widehat{\btheta}_n$, for $\btheta \in \Theta$, let
\begin{equation}
\begin{aligned} 
\Gamma(\btheta) &=
\frac{1}{2(2\pi)^{d}} \int_{D} \bigg[ 
\left(f(\ob) - f_{\btheta}(\ob) \right) \nabla^2 f_{\btheta}^{-1} (\ob) 
+( \nabla \log f_{\btheta} (\ob)) ( \nabla \log f_{\btheta} (\ob))^\top 
\bigg] d\ob, \\
S_1(\btheta) &= \frac{1}{2(2\pi)^{d}} \int_{D} f(\ob)^2 ( \nabla f_{\btheta}^{-1} (\ob)) ( \nabla f_{\btheta}^{-1} (\ob))^\top d\ob, \quad \text{and} \\
S_2(\btheta) &= \frac{1}{4(2\pi)^{d}} \int_{D^2} f_{4}(\ob_1, -\ob_1, \ob_2)
( \nabla f_{\btheta}^{-1} (\ob)) ( \nabla f_{\btheta}^{-1} (\ob))^\top d\ob_1 d\ob_2.
\end{aligned}
 \label{eq:Gamma-mat}
\end{equation} Here, $\nabla f_{\btheta}$ and $\nabla^2 f_{\btheta}$ are the first- and second-order derivatives of $f_{\btheta}$ with respect to $\btheta$, respectively, and
$f_{4}$ denotes the (true) fourth-order spectral density of $X$. In the scenario where the model is correctly specified, we have $\Gamma(\btheta_0) = S_1(\btheta_0)$.
 
The following theorem addresses the asymptotic behavior of the our proposed estimator under possible model misspecification.

\begin{theorem} \label{thm:Whittle}
Let $X$ be a fourth-order stationary point process on $\R^d$, that is, (\ref{eq:lambda-reduced}) is satisfied for $k=4$.
Suppose that Assumptions \ref{assum:A}, \ref{assum:C} (for $\ell=4$), 
\ref{assum:E}(ii) (for $m=1$), \ref{assum:B}, 
and \ref{assum:G} hold.
 Then, 
\begin{equation}\label{eq:theta-consist}
\widehat{\btheta}_n \Pcon \btheta_0, \quad n \rightarrow \infty.
\end{equation} 
Now, let $d \in \{1,2,3\}$. Suppose $\Gamma(\btheta_0)$ in (\ref{eq:Gamma-mat}) is invertible.
 Then, under Assumptions \ref{assum:A}, \ref{assum:C} (for $\ell=8$), \ref{assum:D}(ii), \ref{assum:E}(ii) (for $m=d+1$), \ref{assum:B}, \ref{assum:F}, and \ref{assum:G},
 we have
\begin{equation} \label{eq:theta-CLT}
|D_n|^{1/2} (\widehat{\btheta}_n - \btheta_0) \Dcon \mathcal{N} \left( \textbf{0}, 
(H_{h,4}/H_{h,2}^2)
\Gamma(\btheta_0)^{-1} \left( S_1(\btheta_0) + S_2(\btheta_0) \right) \Gamma(\btheta_0)^{-1}\right),
~~~ n\rightarrow \infty.
\end{equation} 
\end{theorem}
\begin{proof}
See Appendix \ref{sec:proof2}.
\end{proof}

\begin{remark} \label{rmk:debiased}
The condition on $d$ being less than four in the asymptotic normality of $\widehat{\btheta}_n$ is required to ensure that the bias of $|D_n|^{1/2} \widehat{\btheta}_n$ converges to zero. This restriction is also imposed in the random fields literature (e.g., \cite{p:dah-87, p:mat-09}). By using the debiasing technique considered in \cite{p:gui-22}, one can establish the asymptotic normality of the ``debiased'' Whittle estimator for all $d \in \N$. The details will be reported in a future study.
\end{remark}

\begin{remark}
We provide a summary of the procedure for estimating the asymptotic variance of $\widehat{\btheta}_n$. Recall (\ref{eq:theta-CLT}). $\Gamma(\btheta_0)$ can be easily estimated by replacing $f(\ob)$ and $\btheta_0$ with $\widehat{f}_{n,b}(\ob)$ and $\widehat{\btheta}_n$, respectively, in (\ref{eq:Gamma-mat}). To estimate $S_1(\btheta_0) + S_2(\btheta_0)$ one can employ the subsampling variance estimation method for $A_{h,n}( \nabla f_{\widehat{\btheta}_n}^{-1})$, as described in Appendix \ref{sec:subsampling} (see also, Remark \ref{rmk:SS}). The theoretical properties of this estimated variance will not be investigated in this article.
\end{remark}



\section{Simulation studies} \label{sec:sim}

To corroborate our theoretical results, we conduct some simulations on the model parameter estimation. Additional simulation results can be found in Appendices \ref{eq:per-compute}--\ref{sec:TCP-mis2}. Due to space constraints, we only consider the following two point process models on $\R^2$:
\begin{itemize}
\item Stationary Thomas cluster processes (TCPs) with parameter $\btheta = (\kappa, \alpha, \sigma^2)^\top$ as in Section \ref{sec:ex2-NS}. The spectral density function of TCP, denoted as $f^{(TCP)}_{\btheta}(\cdot)$, is given in (\ref{eq:spectral-thomas}) with the first-order intensity $\lambda = \kappa \alpha$. TCP exhibits clustering behavior.
\item 
Stationary determinantal point processes with Gaussian kernel (GDPPs) with parameter $\btheta = (\lambda, \rho^2)^\top$ as in Section \ref{sec:ex4-DPP}. The spectral density function of GDPP, denoted as $f^{(GDPP)}_{\btheta}(\cdot)$, is given in (\ref{eq:DPP-G}).
GDPP exhibits repulsive behavior.
\end{itemize}

For each model, we generate spatial point patterns within the observation domain (window) $D_n = [-A/2, A/2]^2$ for varying side lengths $A \in \{10,20,40\}$.
To assess the performance of the different parameter estimation methods, we compare our estimator as in (\ref{eq:thetahat}) with two existing methods in the spatial domain: the maximum likelihood-based method (ML) and the minimum contrast method (MC). 

Specifically, for the ML method, given the intractable nature of the likelihood functions for TCPs and GDPPs, we maximize the log-Palm likelihood (\cite{p:tan-08}) for the TCPs and use the asymptotic approximation of the likelihood (\cite{p:poi-lav-23}) for GDPPs. 
For the MC method, we minimize the contrast function of form
$K(\btheta) = \int_{r_{\text{min}}}^{r_{\text{max}}} |g(t;\btheta)^c - \widehat{g}(t)^c |^{2} dt$ where $g(\cdot; \btheta)$ denotes the parametric pair correlation function (PCF) for the isotropic process
 and $\widehat{g}(\cdot)$ is an estimator of the PCF. Since the PCF of the TCP model does not include the parameter $\alpha$, we do not include the estimation of $\alpha$ in the MC method for TCPs.  Similarly, as the PCF of GDPP is solely a function of $\rho^2$, we do not include the estimation of $\lambda$ in the MC method for GDPPs.
Finally, following the guidelines from \cite{p:bis-17} (see also, \cite{b:dig-13}), for MC methods, we choose the tuning parameters $r_{\text{min}} = 0.01 A$ and $r_{\text{max}} = 0.25 A$ where $A \in \{10,20,40\}$ is the length of the window and $c=0.25$ for TCPs and $c=0.5$ for GDPPs. 

Lastly, all simulations are conducted over 500 independent replications of spatial point patterns and for each replication, we compute the three previously mentioned three model parameter estimators.

\subsection{Practical guidelines for the frequency domain method} \label{sec:practice}
We now discuss three practical issues arising during the evaluation of our estimator.

\vspace{0.5em}

\noindent \textit{Choice of the data taper}. \hspace{0.1em}
 We use the data taper $h(\xb) = \prod_{j=1}^{d} h_{0.025}(x_j)$, where for $a \in (0,1/2)$,
\begin{equation} \label{eq:taper-ex}
h_a(x) = 
\begin{cases}
(x+0.5)/a - \frac{1}{2\pi} \sin(2\pi (x+0.5) / a ), & -1/2 \leq x \leq (-1/2)+a. \\
1, & (-1/2) + a< x <(1/2) -a. \\
h_a(-x), & (1/2) -a < x \leq 1/2.
\end{cases}
\end{equation}
Then, it is easily seen that $h$ satisfies (\ref{eq:h-summable}), in turn, meeting the condition on $h$ in Theorem \ref{thm:Whittle}. However, in our simulations, selection of $a \in (0,1/2)$ seems not notably impact the performance of our estimator. 

\vspace{0.5em}

\noindent \textit{Choice of $D$}. \hspace{0.1em}
 In practice, we select the prespecified domain $D \subset \R^d$ for the Whittle likelihood in (\ref{eq:Whittle}) as $D = \{ \ob \in \R^2: d_0 \leq \|\ob\|_{\infty} \leq d_1\}$ for some $0 \leq d_0<d_1<\infty$. Inspecting 
Theorem \ref{thm:In-bias} (also Theorem \ref{thm:asymp-equiv} in the Appendix) we exclude the frequencies near the origin (corresponding to frequencies such that $\|\ob\|_{\infty} < d_0$) due to the large bias of the periodogram at frequencies close to the origin. The upper bound $d_1 \in (0,\infty)$ can be chosen such that $|f(\ob) - (2\pi)^{-d} \lambda| \approx 0$ for $\|\ob\|_{\infty} > d_1$, ensuring information outside $D$ has little contribution to the form of spectral density function. In case where no information on the true spectral density function is available, $f$ can be replaced with its kernel smoothed periodogram $\widehat{f}_{n,b}$ in the selection criterion of $d_1$.

\vspace{0.5em}

\noindent \textit{Discretization}. \hspace{0.1em}
Since the Whittle likelihood $L_n(\btheta)$ in (\ref{eq:Whittle}) is defined as an integral, we approximate $L_n(\btheta)$ with its Riemann sum
\begin{equation} \label{eq:discrete-Whittle}
L_n^{(R)}(\btheta) = \sum_{\ob_{\kb,\Omega} \in D}\left( \frac{\widehat{I}_{h,n}(\ob_{\kb,\Omega})}{f_{\btheta}(\ob_{\kb,\Omega})} + \log f_{\btheta} (\ob_{\kb,\Omega}) \right), \quad n \in \N, \quad \btheta \in \Theta \subset \R^p,
\end{equation} where for $\Omega >0 $ and $\kb = (k_1, k_2)^\top \in \Z^{2}$, $\ob_{\kb,\Omega} = (2\pi k_1 / \Omega, 2\pi k_2/\Omega)^\top$.  The feasible criterion of $\widehat{\btheta}_n$ in (\ref{eq:thetahat}) is $\widehat{\btheta}_n^{(R)} = \arg\min_{\btheta \in \Theta} L_n^{(R)}(\btheta)$. An efficient way to compute the periodograms on a grid is discussed in Appendix \ref{eq:per-compute}.

In simulations, we set  $\Omega = A$, where $A >0$ is the side length of the window. As a theoretical justification, \cite{p:sub-18} proved the asymptotic normality of the averaged periodogram under an irregularly spaced spatial data framework. She also showed that setting $\Omega \propto A$ is ''optimal'' in the sense that a finer grid ($\Omega >> A$) does not improve the variance of the averaged periodogram. However, we do not yet have theoretical results for the asymptotics of $\widehat{\btheta}_n^{(R)}$ in the spatial point process framework. These will be investigated in future research.

\subsection{Results under correctly specified models} \label{sec:CS}
In this section, we simulate the spatial point patterns from the TCP model with parameter $(\kappa_0, \alpha_0, \sigma_0^2) = (0.2, 10, 0.5^2)$ and the GDPP model with parameter $(\lambda_0, \rho_0^2) = (1, 0.55^2)$. 
For spatial patterns generated by the TCPs (resp. GDPPs), we fit the parametric TCP models (resp. GDPP models) using three different estimation methods. Following the guideline in Section \ref{sec:practice}, we set the prespecified domain $D_{2\pi} = \{ \ob \in \R^2: \frac{1}{10}\pi \leq \|\ob\|_{\infty} \leq 2\pi \}$ in our estimator for both TCP and GDPP. This choice captures the shape of the spectral densities without adding unnecessary computation.

The bias and standard errors of three different methods are presented in Table \ref{tab:TCP-DGPP}. See also, Figures \ref{fig:TCP-density} and \ref{fig:GDPP-density} in the Appendix for the empirical distributions 

\begin{table}[h]
    \centering
  \begin{tabular}{ccc|ccccc}
\multirow{2}{*}{Model} & \multirow{2}{*}{Window} & \multirow{2}{*}{Parameter} & \multicolumn{3}{c}{Method} \\
\cline{4-6}
& & & Ours & ML & MC \\ \hline \hline
\multirow{12}{*}{TCP}
& \multirow{4}{*}{$[-5,5]^2$} & $\kappa$ &  \textbf{-0.04(0.11)} & -0.07(0.70) & -0.04(0.10) \\  
&				     & $\alpha$ &  \textbf{0.72(3.52)} & -0.62(9.51) & --- \\
&				     & $\sigma^2$ &  \textbf{0.02(0.07)} & -0.06(0.35) & 0.01(0.10) \\ 
&				     & Time(sec) &  0.74 & 0.38 & 0.07 \\ 
\cline{2-6}
& \multirow{4}{*}{$[-10,10]^2$} & $\kappa$ &  -0.02(0.05) & -0.02(0.05) & \textbf{-0.01(0.05)} \\  
&				     & $\alpha$ &  \textbf{0.60(1.77)} & 0.24(3.37) & --- \\
&				     & $\sigma^2$ &  \textbf{0.01(0.04)} & -0.02(0.20) & 0.00(0.06) \\ 
&				     & Time(sec) &  2.38 & 5.67 & 0.23 \\ 
\cline{2-6}
& \multirow{4}{*}{$[-20,20]^2$} & $\kappa$ &  -0.01(0.04) & \textbf{0.00(0.03)} & \textbf{0.00(0.03)} \\  
&				     & $\alpha$ &  \textbf{0.25(1.03)} & 0.15(1.23) & --- \\
&				     & $\sigma^2$ &  \textbf{0.01(0.02)} &0.00(0.03) & 0.00(0.03) \\ 
&				     & Time(sec) &  9.15 & 173.66 & 1.95 \\ \hline

\multirow{9}{*}{GDPP}
& \multirow{3}{*}{$[-5,5]^2$} & $\lambda$ &  0.00(0.10) & \textbf{0.03(0.07)} & --- \\  
&				     & $\rho^2$ &  0.01(0.09) & \textbf{-0.02(0.03)} & 0.05(0.07) \\
&				     & Time(sec) &  0.14 & 1.84 & 0.05 \\ 
\cline{2-6}
& \multirow{3}{*}{$[-10,10]^2$} & $\lambda$ &  0.00(0.06) & \textbf{0.01(0.04)} & --- \\  
&				     & $\rho^2$ &  0.01(0.04) & \textbf{-0.01(0.01)} & 0.02(0.03) \\ 
&				     & Time(sec) &  0.39 & 30.70 & 0.08 \\ 
\cline{2-6}
& \multirow{3}{*}{$[-20,20]^2$} & $\lambda$ &  0.00(0.03) & \textbf{0.01(0.02)} & --- \\  
&				     & $\rho^2$ &  0.00(0.02) & \textbf{0.00(0.01)} & 0.00(0.02) \\
&				     & Time(sec) &  1.62 & 590.02 & 0.67 \\ \hline
\end{tabular} 
\caption{\textit{The bias and the standard errors (in the parentheses) of the estimated parameters based on three different approaches for the TCP and the GDPP. The true parameters are $(\kappa_0, \alpha_0, \sigma_0^2) = (0.2, 10, 0.5^2)$ for TCP and $(\lambda_0, \rho_0^2) = (1, 0.55^2)$ for GDPP. The time is calculated as an averaged computational time (using a parallel computing in R on a desktop computer with an i7-10700 Intel CPU) of each method per one simulation from 500 independent replications. We use \textbf{bold} to denote the smallest RMSE.
}}
\label{tab:TCP-DGPP}
\end{table}

First, we examine the accuracy of the estimators. Our estimator exhibits the smallest root-mean-squared erorr (RMSE) of $\alpha$ and $\sigma^2$ in the TCP model across all windows; has the smallest RMSE for $\kappa$ when $D_n=[-5,5]^2$; and has reliable estimation results for the GDPP model. 
The MC estimator performs well for both the TCP and GDPP models, achieving the smallest RMSE for $\kappa$ in the TCP model for $D_n = [-10,10]^2$ and $[-20,20]^2$. The ML estimator consistently performs the best for the GDPP model. For the TCP model, the ML estimators of $\kappa$ and $\sigma^2$ yield satisfactory finite sample results, but that of $\alpha$ exhibits a relatively large standard error. Overall, the biases and standard errors of all three estimators tend to decrease to zero as the window size increases. 

Moving on, we consider the computation time of each method. Firstly, the MC method has the fastest computation time for both models and all windows. 
Secondly, the ML method exhibits a reasonable computation speed for both models when $D_n = [-5,5]^2$, yielding the expected number of observations equal to 200 (for TCP) or 100 (for GDPP). However, when the number of observations are in the order of a few thousands (corresponding to $D_n = [-20,20]^2$), the ML method incurs the longest computational time. 
Lastly, our method takes less than 10 seconds to compute the parameter estimates for the TCP model and less than 2 seconds for the GDPP model both under $D_n = [-20,20]^2$. Specifically, the computation for our method involves two steps: (a) evaluation of $\{\widehat{I}_{h,n}(\ob_{\kb,A}): \ob_{\kb,A} \in D_{2\pi}\}$ and (b) optimization of $L_n^{(R)}(\btheta)$ in (\ref{eq:discrete-Whittle}). Once step (a) is done, there is no need to update the set of periodograms in the optimization step. In the simulations, it takes, on average, less than 0.5 seconds to evaluate $\{\widehat{I}_{h,n}(\ob_{\kb,40}): \ob_{\kb,40} \in D_{2\pi}\}$ for both TCP and GDPP. This indicates that most of the computational burden for our method stems from the optimization of the Whittle likelihood. By employing a coarse grid in (\ref{eq:discrete-Whittle}), i.e., using $\Omega = A/2$ instead of $\Omega = A$, the computational time can dramatically decrease. 

As a final note, bear in mind that the number of parameters for the MC method is one for TCP and two for GDPP, representing a lower parameter count compared to the corresponding ours or ML estimators. Additionally, the contrast function of the MC method for both TCP and GDPP is specifically designed for isotropic processes. Therefore, for the models we consider in this simulations, the MC method clearly holds a computational advantage over both our method and the ML method.

\subsection{Results under misspecified models} \label{sec:MS}
Now, we consider the case when the models fail to identify the true point patterns. For the data-generating process, we simulate from the LGCP model driven by the latent intensity field $\Lambda(\xb)$, $\xb \in \R^2$, where the first-order intensity is $\lambda^{(true)} = \exp(0.5) \approx 1.65$ and the covariance function is $R(\xb) = \cov(\log \Lambda(\xb), \log \Lambda(\textbf{0})) = 2\exp(-\|\xb\|)$, $\xb \in \R^2$. 
Following the arguments in Section \ref{sec:LGCP}, the true spectral density function $f(\ob)$, $\ob \in \R^2$, is given by
\begin{equation} \label{eq:fLGCP}
f(\ob) = (2\pi)^{-d} \left[ \lambda^{(true)} + (\lambda^{(true)})^{2}
\int_{\R^2} \left( \exp\{ 2 \exp(-\|\xb\|)\} - 1 \right) e^{-i \xb^\top \ob} d\xb \right].
\end{equation} 

In each simulation, we fit the TCP model with parameter $\btheta = (\kappa,\alpha, \sigma^2)^\top$. 
To examine the effect of the selection of the prespecified domain, we use $D_{2\pi} = \{ \ob \in \R^2: \frac{1}{10}\pi \leq \|\ob\|_{\infty} \leq 2\pi \}$ and $D_{5\pi} = \{ \ob \in \R^2: \frac{1}{10}\pi \leq \|\ob\|_{\infty} \leq 5\pi \}$ when evaluating our estimator.

Next, we consider the best fitting TCP model. The best fitting TCP parameter is given by $\btheta_{0}(D,A) = \arg\min_{\btheta \in \Theta} \mathcal{L}^{(R)}(\btheta)$, where
\begin{equation} \label{eq:mathcalL-R}
\mathcal{L}^{(R)}(\btheta) = \sum_{\ob_{\kb,A} \in D}\left( \frac{f(\ob_{\kb,A})}{f_{\btheta}^{(TCP)}(\ob_{\kb,A})} + \log f_{\btheta}^{(TCP)}(\ob_{\kb,A}) \right), \quad \btheta \in \Theta,
\end{equation} is the Riemann sum analogue of $\mathcal{L}(\btheta)$ in (\ref{eq:mathcalL}). 
Figure \ref{fig:mis-spec} illustrates the (log-scale of) the true spectral density ($f$; solid line) and the best fitting TCP spectral density functions 
$f_{\btheta}^{(TCP)}$ for $\btheta = \btheta_0(D_{2\pi}, 10)$ (dashed line) and $\btheta =\btheta_0(D_{5\pi}, 10)$ (dotted line). We note that the best TCP spectral densities evaluated on two different domains ($D_{2\pi}$ and $D_{5\pi}$) have distinct characteristics. In detail, $f_{\btheta}^{(TCP)}$ for $\btheta = \btheta_0(D_{2\pi}, 10)$ captures the peak and the curvature of the true spectral density more accurately but fails to identify the true asymptote (horizontal line with amplitude $(2\pi)^{-2}\lambda^{(true)}$). On the other hand, $f_{\btheta}^{(TCP)}$ for $\btheta = \btheta_0(D_{5\pi}, 10)$ successfully captures the asymptote of the true density but underestimates the power near the origin. 

\begin{figure}[ht!]
\centering
\includegraphics[width=0.5\textwidth]{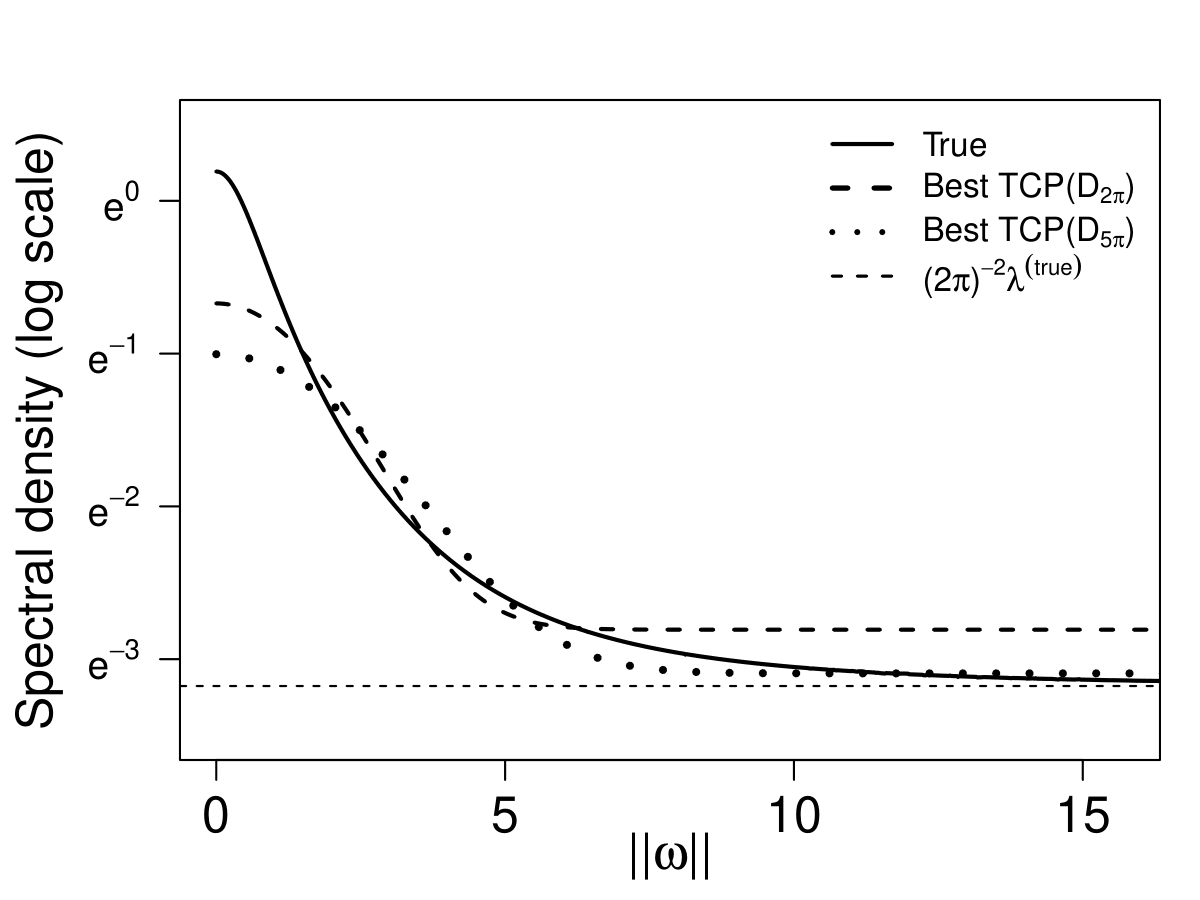}
\caption{\textit{ The true spectral density function ($f(\ob)$; solid line) as in (\ref{eq:fLGCP}) and the two best fitting TCP spectral densities for $A = 10$ evaluated on $ D_{2\pi}$ (dashed line) and on $D_{5\pi}$ (dotted line). All three spectral densities are plotted in log-scale against $\|\ob\| \in [0,\infty)$. The horizontal line indicates the asymptote of $f$ which takes a value $(2\pi)^{-2} \lambda^{(true)}$.
}}
\label{fig:mis-spec}
\end{figure}

Table \ref{tab:fTCP-mis} below summarizes parameter estimation results. The empirical distribution of each estimator can be found in Figure \ref{fig:LGCP-density} in Appendix. 
For the parameter fitting, we also report the estimation of the first-order intensity $\lambda = \kappa \alpha$.

\begin{table}[h]
    \centering
  \begin{tabular}{cc|cc|cccccc}
\multirow{2}{*}{Window} & \multirow{2}{*}{Par.} & \multicolumn{2}{c|}{Best Par.} &
\multicolumn{4}{c}{Method} \\
\cline{3-8}
& & $D_{2\pi}$ & $D_{5\pi}$ & Ours($D_{2\pi}$) & Ours($D_{5\pi}$)  & ML & MC \\ \hline \hline
 \multirow{5}{*}{$[-5,5]^2$}& $\kappa$ &  0.32  & 0.25                           & 0.38(0.23)  & 0.38(0.23)   & 0.43(1.73) & 0.25(0.28) \\  
				     & $\alpha$ &    7.46 &  7.08                          & 7.49(7.75)  & 6.41(5.19)  & 14.96(19.59) & --- \\
				     & $\sigma^2$& 0.17   &  0.10                       & 0.32(0.80)   & 0.16(0.28)  & 0.34(0.34) & 0.24(0.17) \\
				     & $\lambda = \kappa \alpha$ &  2.38 &  1.79  & 2.16(1.34) & 1.76(0.75)  & 1.49(0.53) & --- \\ 
				     & Time(sec) &  ---   & ---                               &   0.61         & 3.43         & 0.21          & 0.07 \\ 
\cline{1-8}
 \multirow{5}{*}{$[-10,10]^2$}& $\kappa$ &  0.31 &  0.24                         & 0.36(0.13)    & 0.30(0.10)   & 0.12(0.07) & 0.13(0.06) \\  
				     & $\alpha$         &  7.74 &  7.37                     & 7.33(3.79) & 6.88(3.69)   & 20.02(18.46) & --- \\
				     & $\sigma^2$    & 0.18  &  0.10                      & 0.20(0.08)   &  0.11(0.05)  & 0.48(0.53) & 0.36(0.25) \\ 
				     & $\lambda = \kappa \alpha$ &  2.43 &  1.80   & 2.36(0.92)  & 1.79(0.43)   & 1.60(0.31) & --- \\ 
				     & Time(sec) &   ---  & ---                               & 1.68            &  11.13       &  2.99          & 0.16 \\ 
\cline{1-8}
 \multirow{5}{*}{$[-20,20]^2$}& $\kappa$ &   0.32 & 0.25                     & 0.34(0.08)  & 0.27(0.06) & 0.09(0.03) & 0.09(0.03) \\  
				     & $\alpha$ &    7.63   & 7.13                        & 7.53(2.21)    & 7.16(2.36) & 20.44(10.52) & --- \\
				     & $\sigma^2$ & 0.18 & 0.10                       &0.19(0.04)   & 0.11(0.03) & 0.41(0.20) & 0.46(0.19) \\ 
				     & $\lambda = \kappa \alpha$ & 2.44 & 1.80  & 2.44(0.59) &  1.81(0.24) & 1.64(0.17) & --- \\ 
				     & Time(sec) &  ---   & ---                              & 5.17          & 40.45        & 71.16      & 1.36  \\ \hline
\cline{1-8}
\end{tabular} 
\caption{\textit{The average and the standard errors (in the parentheses) of the estimated parameters for the misspecified LGCP fitting with the TCP model. The best fitting parameters are calculated by minimizing $\mathcal{L}^{(R)}(\btheta)$ in (\ref{eq:mathcalL-R}).
When evaluating our estimator, we use two different prespecified domains, $D_{2\pi}$ and $D_{5\pi}$. 
}}
\label{tab:fTCP-mis}
\end{table}

We first discuss the best fitting parameters. As already shown in Figure \ref{fig:mis-spec} above, the best fitting parameter results reveal substantial differences between $D_{2\pi}$ and $D_{5\pi}$. Surprisingly, the first-order intensity of the best fitting TCP model on $D_{2\pi}$ is $2.38$ which significantly deviates from the true first-intensity ($\approx 1.65$). The rationale behind this discrepancy is that the true spectral density function $f(\ob)$ in $\ob \in D_{2\pi}$ is still distant from its asymptote value $(2\pi)^{-2} \lambda^{(true)}$. Therefore, the information contained in $\{\widehat{I}_{h,n}(\ob_{}): \ob \in D_{2\pi}\}$ proves insufficient for estimating the true first-order intensity. 
As a remedy for the discrepancy between the fitted first-order intensity and the true first-order intensity, one may fit the "reduced" TCP model. Specifically, we can fit the TCP model with a parameter constraint $\alpha = \widehat{\lambda}_n / \kappa$, where $\widehat{\lambda}_n = N_{X}(D_n) / |D_n|$ is the unbiased estimator of $\lambda^{(\text{true})}$. One advantage of using such a parameter constraint is that the estimated first-order intensity of the reduced TCP model always correctly estimates the true first-order intensity. Please refer to Appendix \ref{sec:TCP-mis2} for details on the construction of the reduced TCP model and its parameter fitting results.

Next, we examine the estimation outcomes from different methods. As the window size increases, our estimators evaluated on $D_{2\pi}$ and $D_{5\pi}$ converge toward the corresponding best fitting parameters. These results substantiate the asymptotic behavior of our estimator in Theorem \ref{thm:Whittle}. For the ML estimator, the standard errors of the estimation of $\kappa$ and $\sigma^2$ decrease as the window size increases, however, the standard error of the estimated $\alpha$ is still large even the window size is large. It is observed that $\sigma^2$ tends to converges to a fixed value (approximately 1.6), but the estimated $\alpha$ value (resp. estimated $\kappa$ value) tends to increase (resp. decrease) as $D_n$ increases. It is intriguing that the ML estimator estimates the true first-order intensity of the process even the model is misspecified. However, to the best of our knowledge, there are no theoretical results available for the ML estimator under model misspecification. Lastly, the standard errors of the MC estimators decreases as the window increases. However, based on the results on Table \ref{tab:fTCP-mis}, it is not clear that the MC estimator under model misspecification converges to a fixed parameter which is non-shrinking or non-diverging. 

Lastly, our estimator based on the prespecified domain $D_{2\pi}$ is reasonably fast even for the largest window $D_n = [-20,20]^2$. However, when the prespecified domain is $D_{5\pi}$, the computation time can take up to 40 seconds per simulation. This is because when computing $L_n^{(R)}(\btheta)$, the number of computational grids for $D_{5\pi}$ is about $(5/2)^2=6.25$ times larger than that of $D_{2\pi}$. To reduce the computation time, one may consider using a coarse grid on $D_{5\pi}$ such as using a coarse grid $\Omega = A/2$.

\section{Concluding remarks and possible extensions}
In this article, we study the frequency domain estimation and inferential methods for spatial point processes. We show that the DFTs for spatial point processes still satisfy the asymptotic joint Gaussianity, which is a classical result for the DFTs applicable to time series or spatial statistics. 
Our approach accommodates irregularly scattered point patterns, thus the fast Fourier transform algorithm is not applicable for evaluating the DFTs on a set of grids. Nevertheless, our simulations indicate that the DFT based model parameter estimation method remains computationally attractive with satisfactory finite sample performance. The advantage of our method becomes more pronounced when fitting a model with misspecification. We prove that our proposed model parameter estimator is asymptotically Gaussian, estimating the ``best'' fitting parameter that minimizes the spectral divergence between the true and conjectured spectra. According to our simulation results, it appears that our method is the only promising approach that exhibits satisfactory large sample properties under model misspecification, distinguishing itself from other two spatial domain methods---the likelihood-based method and least square method.

We anticipate that our frequency domain approaches can be well extended to multivariate point processes. In multivariate case, one need to consider the ``joint'' higher-order intensity and cumulant intensity functions, as introduced in \cite{p:zhu-23}, Section 2. Additionally, the sets of assumptions presented in this paper need appropriate reformulation (see Section 4.1 of the same reference). In Appendix \ref{sec:IRS}, we show that our DFT-based approaches can also be extended to the class of inhomogeneous processes, but a significant portion of the theoretical development on frequency domain methods of the inhomogeneous processes are remained open. This will be a good revenue for the furture research. 

\section*{Acknowledgments}
JY acknowledge the support of the Taiwan’s National Science and Technology Council
(grants 110-2118-M-001-014-MY3 and 113-2118-M-001-012). The authors thank Hsin-Cheng Huang and Suhasini Subba Rao for fruitful comments and suggestions, and Qi-Wen Ding for assistance with simulations.
The authors also wish to thank the two anonymous referees and editors for their valuable comments and corrections, which have greatly improved the article in all aspects.

\bibliography{bib-WhittleSPP}
\bibliographystyle{plainnat}


\pagebreak

\appendix
\counterwithin{figure}{section}
\counterwithin{table}{section}


\section{Proof of Theorem \ref{thm:A1}} \label{sec:proofA1}

\subsection{Equivalence between the feasible and infeasible criteria}
Let
\begin{equation} \label{eq:Atildephi}
\widetilde{A}_{h,n}(\phi) = \int_{D} \phi(\ob) I_{h,n}(\ob) d\ob, \quad n\in \N,
\end{equation} be a feasible criterion of the integrated periodogram $\widehat{A}_{h,n}(\phi)$ as in (\ref{eq:Anphi}). In theorem below, we show that
$|D_n|^{1/2} (\widehat{A}_{h,n}(\phi) - A(\phi))$ and $|D_n|^{1/2} (\widetilde{A}_{h,n}(\phi) - A(\phi))$ are asymptotically equivalent. Therefore, both statistics share the same asymptotic distribution.

\begin{theorem} \label{thm:A2}
Let $X$ be a second-order stationary point process on $\R^d$. 
Suppose that Assumptions \ref{assum:A}, \ref{assum:C} (for $\ell=4$), and \ref{assum:E}(i) hold. 
Then,
\begin{equation*}
|D_n|^{1/2} (\widehat{A}_{h,n}(\phi) - A(\phi)) - |D_n|^{1/2} (\widetilde{A}_{h,n}(\phi) - A(\phi)) \Lcon 0, \quad n\rightarrow \infty,
\end{equation*} where $\Lcon$ denotes convergences in $L_2$.
\end{theorem}
\textit{Proof}. Let $K_{h,n}(\ob) = c_{h,n}(\ob) J_{h,n}(-\ob) + c_{h,n}(-\ob)  J_{h,n}(\ob)$, $\ob\in\R^d$, and let
\begin{equation*}
R_1(\ob) = -(\widehat{\lambda}_{h,n} - \lambda) K_{h,n}(\ob)
~~\text{and}~~
R_2(\ob) = |c_{h,n}(\ob)|^2  (\widehat{\lambda}_{h,n} - \lambda)^2, \quad \ob \in \R^d.
\end{equation*}
Then, we have $\widehat{I}_{h,n}(\ob) - I_{h,n}(\ob) = R_1(\ob) + R_2(\ob)$, $\ob \in \R^d$. Therefore, the difference between the feasible integrated periodogram and its theoretical counterpart can be expressed as
\begin{eqnarray*}
|D_n|^{1/2} (\widehat{A}_{h,n}(\phi) - \widetilde{A}_{h,n}(\phi))
&=& |D_n|^{1/2}\int_{D} \phi(\ob) \left( \widehat{I}_{h,n}(\ob) - I_{h,n}(\ob)\right) d\ob\\
&=& S_1 + S_2,~~ \ob\in\R^d,
\end{eqnarray*} where $S_i = |D_n|^{1/2}\int_{D}\phi(\ob) R_{i}(\ob) d\ob$, $i\in \{1,2\}$.
By using Theorem \ref{thm:S12} below, both $S_1$ and $S_2$ converges to zero in $L_2$ as $n \rightarrow \infty$. Thus, we get the desired results.
\hfill $\Box$

\vspace{0.5em}

Thanks to the above theorem, it is enough to prove Theorem \ref{thm:A2} for $\widetilde{A}_{h,n}(\phi)$ replacing $\widehat{A}_{h,n}(\phi)$ in the statement.

\subsection{Proof of the asymptotic bias} \label{sec:A1-1}

By using Theorem \ref{thm:DFT2} below, an expectation of the (theoretical) periodogram can be expressed as
\begin{equation*}
\Ex [I_{h,n}(\ob)]
= \int_{\R^d} f(\xb) F_{h,n}(\ob-\xb) d\xb, \quad n \in \N, \quad \ob \in \R^d,
\end{equation*} where $F_{h,n}$ is the Fej\'{e}r Kernel defined as in (\ref{eq:fejer}).
Therefore, applying Lemma \ref{lemma:Cn}(b) to the above expression, we have
\begin{equation} \label{eq:expect-I}
\Ex [I_{h,n}(\ob)] - f(\ob) = \int_{\R^d} f(\xb) F_{h,n}(\ob-\xb) d\xb- f(\ob) = O(|D_n|^{-2/d}), \quad n \rightarrow \infty
\end{equation} uniformly in $\ob \in \R^d$.
Therefore, an expectation of $\widetilde{A}_{h,n}(\phi) - A(\phi)$ is bounded by
\begin{equation*}
|\Ex[\widetilde{A}_{h,n}(\phi)] - A(\phi)| \leq \int_{D} |\phi(\ob)| |\Ex[I_{h,n}(\ob)] - f(\ob)| d\ob
\leq C |D_n|^{-2/d} \int_{D} |\phi(\ob)| d\ob = O(|D_n|^{-2/d})
\end{equation*} as $n\rightarrow \infty$. Thus, combining the above with Theorem \ref{thm:A2}, we show (i). \hfill $\Box$

\subsection{Proof of the asymptotic variance} \label{sec:A1-2}
To show the asymptotic variance, we first fix term. First, we view $\phi$ as a function on $\R^d$ by letting $\phi(\ob) =0$ when $\ob \notin D$. 
Since $\phi \in L^1(\R^d)$, let
\begin{equation} \label{eq:phi-FT}
\widehat{\phi}(\blambda) = \mathcal{F}(\phi)(\blambda)
= \int_{\R^d} \phi(\ob) \exp(i\blambda^\top \ob) d\ob, \quad \blambda \in \R^d
\end{equation} be the Fourier transform of $\phi$. Next, for a finite region $B_M = \prod_{i=1}^{d} [-M_i, M_i] \subset \R^d$, let
\begin{equation} \label{eq:phiM-FT}
\phi_{M}(\blambda) = \mathcal{F}^{-1}(\widehat{\phi}(\ob) I_{B_M}(\ob))(\blambda), \quad \blambda \in \R^d,
\end{equation} where $I_{M}(\xb) = 1$ if $\xb \in B_M$ and zero otherwise.
Therefore, the Fourier transform of $\phi_{M}(\blambda)$, denoted $\widehat{\phi}_M(\ob)$, is equal to $\widehat{\phi}(\ob) I_{B_M}(\ob)$ which vanishes outside $B_M$. 
By using similar arguments as in \cite{p:mat-09}, Lemmas 4 and 5, we have $\phi_M \rightarrow \phi$ as $B_M \rightarrow \R^d$ in $L^2$ sense and
for large enough $B_M$, $|D_n| \var ( \widetilde{A}_{h,n}(\phi))$ is closely approximated with $|D_n| \var ( \widetilde{A}_{h,n}(\phi_M))$ uniformly for all $n \in \N$. 

Now, we make an expansion of $|D_n| \var ( \widetilde{A}_{h,n}(\phi_M))$. 
By using that $I_{h,n}(\ob) = |J_{h,n}(\ob)|^2 = J_{h,n}(\ob) J_{h,n}(-\ob)$, for $\ob_1, \ob_2 \in \R^d$, we have
\begin{equation}
\begin{aligned}
\cov ( I_{h,n}(\ob_1), I_{h,n}(\ob_2)) &= \cov ( J_{h,n}(\ob_1) J_{h,n}(-\ob_1), J_{h,n}(\ob_2) J_{h,n}(-\ob_2)) \\
&= \cov (J_{h,n}(\ob_1), J_{h,n}(\ob_2)) \cov (J_{h,n}(-\ob_1), J_{h,n}(-\ob_2)) \\
&+ \cov (J_{h,n}(\ob_1), J_{h,n}(-\ob_2)) \cov (J_{h,n}(-\ob_1), J_{h,n}(\ob_2))\\
& + \cum (J_{h,n}(\ob_1), J_{h,n}(-\ob_1), J_{h,n}(\ob_2), J_{h,n}(-\ob_2)).
\end{aligned}
\label{eq:I-expansion}
\end{equation} Therefore, we have $|D_n| \var ( \widetilde{A}_{h,n}(\phi_M)) = A_1 +A_2 +A_3$, where
\begin{eqnarray*}
A_1 &=& |D_n|\int_{\R^{2d}} \phi_M(\ob_1) \phi_M(\ob_2) \cov (J_{h,n}(\ob_1), J_{h,n}(\ob_2)) \cov (J_{h,n}(-\ob_1), J_{h,n}(-\ob_2)) d\ob_1 d\ob_2, \\
A_2 &=& |D_n|\int_{\R^{2d}} \phi_M(\ob_1) \phi_M(\ob_2) \cov (J_{h,n}(\ob_1), J_{h,n}(-\ob_2)) \cov (J_n(-\ob_1), J_{h,n}(\ob_2)) d\ob_1 d\ob_2, \\
\text{and} \quad
A_3 &=& |D_n|\int_{\R^{2d}} \phi_M(\ob_1) \phi_M(\ob_2) \cum (J_{h,n}(\ob_1), J_{h,n}(-\ob_1), J_{h,n}(\ob_2), J_{h,n}(-\ob_2)) d\ob_1 d\ob_2.
\end{eqnarray*}
By using Theorem \ref{thm:limA1} below, we have
\begin{equation} 
\lim_{n\rightarrow \infty} A_1 
= (2\pi)^{d} (H_{h,4}/H_{h,2}^2) \int_{\R^d} f(\ob)^2 \phi_M(\ob)^2 d\ob.
\label{eq:limA1}
\end{equation} Therefore, for sufficiently large $B_M$, $\lim_{n\rightarrow \infty} A_1$ is arbitrary close to $(2\pi)^{d} (H_{h,4}/H_{h,2}^2) \int_{D} f^2 \phi^2$.

Similarly, the limit of $A_2$ is
\begin{eqnarray} 
\lim_{n\rightarrow \infty} A_2 
&=& (2\pi)^{d} (H_{h,4}/H_{h,2}^2) \int_{\R^d} \phi_M(\ob) \phi_M(-\ob) f(\ob)^2 d\ob \nonumber \\
&\approx&
(2\pi)^{d} (H_{h,4}/H_{h,2}^2) \int_{D} \phi(\ob) \phi(-\ob) f(\ob)^2 d\ob.
\end{eqnarray}

Lastly, by using Theorem \ref{thm:limA3} below, we have 
\begin{equation}
\begin{aligned} 
\lim_{n\rightarrow \infty} A_{3} &= (2\pi)^d (H_{h,4}/H_{h,2}^{2}) \int_{\R^{2d}} \phi_{M}(\blambda_1) \phi_M(\blambda_3) f_4(\blambda_1, -\blambda_1, \blambda_3) d\blambda_1 d\blambda_3 \\
&\approx (2\pi)^d (H_{h,4}/H_{h,2}^{2})\int_{D^2} \phi(\blambda_1) \phi(\blambda_3) f_4 (\blambda_1,- \blambda_1, \blambda_3) d\blambda_1 d\blambda_3.
\end{aligned}
\label{eq:limA3}
\end{equation}
Combining (\ref{eq:limA1})--(\ref{eq:limA3}), we have
\begin{equation*}
\lim_{n\rightarrow \infty} |D_n| \var(\widetilde{A}_{h,n}(\phi)) = (2\pi)^d (H_{h,4}/H_{h,2}^{2}) (\Omega_1 + \Omega_2),
\end{equation*} where $\Omega_1$ and $\Omega_2$ are defined as in (\ref{eq:Omegas}).
Thus, combining the aboves, we show (ii). \hfill $\Box$

\subsection{Proof of the asymptotic normality} \label{sec:A1-3}

Because of Theorem \ref{thm:A2}, it is enough to show the asymptotic normality of the feasible $G_{h,n}(\phi) = |D_n|^{1/2} ( \widetilde{A}_{h,n}(\phi) - A(\phi))$. Let $d\in \{1,2,3\}$ and let
\begin{equation} \label{eq:Gtilde}
\widetilde{G}_{h,n}(\phi) = |D_n|^{1/2} \int_{D} \phi(\ob) \left\{ I_{h,n}(\ob) - \Ex[I_{h,n}(\ob)]\right\} d\ob = 
|D_n|^{1/2} \left( \widetilde{A}_{h,n}(\phi) - \Ex[\widetilde{A}_{h,n}(\phi)]\right).
\end{equation} Then, $G_{h,n}(\phi) - \widetilde{G}_{h,n}(\phi)$ is nonstochastic and is 
bounded by $O(|D_n|^{1/2-(2/d)}) = o(1)$ as $n \rightarrow \infty$ due to Theorem \ref{thm:A1}(i). Therefore, $G_{h,n}(\phi)$ and $\widetilde{G}_{h,n}(\phi)$ are asymptotically equivalent. 

Now, we focus on the asymptotic distribution of $\widetilde{G}_{h,n}(\phi)$. Since $I_{h,n}(\cdot)$ cannot be written as an additive form of the periodograms of the sub-blocks, one cannot directly apply for the standard central limit theorem techniques that are reviewed in \cite{p:bis-19}, Section 1. Instead, we will ``linearize'' the periodogram and show that the associated linear term dominates. 

Without loss of generality, we assume that there exists $C \in (1,\infty)$ such that $C^{-1} n^{d} \leq |D_n| \leq C n^{d}$, $n \in \N$. Therefore, $A_1, \cdots, A_d$ increases proportional to the order of $n$. Next, let $\beta, \gamma \in (0,1)$ be chosen such that $2d / \varepsilon<\beta <\gamma < 1$, where $\varepsilon>2d$ is from Assumption \ref{assum:D}(ii). Let
\begin{equation*}
A_n = \{\kb: \kb \in n^{\gamma}\Z^d~~\text{and}~~ D_n^{(\kb)} = \kb + [-(n^{\gamma}-n^{\beta})/2, (n^{\gamma}-n^{\beta})/2]^d \subset D_n\},
~~n \in \N.
\end{equation*} Therefore, $\bigcup_{\kb \in A_n}D_n^{(\kb)}$ is a disjoint union that is included in $D_n$. Let $J_{h,n}^{(\kb)}(\ob) = \mathcal{J}_{h,n}^{(\kb)}(\ob) - \Ex[\mathcal{J}_{h,n}^{(\kb)}(\ob)]$, where
\begin{equation*}
\mathcal{J}_{h,n}^{(\kb)}(\ob) = 
(2\pi)^{-d/2} H_{h,2}^{-1/2} | D_n^{(\kb)}|^{-1/2} \sum_{\xb \in X \cap D_n^{(\kb)}}
h(\xb / \aB) \exp(-i \xb^\top \ob), \quad \kb \in A_n, \quad \ob \in \R^d.
\end{equation*}
Therefore, $\mathcal{J}_{h,n}^{(\kb)}(\cdot)$ is the DFT evaluated within the sub-block $D_n^{(\kb)}$ of $D_n$. Let
\begin{equation*}
\widetilde{G}_{h,n}^{(\kb)}(\phi) = | D_n^{(\kb)}|^{1/2} \int_{D} \phi(\ob) \left(
|J_{h,n}^{(\kb)}(\ob)|^2 - \Ex[ |J_{h,n}^{(\kb)}(\ob)|^2]
\right) d\ob, \quad n \in \N, \quad \kb \in A_n.
\end{equation*}
Let $k_n = |A_n|$ and let
\begin{equation} \label{eq:Vhn}
V_{h,n}(\phi) = k_n^{-1/2} \sum_{\kb \in A_n} \widetilde{G}_{h,n}^{(\kb)}(\phi), \quad n \in \N.
\end{equation}
In Theorem \ref{thm:Vhn} below, we show that $\widetilde{G}_{h,n}(\phi)$ and $V_{h,n}(\phi)$ are asymptotically equivalent.
 An advantage of using $V_{h,n}(\phi)$ over $\widetilde{G}_{h,n}(\phi)$ is that $V_{h,n}(\phi)$ is written in terms of the sum of $\{\widetilde{G}_{h,n}^{(\kb)}(\phi)\}_{\kb \in A_n}$ which are based on the statistics on the non-overlapping sub-blocks of $D_n$. Therefore, one can show the $\alpha$-mixing CLT for $V_{h,n}(\phi)$ using the standard independent and telescoping sum techniques (cf. \cite{p:gua-07}). Details can be found in Appendix \ref{sec:mixing22} (Theorem \ref{thm:CLT-Vhn}). This together with Theorem \ref{thm:A1}(i) and (ii), we get the desired results.
\hfill $\Box$


\section{Additional proofs of the main results}

\subsection{Proof of Theorem \ref{thm:In-bias}} \label{sec:proof0}
Below we show that the feasible DFT $\widehat{J}_{h,n}(\ob)$ is asymptotically equivalent to its theoretical counterpart $J_{h,n}(\ob)$ as in (\ref{eq:Jn-h}).
\begin{theorem} \label{thm:asymp-equiv}
Let $X$ be a second-order stationary point process on $\R^d$ and let $h$ be the data taper such that $\sup_{\ob \in \R^d} h(\ob) <\infty$.
Suppose that Assumptions \ref{assum:A} and \ref{assum:C}(for $\ell=2$) hold.
Let $\{\ob_{n}\}$ be a sequence on $\R^d$ that is asymptotically distant from $\{\textbf{0}\}$. Then,
\begin{equation*}
\widehat{J}_{h,n}(\ob_n) - J_{h,n}(\ob_n) \Lcon 0, \quad n\rightarrow \infty.
\end{equation*} 
\end{theorem}
\textit{Proof}. By definition, $\widehat{J}_{h,n}(\ob_n) - J_{h,n}(\ob_n) = -(\widehat{\lambda}_{h,n} - \lambda)c_{h,n}(\ob_{n})$. By using Lemma \ref{lemma:cumulants}(b) below, we have $\Ex[|\widehat{\lambda}_{h,n} - \lambda|^2] = \var(\widehat{\lambda}_{h,n}) \leq C|D_n|^{-1}$ for some $C \in (0,\infty)$. Therefore, $\Ex[|\widehat{J}_{h,n}(\ob_{n}) - J_{h,n}(\ob_{n})|^2] \leq C|D_n|^{-1}|c_{h,n}(\ob_{n})|^2$. Next, by using (\ref{eq:Cn-h}) and Lemma \ref{lemma:Hkn-bound} below, we have \\ $\lim_{n\rightarrow \infty}|D_n|^{-1/2}|c_{h,n}(\ob_{n})|=0$. Thus, we get the desired result.
\hfill $\Box$

\vspace{0.5em}

Now we are ready to prove Theorem \ref{thm:In-bias}.
Thanks to the above theorem, it is enough to prove the theorem for $J_{h,n}(\ob_n)$ replacing $\widehat{J}_{h,n}(\ob_n)$ in the statement.
First, we will show (\ref{eq:lim-DFT}). Recall $H_{h,k}^{(n)}$ in (\ref{eq:Hkn}). By using Theorem \ref{thm:cov-exp} below, the leading term of $\cov(J_{h,n}(\ob_{1,n}), J_{h,n}(\ob_{2,n}))$ is $|D_n|^{-1} H_{h,2}^{-1} f(\ob_{1,n}) H_{h,2}^{(n)}(\ob_{1,n}-\ob_{2,n})$. Therefore, by using Lemma \ref{lemma:Hkn-bound} below, we have
\begin{eqnarray*}
\cov(J_{h,n}(\ob_{1,n}), J_{h,n}(\ob_{2,n})) 
&=& |D_n|^{-1} H_{h,2}^{-1} f(\ob_{1,n}) H_{h,2}^{(n)}(\ob_{1,n}-\ob_{2,n})
+o(1)  \\
&\leq& H_{h,2}^{-1} \big( \sup_{\ob \in \R^d} f(\ob) \big) o(1) + o(1), \quad n\rightarrow \infty.
\end{eqnarray*}
Since Assumption \ref{assum:C}(for $\ell=2$) implies $\sup_{\ob \in \R^d}f(\ob) <\infty$, by taking a limit on each side above, we show (\ref{eq:lim-DFT}). 

Next, we will show (\ref{eq:lim-DFT2}). Using Theorem \ref{thm:cov-exp} again together with $H_{h,2}^{(n)}(\textbf{0}) = |D_n| H_{h,2}$, we have
\begin{eqnarray} \label{eq:varJ-o1}
\var(J_{h,n}(\ob_{n})) = 
\cov(J_{h,n}(\ob_{n}), J_{h,n}(\ob_{n})) =
f(\ob_{n}) + o(1), \quad n\rightarrow \infty,
\end{eqnarray} where $o(1)$ error above is uniform in $\ob_{n} \in \R^d$. Since $f$ is continuous, provided Assumption \ref{assum:C} for $\ell=2$, the right hand side above converges to $f(\ob)$ as $n\rightarrow \infty$. 
Thus, we show (\ref{eq:lim-DFT2}). 

Lastly, to show (\ref{eq:lim-DFT3}), by using an expansion (\ref{eq:I-expansion}) together with (\ref{eq:lim-DFT}) and (\ref{eq:lim-DFT2}), we have
\begin{eqnarray*}
\cov(I_{h,n}(\ob_{1,n}), I_{h,n}(\ob_{2,n})) &=& |\cov(J_{h,n}(\ob_{1,n}),J_{h,n}(\ob_{2,n}))|^2
+ |\cov(J_{h,n}(\ob_{1,n}), J_{h,n}(-\ob_{2,n}))|^2 \\
&&~~ +\cum(J_{h,n}(\ob_{1,n}),J_{h,n}(-\ob_{1,n}),J_{h,n}(\ob_{2,n}),J_{h,n}(-\ob_{2,n})) \\
&=& f(\ob)^2 I(\ob_{1,n}=\ob_{2,n})  + o(1) \\
&&~+ \cum(J_{h,n}(\ob_{1,n}),J_{h,n}(-\ob_{1,n}),J_{h,n}(\ob_{2,n}),J_{h,n}(-\ob_{2,n}))
\end{eqnarray*} as $n\rightarrow \infty$.
The last term above is $O(|D_n|^{-1})$ as $n\rightarrow \infty$ due to Lemma \ref{lemma:k4-bound} below. Therefore, we have
$\lim_{n\rightarrow \infty} \cov(I_{h,n}(\ob_{1,n}), I_{h,n}(\ob_{2,n})) = 0$ and
$\lim_{n\rightarrow \infty} \var(I_{h,n}(\ob_{n})) = f(\ob)^2$. This proves (\ref{eq:lim-DFT3}).
All together, we prove the theorem.
\hfill $\Box$

\subsection{Proof of Theorem \ref{thm:asymp-DFT}} \label{sec:proof1}

Let $\Re J_{h,n}(\ob)$ and $\Im J_{h,n}(\ob)$ be the real and imaginary parts of $J_{h,n}(\ob)$, respectively.
Then, by using Theorem \ref{sec:proof0} above, it is enough to show
\begin{equation} \label{eq:joint-normal1}
\left( \frac{\Re J_{h,n}(\ob_{1,n})}{(f(\ob_1)/2)^{1/2}},
\frac{\Im J_{h,n}(\ob_{1,n})}{(f(\ob_1)/2)^{1/2}}, \dots, 
\frac{\Im J_{h,n}(\ob_{r,n})}{(f(\ob_{r})/2)^{1/2}}
 \right)^\top \Dcon \mathcal{SN}_{2r}, \quad n \rightarrow \infty,
\end{equation} where $\mathcal{SN}_{2r}$ is the $2r$-dimensional standard normal random variable. 

To show (\ref{eq:joint-normal1}), we will first show that the asymptotic variance of the left hand side of above is an unit matrix. Note that
\begin{equation*}
\Re J_{h,n}(\ob) = \frac{1}{2} ( J_{h,n}(\ob) + J_{h,n}(-\ob)) \quad \text{and} \quad
\Im J_{h,n}(\ob) = \frac{1}{2i} ( J_{h,n}(\ob) - J_{h,n}(-\ob)), \quad \ob \in \R^d.
\end{equation*} Therefore, for $i,j \in \{1, \dots, r\}$, we have
\begin{eqnarray*}
\cov (\Re J_{h,n}(\ob_{i,n}) , \Re J_{h,n}(\ob_{j,n}))
&=& \frac{1}{4} \big( 
\cov (J_{h,n}(\ob_{i,n}), J_{h,n}(\ob_{j,n})) + \cov (J_{h,n}(\ob_{i,n}), J_{h,n}(-\ob_{j,n})) \\
&& ~~+ \cov (J_{h,n}(-\ob_{i,n}), J_{h,n}(\ob_{j,n})) + \cov (J_{h,n}(-\ob_{i,n}), J_{h,n}(-\ob_{j,n}))
\big).
\end{eqnarray*} Therefore, by using Theorem \ref{thm:In-bias}, one can show
\begin{equation*}
\lim_{n \rightarrow \infty} \cov (\Re J_{h,n}(\ob_{i,n}) , \Re J_{h,n}(\ob_{j,n}))
= \frac{1}{2} f(\ob_i) I(i=j).
\end{equation*} Similarly, for $i,j \in \{1, \dots, r\}$,
\begin{eqnarray*}
\lim_{n \rightarrow \infty} \cov (\Re J_{h,n}(\ob_{i,n}) , \Im J_{h,n}(\ob_{j,n})) &=& 0, \\
\lim_{n \rightarrow \infty} \cov (\Im J_{h,n}(\ob_{i,n}) , \Re J_{h,n}(\ob_{j,n})) &=& 0, \\
\text{and} \quad \lim_{n \rightarrow \infty} \cov (\Im J_{h,n}(\ob_{i,n}) , \Im J_{h,n}(\ob_{j,n})) &=&
\frac{1}{2} f(\ob_i) I(i=j).
\end{eqnarray*}
All together, we show that the limiting variance of the left hand side of (\ref{eq:joint-normal1}) is a unit matrix.

Next, let $\{a_j\}_{j=1}^{r}, \{b_j\}_{j=1}^{r} \in \R$. We define 
$\mathcal{Z}_n = |D_n|^{-1/2} \sum_{\xb \in X \cap D_n} g_n(\xb)$, where 
\begin{equation} 
g_n(\xb) = \sum_{j=1}^{r} a_j h(\xb/\aB) \frac{1}{2}(e^{- \xb^\top \ob_{j,n}} +e^{\xb^\top \ob_{j,n}}) 
+  \sum_{j=1}^{r} b_j h(\xb/\aB)\frac{1}{2i} ( e^{ - \xb^\top \ob_{j,n}} - e^{\xb^\top \ob_{j,n}}).
 \label{eq:gfunction}
\end{equation} Then, it is easily seen that any linear combination of the left hand side of (\ref{eq:joint-normal1}) can be expressed as 
$\mathcal{Z}_n - \Ex[\mathcal{Z}_n]$ for an appropriate $g_n$. Therefore, thanks to Cram\'{e}r-Wald device, to show the asymptotic normality in (\ref{eq:joint-normal1}),  it is enough to show that $\mathcal{Z}_n - \Ex[ \mathcal{Z}_n]$ is asymptotically normal. CLT for $\mathcal{Z}_n - \Ex[ \mathcal{Z}_n]$ under $\alpha$-mixing condition can be easily seen using standard techniques (cf. \cite{p:gua-07}).
One thing that is needed to verify is to show there exists $\delta >0$ such that
\begin{equation} \label{eq:delta-cond}
\sup_{n \in \N}  \Ex \left|\mathcal{Z}_n - \Ex[ \mathcal{Z}_n] \right|^{2+\delta} <\infty.
\end{equation} We will show the above holds for $\delta = 2$, provided Assumption \ref{assum:C}(i) for $\ell=4$.
We note that
\begin{equation} \label{eq:Ex4}
\Ex \left|\mathcal{Z}_n - \Ex[ \mathcal{Z}_n] \right|^4 = \Ex (\mathcal{Z}_n - \Ex[ \mathcal{Z}_n] )^4
= \kappa_4(\mathcal{Z}_n) + 3 \kappa_2(\mathcal{Z}_n)^2,
\end{equation} where $\kappa_2(X) = \cum(X,X)$ and $\kappa_4(X) = \cum(X,X,X,X)$. Therefore, since $g_n(\xb)$ in (\ref{eq:gfunction}) is bounded above uniformly in $n \in \N$ and $\xb \in \R^d$, we can apply Lemma \ref{lemma:k4-bound} and get
\begin{equation*}
\sup_{n \in \N} \Ex \left|\mathcal{Z}_n - \Ex[ \mathcal{Z}_n] \right|^4 \leq
\sup_{n \in \N} |\kappa_4(\mathcal{Z}_n)|  + 3 \sup_{n \in \N}\kappa_2(\mathcal{Z}_n)^2=
O(1).
\end{equation*} Thus, we show (\ref{eq:delta-cond}) for $\delta=2$.
The remaining parts for proving CLT are routine (which may be obtained from the authors upon request). 
\hfill $\Box$


\subsection{Proof of Theorem \ref{thm:KSD}} \label{sec:proofKSD}

Let 
\begin{equation} \label{eq:KSD-t}
\widetilde{f}_{n,b}(\ob)
= \int_{\R^d} W_{b}(\ob - \xb) I_{h,n}(\xb)  d\xb
, \quad n \in \N, \quad \ob \in \R^d,
\end{equation} be the theoretical counterpart of the kernel spectral density estimator. Then, in Corollary \ref{coro:KDE-asym} below, we show that $\widetilde{f}_{n,b}(\ob)$ and $\widehat{f}_{n,b}(\ob)$ are asymptotically equivalent. Therefore, it is enough to show that $\widetilde{f}_{n,b}(\ob)$ consistently estimates the spectral density $f(\ob)$ for all $\ob \in \R^d$. 
By using (\ref{eq:varJ-o1}) together with $\int_{\R^d} W_{b}(\ob - \xb)d\xb = 1$, we have
\begin{equation*}
\left| \Ex[\widetilde{f}_{n,b}(\ob)] 
- \int_{\R^d} W_{b}(\ob-\xb) f(\xb) d\xb
 \right| \leq \int_{\R^d} W_{b}(\ob - \xb) o(1) d\xb = o(1), \quad n \rightarrow \infty.
\end{equation*} Moreover, by using classcial kernel method (cf. \cite{p:dsy-24}, proof of Theorem 5.1), it can be easily seen that $|\int_{\R^d} W_{b}(\ob-\xb) f(\xb) d\xb - f(\ob)| = o(1)$ as $b + |D_n|^{-1}b^{-d}\rightarrow \infty$. Therefore, by using triangular inequality, we have
\begin{equation} \label{eq:fnb-exp}
\lim_{n\rightarrow \infty} \big| \Ex[\widetilde{f}_{n,b}(\ob)] 
- f(\ob)\big| =0.
\end{equation}

Next, by using a similar argument as in Appendix \ref{sec:A1-2} (see also, the proof of \cite{p:dsy-24}, Theorem 5.1), it can be seen that the variance of $\widetilde{f}_{n,b}(\ob)$ is bounded by
\begin{equation} \label{eq:fnb-var}
\var(\widetilde{f}_{n,b}(\ob)) \leq C |D_n|^{-1} \int_{\R^d} W_{b}(\ob - \xb)^2 d\xb = O(|D_n|^{-1} b^{-d}) = o(1), \quad n\rightarrow \infty.
\end{equation}
We mention that unlike the case of Appendix \ref{sec:A1-2}, we do not require Assumption \ref{assum:F} to prove (\ref{eq:fnb-var}). This is because in the expansion of $\var(\widetilde{f}_{n,b}(\ob))$ using (\ref{eq:I-expansion}), the fourth-order cumulant term is bounded by $|D_n|^{-1} \int_{\R^d} W_{b}(\xb)^2 d\xb$ due to Lemma \ref{lemma:k4-bound} below. Therefore, combining (\ref{eq:fnb-exp}) and (\ref{eq:fnb-var}), we have $\widetilde{f}_{n,b}(\ob) \Pcon f(\xb)$ as $n\rightarrow \infty$. This proves the theorem.
\hfill $\Box$


\subsection{Proof of Theorem \ref{thm:KSD-2}} \label{sec:proofKSD-2}
By using Corollary \ref{coro:KDE-asym} below, $\sqrt{|D_n|b^d} (\widetilde{f}_{n,b}(\ob)-f(\ob))$ and 
 $\sqrt{|D_n|b^d} (\widehat{f}_{n,b}(\ob)-f(\ob))$ share the same asymptotic distribution. Therefore, it is enough to show the asymptotic normality of 
$\sqrt{|D_n|b^d} (\widetilde{f}_{n,b}(\ob)-f(\ob))$. Since the scaled kernel function $W_{b}$ has a support on $[-b/2,b/2]^d$, $\widetilde{f}_{n,b}(\ob)$ can be written as an (theoretical) integrated periodogram by setting $\phi_{b}(\xb) = W_{b}(\ob - \xb)$, $\xb \in \R^d$. Therefore, the proof of the asymptotic normality is almost identical to that of the proof of Theorem \ref{thm:A1}. We will only sketch the proof.

\noindent \textit{The bias}. \hspace{0.1em} By using Lemma \ref{lemma:Cn}(b), we have
\begin{equation*}
\bigg| \Ex[\widetilde{f}_{n,b}(\ob)] - \int_{\R^d} W_b(\ob-\xb) f(\xb) d\xb \bigg|  = \int_{\R^d} W_b(\ob-\xb) O(|D_n|^{-2/d}) d\xb = O(|D_n|^{-2/d}), \quad n\rightarrow \infty.
\end{equation*} Moreover, by using classical nonparametric kernel estimation results (cf. \cite{p:dsy-24}, Theorem 5.2), we have
$|\int_{\R^d} W_b(\ob-\xb) f(\xb) d\xb -f(\ob)| = O(b^2)$, $n\rightarrow \infty$.
Therefore, \\ $\lim_{n\rightarrow \infty} \sqrt{|D_n|b^d} \big| \Ex[\widetilde{f}_{n,b}(\ob)] - f(\ob)\big| =0$, provided
$\lim_{n\rightarrow \infty} |D_n|^{1/2} b^{d/2} \{ |D_n|^{-2/d} + b^{2}\} = 0$.

\vspace{0.5em}

\noindent \textit{The variance}. \hspace{0.1em} By using a similar argument to prove Theorem \ref{thm:A1}(ii), one can get
\begin{equation*}
|D_n| b^{d} \var(\widetilde{f}_{n,b}(\ob))
= (2\pi)^d (H_{h,4}/H_{h,2}^2) b^{d} (\Omega_{b,1} + \Omega_{b,2}) + o(1), \quad n\rightarrow \infty,
\end{equation*} where
\begin{eqnarray*}
\Omega_{b,1} &=& \int_{\R^d} W_b(\ob-\xb) \{ W_b(\ob-\xb) + W_b(\ob + \xb) \} f(\xb)^2 d\xb \nonumber \\
\text{and} \quad 
\Omega_{b,2} &=&   \int_{\R^{2d}} W_b(\ob-\xb_1) W_b(\ob-\xb_2) f_4(\xb_1, -\xb_1, \xb_2) d\xb_1 d\xb_2.
\end{eqnarray*} Now, we calculate the limit of $\Omega_{b,i}$ for $i \in \{1,2\}$. We note that $b^{d} W_b(\ob-\xb)^2 = b^{-d}(W^2)(b^{-1}(\ob-\xb))$, $\xb \in \R^d$. Therefore, by treating $W^2$ as a new kernel function in the convolution equation and
using that $\int_{\R^{d}} b^{-d}(W^2)(b^{-1}\xb) d\xb = \int_{\R^d} W(\xb)^2d\xb = W_2$, it is easily seen that
\begin{equation*} 
\lim_{n\rightarrow \infty} b^{d}  \int_{\R^d} W_b(\ob-\xb)^2  f(\xb)^2 d\xb
= W_2 f(\ob)^2.
\end{equation*} Similarly, one can show that
\begin{equation*}
\lim_{n\rightarrow \infty} b^{d}  \int_{\R^d} W_b(\ob-\xb) W_b(\ob+\xb)  f(\xb)^2 d\xb
= \begin{cases}
W_2 f(\ob)^2, & \ob = \textbf{0}.\\
0, & \ob \in \R^d \backslash \{\textbf{0}\}.
\end{cases}
\end{equation*} To calculate the limit of $\Omega_{b,2}$, we note that
\begin{eqnarray*}
&& b^{d} \int_{\R^{2d}} W_b(\ob-\xb_1) W_b(\ob-\xb_2) f_4(\xb_1, -\xb_1, \xb_2) d\xb_1 d\xb_2 \\
&& = b^{d} \int_{\R^d} W_b(\ob-\xb_2) \left( \int_{\R^d} W_b(\ob-\xb_1) f_4(\xb_1, -\xb_1, \xb_2) d\xb_1\right) d\xb_2 \\
&& = b^{d} \int_{\R^d} W_b(\ob-\xb_2) \left( f_4(\ob, -\ob, \xb_2)  + o(1) \right) d\xb_2 = O(b^{d}), \quad n\rightarrow \infty.
\end{eqnarray*} Therefore, $\lim_{n\rightarrow \infty} b^{d} \Omega_{2,b} = 0$. All together, we have
\begin{equation*}
\lim_{n\rightarrow \infty} |D_n| b^{d} \var(\widetilde{f}_{n,b}(\ob)) =
\begin{cases}
2(2\pi)^d (H_{h,4}/H_{h,2}^2)W_2 f(\ob)^2, & \ob = \textbf{0}.\\
(2\pi)^d (H_{h,4}/H_{h,2}^2) W_2 f(\ob)^2, & \ob \in \R^d \backslash \{\textbf{0}\}.
\end{cases}
\end{equation*}

\vspace{0.5em}

\noindent \textit{The asymptotic normality}. \hspace{0.1em} Proof of the asymptotic normality of $\sqrt{|D_n|b^d} (\widetilde{f}_{n,b}(\ob)-f(\ob))$ is almost identical with the proof of the asymptotic normality of $\widetilde{G}_{h,n}(\phi)$ in (\ref{eq:Gtilde}) (we omit the details).
All together, we prove the theorem.
\hfill $\Box$


\subsection{Proof of Theorem \ref{thm:Whittle}} \label{sec:proof2}
First, we will show the consistency of $\widehat{\btheta}_n$. Recall $\mathcal{L}(\btheta)$ in (\ref{eq:mathcalL}). 
We will first show the uniform convergence of $L_n(\cdot)$, that is,
\begin{equation} \label{eq:sup}
\sup_{\btheta \in \Theta} |L_n(\btheta) - \mathcal{L}(\btheta)| \Pcon 0.
\end{equation}

By using Theorem \ref{thm:A1}(i) together with uniform boundedness of $f_{\btheta}^{-1}(\ob)$, we have \\ $\sup_{\btheta \in \Theta} |\Ex[ L_n(\btheta)] - \mathcal{L}(\btheta)| = o(1)$, $n \rightarrow \infty$. Next, since $\var(L_n(\btheta)) = O(|D_n|^{-1})$ as $n \rightarrow \infty$ due to argument in Appendix \ref{sec:A1-2}, we have $L_n(\btheta) - \Ex[L_n(\btheta)] \Pcon 0$ for each $\btheta \in \Theta$. Therefore, to show (\ref{eq:sup}), it is enough to show that $\{L_n(\btheta): \btheta \in \Theta\}$ is stochastic equicontinuous (\cite{p:new-91}, Theorem 2.1).

Let $\delta >0$ and we choose 
$\btheta_1, \btheta_2 \in \Theta$ such that $\|\btheta_1 - \btheta_2\| \leq \delta$. Then, since $f_{\btheta}^{-1}(\ob)$ and $\log f_{\btheta}$ has a first order derivative with respect to $\btheta$ which are continuous on the compact domain $\Theta \times D$, there exist $C_1, C_2 \in (0,\infty)$ such that
\begin{equation*}
|f_{\btheta_1}^{-1}(\ob) - f_{\btheta_2}^{-1}(\ob)| \leq C_1 \delta \quad \text{and}
|\log f_{\btheta_1}(\ob) - \log f_{\btheta_2}(\ob)| \leq C_2 \delta, ~~ \ob \in D.
\end{equation*}
Therefore, for arbitrary $\btheta_1, \btheta_2 \in \Theta$ with $\|\btheta_1 - \btheta_2\| \leq \delta$
\begin{equation*}
|L_n(\btheta_1) - L_n(\btheta_2)| \leq \delta \left( C_1 \int_{D} \widehat{I}_{h,n}(\ob) d\ob + C_2 |D| \right),~~n\rightarrow \infty.
\end{equation*} 
Using Theorem \ref{thm:A1}(i,ii) and
 the term in the bracket above is $O_p(1)$ as $n\rightarrow \infty$ uniformly over $\btheta_1, \btheta_2 \in \Theta$. Therefore, $\{L_n(\btheta): \btheta \in \Theta\}$ is stochastic equicontinuous and, in turn, we show the uniform convergence (\ref{eq:sup}). 

Next, recall $\widehat{\btheta}_n$ and $\btheta_0$ are minimizers of $L_n$ and $\mathcal{L}$, respectively. Then, by using (\ref{eq:sup}),
\begin{eqnarray*}
0 \leq
\mathcal{L}(\widehat{\btheta}_n) - \mathcal{L}(\btheta_0)
&\leq& (\mathcal{L}(\widehat{\btheta}_n) - L_n(\widehat{\btheta}_n))
+ (L_n(\widehat{\btheta}_n) - L_n(\btheta_0))
+ (L_n(\btheta_0) -  \mathcal{L}(\btheta_0)) \\
&\leq& 2\sup_{\btheta\in\Theta} |\mathcal{L}(\btheta) - L_n(\btheta)|.
\end{eqnarray*} Therefore, we have $\mathcal{L}(\widehat{\btheta}_n) - \mathcal{L}(\btheta_0) \Pcon 0$ and since, by assumption, $\btheta_0$ is the unique minimizer of $\mathcal{L}$, we have $\widehat{\btheta}_n \Pcon \btheta_0$. This proves (\ref{eq:theta-consist}).

\vspace{0.5em}

Next, we show the asymptotic normality of $\widehat{\btheta}_n$. By using a Taylor expansion and using that $(\partial L_n/\partial \btheta) (\widehat{\btheta}_n) = 0$,
 there exists $\widetilde{\btheta}_n$,
a convex combination of $\widehat{\btheta}_n$ and $\btheta_0$, such that
\begin{equation*}
|D_n|^{1/2}(\widehat{\btheta}_n - \btheta_0) = 
 \left( \nabla^2 L_n(\widetilde{\btheta}_n) \right)^{-1}
\left( -
|D_n|^{1/2}  \nabla L_n (\btheta_0) \right)
= P_n(\widetilde{\btheta}_n)^{-1} Q_n(\btheta_0).
\end{equation*} 

We first focus on $P_n(\widetilde{\btheta}_n)$. By simple algebra, we have
\begin{eqnarray*}
P_n(\btheta) = \nabla^2 L_n(\btheta) &=&
2 \int_{D} ( \nabla \log f_{\btheta} (\ob)) ( \nabla \log f_{\btheta} (\ob))^\top  \frac{1}{f_{\btheta}(\ob)}
\left(\widehat{I}_{h,n}(\ob) - f_{}(\ob) \right)
 d\ob \\
&&~ 
-\int_{D} \frac{1}{f_{\btheta}(\ob)^2} \nabla^2 f_{\btheta} (\ob) \left(\widehat{I}_{h,n}(\ob)
- f_{}(\ob) \right) d\ob
+2(2\pi)^{d} \Gamma(\btheta)
, ~~ \btheta \in \Theta.
\end{eqnarray*}
By using the (uniform) continuity of $f_{\btheta}^{-1}$, $\nabla \log f_{\btheta}$, $\nabla^2 f^{-1}_{\btheta}$, and  $\nabla^2 f_{\btheta}$, and since $\widetilde{\btheta}_n$ is also a consistent estimator of $\btheta_0$, we have $P_n(\widetilde{\btheta}_n) - P_n(\btheta_0) = o_p(1)$ as $n \rightarrow \infty$. Next, by using Theorems \ref{thm:A1}, the first two terms in $P_n(\btheta_0)$ is $o_p(1)$ as $n \rightarrow \infty$. Therefore, we have
\begin{equation} \label{eq:Pntheta}
P_n(\widetilde{\btheta}_n)  = 2(2\pi)^{d} \Gamma(\btheta_0) + o_p(1), \quad n\rightarrow \infty.
\end{equation}
Next, we focus on $Q_n(\btheta_0)$. By simple algebra, we have
\begin{eqnarray*}
&& Q_n(\btheta_0) 
\\&& 
= -|D_n|^{1/2}
\int_{D} \nabla f_{\btheta_0}^{-1} (\ob)\left(\widehat{I}_{h,n}(\ob)
- f_{\btheta_0}(\ob) \right) d\ob \\
&&= -|D_n|^{1/2}\int_{D}\nabla f_{\btheta_0}^{-1} (\ob) \left( f(\ob)
- f_{\btheta_0}(\ob) \right) d\ob
- |D_n|^{1/2}\int_{D} \nabla f_{\btheta_0}^{-1} (\ob) \left[\widehat{I}_{h,n}(\ob)
- f(\ob) \right] d\ob.
\end{eqnarray*}
The first term above is zero since $\nabla \mathcal{L}(\btheta_0) = 0$ and by using 
Theorem \ref{thm:A1} with a help of Cram\'{e}r-Wald device, the second term is asymptotically centered normal with
variance \\$(2\pi)^{d} (H_{h,4}/H_{h,2}^2) (\Omega_{1}(\btheta_0) + \Omega_2(\btheta_0) )$ where
\begin{eqnarray*}
\Omega_1(\btheta_0) &=& 
2\int_{D} (\nabla f_{\btheta_0}^{-1} (\ob)) 
(\nabla f_{\btheta_0}^{-1} (\ob))^\top
 f_{}(\ob)^2 d\ob 
= 4(2\pi)^d S_1(\btheta_0) \quad \text{and} \\
\Omega_2(\btheta_0) &=& 
\int_{D^2} (\nabla f_{\btheta_0}^{-1} (\ob)) 
(\nabla f_{\btheta_0}^{-1} (\ob))^\top f_{4}(\ob_1, -\ob_1, \ob) d\ob_1 d\ob_2 
= 4(2\pi)^d S_2(\btheta_0).
\end{eqnarray*} Therefore, we conclude,
\begin{equation} \label{eq:Qntheta}
Q_n(\btheta_0) \Dcon \mathcal{N} \left( \textbf{0}, 4(2\pi)^{2d} (H_{h,4}/H_{h,2}^2) ( S_1(\btheta_0) + S_2(\btheta_0) ) \right).
\end{equation}
Combining (\ref{eq:Pntheta}) and (\ref{eq:Qntheta}) and by using continuous mapping theorem, we have
\begin{eqnarray*}
|D_n|^{1/2}(\widehat{\btheta}_n - \btheta_0) &=& P_n(\widetilde{\btheta}_n)^{-1} Q_n(\btheta_0) \\
&\Dcon& \mathcal{N}\left( \textbf{0}, (H_{h,4}/H_{h,2}^2) \Gamma(\btheta_0)^{-1} ( S_1(\btheta_0) + S_2(\btheta_0) )  \Gamma(\btheta_0)^{-1}\right)), \quad n \rightarrow \infty.
\end{eqnarray*} Thus, we show (\ref{eq:theta-CLT}).
All together, we get the desired results. \hfill $\Box$

\section{Representations and approximations of the Fourier transform of the data taper} \label{sec:taper-representation}

Let $D_n$ has a form in (\ref{eq:Dn}). Recall $H_{h,k}^{(n)}(\ob)$ in (\ref{eq:Hkn})
\begin{equation*}
H_{h,k}^{(n)}(\ob) = \int_{D_n} h(\xb/\aB)^k \exp(-i\xb^\top \ob) d\xb, \quad k,n \in \N, \quad \ob \in \R^d.
\end{equation*} 
For two data taper functions $f,g$ with support $[-1/2,1/2]^d$, we define
\begin{equation} \label{eq:Rhgn}
R_{h,g}^{(n)} (\tb, \ob)
= \int_{\R^d} h(\frac{\xb}{\aB}) \left(  g(\frac{\xb+\tb}{\aB}) - g(\frac{\xb}{\aB}) \right) \exp(-i\xb^{\top} \ob) d\xb, \quad 
n \in \N,~~
\tb, \ob \in \R^d.
\end{equation} The term $H_{h,k}^{(n)}(\ob)$ and  $R_{h,g}^{(n)} (\tb, \ob)$ frequently appears throughout the proof of main results. For example, they are related to the expression of $\cov(J_{h,n}(\ob_1), J_{h,n}(\ob_2))$. See Lemma \ref{lemma:cov-exp} below. 

In this section, our focus is to investigate the representations and approximations of $H_{h,k}^{(n)}(\ob)$ and  $R_{h,g}^{(n)} (\tb, \ob)$. We first begin with $R_{h,g}^{(n)} (\tb, \ob)$. 
Since the support of $h(\cdot / \aB)$ is $D_n$, $R_{h,g}^{(n)}$ can be written as $R_{h,g}^{(n)}(\tb,\ob) = R_{h,g,1}^{(n)}(\tb,\ob) +  R_{h,g,2}^{(n)}(\tb,\ob)$, where
\begin{equation}
\begin{aligned}
R_{h,g,1}^{(n)}(\tb,\ob) &= 
\int_{D_n \cap (D_n-\tb)} h(\frac{\xb}{\aB}) \left(  g(\frac{\xb+\tb}{\aB}) - g(\frac{\xb}{\aB}) \right) \exp(-i\xb^{\top} \ob) d\xb  \\
\text{and} \quad 
R_{h,g,2}^{(n)}(\tb,\ob) &= 
- \int_{D_n \backslash (D_n-\tb)} h(\frac{\xb}{\aB}) g(\frac{\xb}{\aB})  \exp(-i\xb^{\top} \ob) d\xb, \quad \tb,\ob \in \R^d.
\end{aligned}
 \label{eq:Rhgn12}
\end{equation}
In the theorem below, we obtain a rough bound for $R_{h,g}^{(n)} (\tb, \ob)$. Let
\begin{equation*}
\rho(x) = \min(x,1), \quad x \in \R.
\end{equation*}
Throughout this section, we let $C \in (0,\infty)$ be a generic constant that varies line by line.

\begin{theorem} \label{thm:Rhgn0}
Let $h$ and $g$ are the data taper on a compact support $[-1/2,1/2]^d$.
Suppose $\lim_{n\rightarrow \infty} |D_n| = \infty$, $\sup_{\ob \in \R^d}h(\ob) <\infty$, and $g$ satisfies Assumption \ref{assum:E}(i). Then,
\begin{eqnarray} 
&& \sup_{n \in \N}\sup_{\tb, \ob \in \R^d} |D_n|^{-1} |R_{h,g}^{(n)} (\tb, \ob)| < \infty  \label{eq:Rhgn000} \\
\text{and} \quad  && 
|D_n|^{-1} \sup_{\ob \in \R^d}|R_{h,g}^{(n)} (\tb, \ob)| = o(1), ~~ \tb \in \R^d, ~~ n \rightarrow \infty.
\label{eq:Rhgn111}
\end{eqnarray}
If we further assume that $g$ is Lipschitz continuous on $[-1/2,1/2]^d$, then the right hand side of (\ref{eq:Rhgn111}) is bounded by $C \rho(\|\tb/\aB\|)$ for some $C \in (0,\infty)$ that does not depend on $\tb \in \R^d$.
\end{theorem}
\textit{Proof}. 
Recall (\ref{eq:Rhgn12}). We bound $R_{h,g,1}^{(n)} (\tb, \ob)$ and $R_{h,g,2}^{(n)} (\tb, \ob)$ separately. First, $R_{h,g,2}^{(n)} (\tb, \ob)$ is bounded by
\begin{equation*}
|R_{h,g,2}^{(n)} (\tb, \ob)| \leq
\sup_{\ob} h(\ob)
\int_{D_n \backslash (D_n-\tb)} \left|g(\xb/\aB)\right| d\xb
\leq C |D_n \backslash (D_n-\tb)|.
\end{equation*} Since $D_n$ has a rectangle shape, it is easily seen that
\begin{equation} \label{eq:Dn-tb}
|D_n \backslash (D_n-\tb)| \leq |D_n| \rho( \sum_{i=1}^{d} |t_i/A_i| )
\leq C |D_n| \rho (\|\tb/\aB\|), \quad \tb \in \R^d.
\end{equation}
Therefore, $|D_n|^{-1} \sup_{\ob \in \R^d} |R_{h,g,2}^{(n)} (\tb, \ob)|$ is bounded by
\begin{equation} \label{eq:gbound11}
|D_n|^{-1}|R_{h,g,2}^{(n)} (\tb, \ob)| \leq C \rho (\|\tb/\aB\|), \quad 
n\in\N,~~ \tb\in\R^d.
\end{equation}
Next, we bound $R_{h,g,1}^{(n)} (\tb, \ob)$. We note that
\begin{equation}
\begin{aligned}
&|D_n|^{-1} \sup_{\ob \in \R^d} |R_{h,g,1}^{(n)} (\tb, \ob)| \\
&~~ \leq |D_n|^{-1} \sup_{\ob} h(\ob)
\int_{D_n \cap (D_n-\tb)} \left|g( (\xb+\tb)/\aB) - g(\xb/\aB)\right| d\xb   \\
&~~\leq 
C \sup_{\xb \in D_n \cap (D_n-\tb)} \left| g((\xb + \tb)/\aB) - g(\xb/\aB) \right|.
\end{aligned}
\label{eq:gbound22}
\end{equation}
Since $g$ is continuous on a compact support, $g$ is bounded and uniformly continuous in $[-1/2,1/2]^d$. Therefore,
\begin{equation*}
\sup_{\tb \in \R^d} \sup_{\xb \in D_n \cap (D_n-\tb)} \left| g((\xb + \tb)/\aB) - g(\xb/\aB) \right| <\infty
\end{equation*}
and for fixed $\tb \in \R^d$,
\begin{equation*}
\lim_{n \rightarrow \infty} \sup_{\xb \in D_n \cap (D_n-\tb)} \left| g((\xb + \tb)/\aB) - g(\xb/\aB) \right|=0.
\end{equation*} Substitute the above two into (\ref{eq:gbound22}), we have $\sup_{n\in\N} \sup_{\tb, \ob \in \R^d} |D_n|^{-1} |R_{h,g,1}^{(n)} (\tb, \ob)| <\infty$ and 
\begin{equation*}
\sup_{\ob \in \R^d} |D_n|^{-1} |R_{h,g,1}^{(n)} (\tb, \ob)| = o(1), \quad \tb \in \R^d, ~~ n \rightarrow \infty.
\end{equation*}
Combining these results with (\ref{eq:gbound11}), we show (\ref{eq:Rhgn000}) and (\ref{eq:Rhgn111}). 

If we further assume that $g$ is Lipschitz continuous on $[-1/2,1/2]^d$, then
\begin{equation*}
\sup_{\xb \in D_n \cap (D_n-\tb)} \left| g((\xb + \tb)/\aB) - g(\xb/\aB) \right| \leq
C \rho(\| \tb/ \aB\|), \quad n \in \N, ~~ \tb \in \R^d.
\end{equation*} Therefore, from (\ref{eq:gbound11}) and (\ref{eq:gbound22}), we have
\begin{equation*}
|D_n|^{-1} \sup_{\ob \in \R^d}|R_{h,g}^{(n)} (\tb, \ob)|
\leq C \rho(\|\tb/\aB\|), \quad n \in \N, ~~ \tb \in \R^d.
\end{equation*} 
Therefore, we prove the assertion. All together, we get the desired results.
\hfill $\Box$

\vspace{0.5em}

To show the limiting variance of the integrated periodogram in (\ref{eq:Anphi}), we require sharper bound for $R_{h,g}^{(n)} (\tb, \ob)$.
Below, we give a bound for $|D_n|^{1/2} |R_{h,g}^{(n)} (\tb, \ob)|$, provided $h$ and $g$ are either constant in $[-1/2,1/2]^d$ or satisfies Assumption \ref{assum:E}(ii) for $m=d+1$.  To do so, 
let $\{h_{\jb}\}_{\jb \in \Z^d}$ be the Fourier coefficients of $h(\cdot)$ that satisfies 
\begin{equation} \label{eq:H-fourier}
h(\xb) = \sum_{\jb \in \Z^d} h_{\jb} \exp(2\pi i \jb^\top \xb), \quad \xb \in [-1/2,1/2]^d.
\end{equation} The Fourier coefficients $\{g_{\jb}\}_{\jb \in \Z^d}$ of $g(\cdot)$ are defined in the same manner. 
For centered rectangle $R  \in \R^d$ (which also includes the degenerate rectangles), we let
\begin{equation} \label{eq:cd}
c_{R}(\ob) = 
\begin{cases}
(2\pi)^{-d/2} |R|^{-1/2} \int_{R} \exp(-i\xb^\top \ob) d\xb, & |R| >0, \\ 
0, & |R| =0,
\end{cases}
\quad \ob \in \R^d.
\end{equation} 

\begin{theorem} \label{thm:Rhgn1}
Let $h$ and $g$ are the data taper on a compact support $[-1/2,1/2]^d$.
Suppose $h$ and $g$ are either constant on $[-1/2,1/2]^d$ or satisfy Assumption \ref{assum:E}(ii) for $m=d+1$.
Then, the following two assertions hold.
\begin{itemize}
\item[(i)]  $R_{h,g}^{(n)} (\tb, \ob)$ satisfies the following identity
\begin{eqnarray*}
R_{h,g}^{(n)} (\tb, \ob)  &=& 
 \sum_{\jb, \kb \in \Z^d} h_{\jb} g_{\kb}  \left( \exp(2\pi i \kb^\top (\tb/\aB) ) - 1 \right)
\int_{D_n \cap (D_n-\tb)} e^{-i\xb^\top (\ob - 2\pi (\jb+\kb)/\aB)} d\xb
 \\
&&\quad - 
 \sum_{\jb, \kb \in \Z^d} h_{\jb} g_{\kb} \int_{D_n \backslash (D_n-\tb)} e^{-i\xb^\top (\ob - 2\pi (\jb+\kb)/\aB)} d\xb, 
\quad n \in \N, ~~ \tb,\ob \in \R^d.
\end{eqnarray*}

\item[(ii)] Let $m_d = 2^{d}-1$. Then, there exist $C \in (0,\infty)$ and $m_d$ number of sequences of centered rectangles (which may include degenerate rectangles)
$\{ D_{n,i}(\tb) \}_{i=0}^{m_d}$ where $D_{n,i}(\tb)$ depends only on $D_n$ and $\tb \in \R^d$ such that for $n \in \N$ and $\tb,\ob \in \R^d$,
\begin{eqnarray*}
&& |D_n|^{-1/2} |R_{h,g}^{(n)} (\tb, \ob)| \\
&& \quad \leq C \sum_{i=0}^{m_d}
\sum_{\jb,\kb \in \Z^d} |h_{\jb}| |g_{\kb}| \rho( \{ \| \kb\|+1\} \|\tb/\aB \|)^{1/2} |c_{D_{n,i}(\tb)}
(\ob - 2\pi(\jb+\kb)/\aB)|.
\end{eqnarray*}
\end{itemize}
\end{theorem}
\textit{Proof}. 
First, we will show (i). We will assume that $h$ and $g$ both satisfies Assumption \ref{assum:E}(ii) for $m=d+1$. The case when either $h$ or $g$ is constant on $[-1/2,1/2]^d$ is straightforward since the corresponding Fourier coefficients for $h(\xb) \equiv c$ in $[-1/2,1/2]^d$ are 
$h_{\jb} = c$ if $\jb = \textbf{0}$ and zero otherwise. By using \cite{b:fol-99}, Theorem 8.22(e) (see also an argument on page 257 of the same reference), we have $|h_{\jb}|, |g_{\jb}|   \leq (1+\|\jb\|)^{-d-1}$ and thus both $\sum_{\jb \in \Z^d} |h_{\jb}|$ and $\sum_{\jb \in \Z^d} |g_{\jb}|$ are finite. For $\tb, \ob \in \R^d$,
\begin{equation}
\begin{aligned}
&R_{h,g,1}^{(n)} (\tb, \ob) \\
&~~ = 
\int_{D_n \cap (D_n-\tb)} h(\frac{\xb}{\aB}) \left(  g(\frac{\xb+\tb}{\aB}) - g(\frac{\xb}{\aB}) \right) e^{-i\xb^{\top}\ob} d\xb \\
&~~ =
\int_{D_n \cap (D_n-\tb)} \sum_{\jb, \kb \in \Z^d} h_{\jb} g_{\kb}  \exp(2\pi i \jb^\top (\xb/\aB) ) \exp(2 \pi i \kb^\top (\xb/\aB)) \\
&\qquad \times
\left( \exp(2\pi i \kb^\top (\tb/\aB) ) - 1 \right) \exp(-i \xb^\top \ob) d\xb \\
&~~ =
 \sum_{\jb, \kb \in \Z^d} h_{\jb} g_{\kb}  \left( \exp(2\pi i \kb^\top (\tb/\aB) ) - 1 \right)
\int_{D_n \cap (D_n-\tb)} e^{-i \xb^\top ( - 2\pi (\jb+\kb)/\aB + \ob )} d\xb.
\end{aligned} 
\label{eq:Rhg1-rep}
\end{equation}
Here, we use Fubini's theorem in the second identity which is due to $\sum_{\jb \in \Z^d} |h_{\jb}| <\infty$ and $\sum_{\jb \in \Z^d} |g_{\jb}| <\infty$. Similarly, for $\tb, \ob\in\R^d$,
\begin{equation}
R_{h,g,2}^{(n)} (\tb, \ob) = 
-  \sum_{\jb, \kb \in \Z^d} h_{\jb} g_{\kb} \int_{D_n \backslash (D_n-\tb)} \exp(-i \xb^\top ( - 2\pi (\jb+\kb)/\aB + \ob )) d\xb.
\label{eq:Rhg2-rep}
\end{equation} Combining the above two expressions, we show (i).

\vspace{0.5em}

Next, we will show (ii). We first focus on $R_{h,g,1}^{(n)} (\tb, \ob)$. From (\ref{eq:Rhg1-rep}), we have
\begin{eqnarray*}
|R_{h,g,1}^{(n)} (\tb, \ob)| &\leq&
\sum_{\jb, \kb \in \Z^d} |h_{\jb}| |g_{\kb}| |\exp(2\pi i \kb^\top (\tb/\aB) ) - 1|
\left| \int_{D_n \cap (D_n-\tb)} e^{-i \xb^\top ( - 2\pi (\jb+\kb)/\aB + \ob )} d\xb\right| \\
&\leq&
C \sum_{\jb, \kb \in \Z^d} |h_{\jb}| |g_{\kb}| \rho(\|\kb\| \|\tb/\aB\|)^{1/2} \left| \int_{D_n \cap (D_n-\tb)} e^{-i \xb^\top ( - 2\pi (\jb+\kb)/\aB + \ob )} d\xb\right|.
\end{eqnarray*} Here, we use 
\begin{equation*}
|e^{i \xb^\top \yb} - e^{0}| \leq 2 \rho(|\xb^\top \yb|) \leq 2 \rho(\|\xb\| \|\yb\|)
\leq 2 \rho(\|\xb\| \|\yb\|)^{1/2}, \quad \xb, \yb \in \R^d
\end{equation*}
 in the second inequality.
We note that $D_n \cap (D_n-\tb)$ is also a rectangle, and $|c_{R+\xb} (\ob)| = |c_{R}(\ob)|$ for all $\xb, \ob \in \R^d$. Therefore, 
for the centered version of the rectangle $D_n \cap (D_n-\tb)$, denoted $D_{n,0}(\tb)$, we have
\begin{eqnarray*}
\left| \int_{D_n \cap (D_n-\tb)} e^{-i \xb^\top ( - 2\pi (\jb+\kb)/\aB + \ob )} d\xb\right|
&=& (2\pi)^{d/2}|D_n \cap (D_n-\tb)|^{1/2} |c_{D_{n,0}(\tb)} (\ob - 2\pi (\jb+\kb)/\aB + \ob)| \\
&\leq& C |D_n|^{1/2}|c_{D_{n,0}(\tb)} (\ob - 2\pi (\jb+\kb)/\aB + \ob)|.
\end{eqnarray*} Substitute this into the upper bound of $|R_{h,g,1}^{(n)} (\tb, \ob)|$, we have
\begin{equation} 
\begin{aligned}
& |D_n|^{-1/2} |R_{h,g,1}^{(n)} (\tb, \ob)|  \\
&~~\leq 
C \sum_{\jb, \kb \in \Z^d} |h_{\jb}| |g_{\kb}| \rho( \{\|\kb\|+1\} \|\tb/\aB\|)^{1/2}
|c_{D_{n,0}(\tb)} (\ob - 2\pi (\jb+\kb)/\aB + \ob)|.
\end{aligned}
\label{eq:R1-upper}
\end{equation}

Secondly, we focus in $R_{h,g,2}^{(n)} (\tb, \ob)$. We first note that $D_n \backslash (D_n-\tb)$ can be written as a disjoint union of finite number of rectangles, where the number of the rectangles are at most $m_d = 2^{d}-1$. See the example in the Figure \ref{fig:1} below for $d=2$. 
\begin{figure}[h]
\centering
\includegraphics[width=0.25\textwidth]{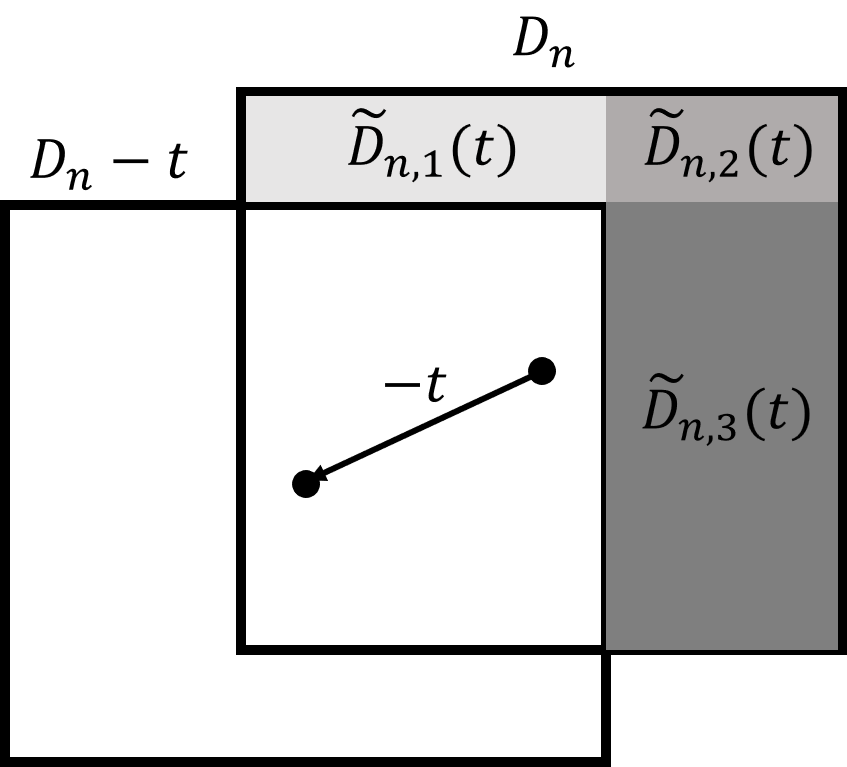}
\caption{Example of disjoint partitions of $D_n \backslash (D_n-\tb)$ when $d=2$.}
\label{fig:1}
\end{figure}

Let $D_n \backslash (D_n-\tb) = \cup_{i=1}^{m_d} \widetilde{D}_{n,i}(\tb)$ be the disjoint union of rectangles and $D_{n,i}(\tb)$ is the centered $\widetilde{D}_{n,i}(\tb)$. Then, by using (\ref{eq:Rhg2-rep}), we have
\begin{eqnarray*}
|R_{h,g,2}^{(n)} (\tb, \ob)| &\leq& (2\pi)^{d/2} \sum_{\jb, \kb \in \Z^d} |h_{\jb}| |g_{\kb}| \sum_{i=1}^{m_d} |D_{n,i}(\tb)|^{1/2} |c_{D_{n,i}(\tb)}(\ob - 2\pi (\jb+\kb)/\aB + \ob)| \\
&\leq& (2\pi)^{d/2} \sum_{i=1}^{m_d} \sum_{\jb, \kb \in \Z^d} |h_{\jb}| |g_{\kb}| |D_n \backslash (D_n-\tb)|^{1/2} |c_{D_{n,i}(\tb)}(\ob - 2\pi (\jb+\kb)/\aB + \ob)| \\
&\leq& C  |D_n|^{1/2} \sum_{i=1}^{m_d} \sum_{\jb, \kb \in \Z^d} |h_{\jb}| |g_{\kb}| \rho(\|\tb/\aB\|)^{1/2} |c_{D_{n,i}(\tb)}(\ob - 2\pi (\jb+\kb)/\aB + \ob)|.
\end{eqnarray*} 
Here, we use (\ref{eq:Dn-tb}) in the last inequality above.
Therefore, we have
 \begin{equation}
\begin{aligned}
& |D_n|^{-1/2} |R_{h,g,2}^{(n)} (\tb, \ob)| \\
&~\leq 
C \sum_{i=1}^{m_d} \sum_{\jb, \kb \in \Z^d} |h_{\jb}| |g_{\kb}| \rho \left( (\|\kb\|+1) \|\tb/\aB\|\right)^{1/2}
|c_{D_{n,i}(\tb)} (\ob - 2\pi (\jb+\kb)/\aB + \ob)|.
\end{aligned}
 \label{eq:R2-upper}
\end{equation}
Combining (\ref{eq:R1-upper}) and (\ref{eq:R2-upper}), we get the desired result.
All together, we prove the theorem.
\hfill $\Box$

\vspace{1em}

Next, we focus on $H_{h,k}^{(n)}(\cdot)$ in (\ref{eq:Hkn}). The following lemma provides an approximation of $H_{h,k}^{(n)}(\cdot)$.

\begin{lemma} \label{lemma:Hkn}
Let $h(\cdot)$ be a data taper that satisfies Assumption \ref{assum:E}(i). Then, for $k \in \{2, 3, \dots, \}$, 
\begin{equation*}
\sup_{n \in \N }\sup_{\ob, \ubb_1, \dots, \ubb_{k-1} \in \R^d}|D_n|^{-1} \bigg| H_{h,k}^{(n)}(\ob) - \int_{\R^{d}} \left(  h(\xb/\aB) \prod_{j=1}^{k-1}
h((\xb+\ubb_j)/\aB) \exp(-i\xb^\top \ob) \right) d\xb \bigg| <\infty.
\end{equation*} Suppose $\lim_{n \rightarrow \infty }|D_n| = \infty$. Then, for fixed $\ubb_1, \dots, \ubb_{k-1} \in \R^d$, as $n \rightarrow \infty$,
\begin{equation*}
|D_n|^{-1} \sup_{\ob \in \R^d} \bigg| H_{h,k}^{(n)}(\ob) - \int_{\R^{d}} \left(  h(\xb/\aB) \prod_{j=1}^{k-1}
h((\xb+\ubb_j)/\aB) \exp(-i\xb^\top \ob) \right) d\xb \bigg| = o(1).
\end{equation*} 
If further assume that $h$ is Lipschitz continuous on $[-1/2,1/2]^d$, then the left term above is bounded by $C \sum_{j=1}^{k-1} \rho(\|\ubb_j / \aB\|)$ uniformly in $\ob \in \R^d$.
\end{lemma}
\textit{Proof}. We will show the lemma for $k=2$. The case for $k\geq 3$ is treated similarly (cf. \cite{b:bri-81}, page 402). 
For $k=2$, the left hand side above is equal to $|D_n|^{-1} |R_{h,h}^{(n)}(\ubb_1, \ob)|$. Thus, by Theorem \ref{thm:Rhgn0}, the above is true for $k=2$.
\hfill $\Box$

\vspace{1em}

Next, we bound $H_{h,k}^{(n)}(\ob)$ for $\ob \neq \textbf{0}$. The lemma below together with Theorem \ref{thm:Rhgn0} are used to prove the asymptotic orthogonality of the DFT in Theorem \ref{thm:In-bias}.

\begin{lemma} \label{lemma:Hkn-bound}
Let $\{D_n\}$ satisfies Assumption \ref{assum:A}.
Let $h$ be a data taper such that $\sup_{\ob \in \R^d} h(\ob) <\infty$. 
Let $\{\ob_{n}\}$ be a sequence on $\R^d$ that is asymptotically distant from $\{\textbf{0}\}$. Then,
\begin{equation} \label{eq:Hkn-bound}
|D_n|^{-1} |H_{h,k}^{(n)}(\ob_n)|  = o(1), \quad n \rightarrow \infty. 
\end{equation}
\end{lemma}
\textit{Proof}. 
Since $h$ has a compact support and 
 $\sup_{\ob \in \R^d} h(\ob) <\infty$, $h \in L^1(\R^d)$. Therefore,
\begin{eqnarray*}
|D_n|^{-1}H_{h,k}^{(n)}(\ob_n) &=& |D_n|^{-1} \int_{D_n} h(\xb/\aB)\exp(-i \xb^\top \ob_n) d\xb
= \int_{[-1/2,1/2]^d} h(\yb)\exp(-i \yb^\top (\aB \cdot \ob_n)) d\yb \\
&=& (2\pi)^d \mathcal{F}^{-1}(h) (\aB \cdot \ob_n).
\end{eqnarray*} Here, we use change of variables $\yb = \xb/\aB$ in the second identity.  Note that
$\|\aB \cdot \ob_n\|_\infty \geq C |D_n|^{1/d} \|\ob_n\|_\infty \rightarrow \infty$ as $n \rightarrow \infty$ due to Assumption \ref{assum:A}. Therefore, by Riemann-Lebesgue lemma, we have $\lim_{n \rightarrow \infty} |\mathcal{F}^{-1}(h) (\aB \cdot \ob_n)| =0$. Thus, we get the desired result.
\hfill $\Box$

\vspace{1em}

Finally, we generalize the Fej\'{e}r Kernel that are associated with data taper $h$ (cf. \cite{p:mat-09}, Equation (21)). For $h(\cdot)$ with support on $[-1/2,1/2]^d$, let
\begin{equation} \label{eq:fejer}
F_{h,n}(\ob) 
=(2\pi)^{-d} H_{h,2}^{-1} |D_n|^{-1} |H_{h,1}^{(n)}(\ob)|^2
= |c_{h,n}(\ob)|^2, \quad n \in \N,~~ \ob \in \R^d.
\end{equation} where $c_{h,n}(\ob)$ is defined as in (\ref{eq:Cn-h}) and $H_{h,2} = \int_{[-1/2,1/2]^d} h(\xb)^2 d\xb$.

\begin{lemma} \label{lemma:Cn}
Let $c_{h,n}(\ob)$ and 
$F_{h,n}(\ob)$ be defined as in (\ref{eq:Cn-h}) and (\ref{eq:fejer}), respectively. Then, the follow assertions hold:
\begin{itemize}
\item[(a)] $\int_{\R^d} F_{h,n}(\ob) d\ob = 1$.
\item[(b)] Suppose that Assumptions \ref{assum:A} and \ref{assum:E}(ii)(for $m=1$) hold. Then, for $\phi(\ob)$ with bounded second derivatives, we have
\begin{equation*}
\int_{\R^d} \phi(\ob) F_{h,n}(\ob) d\ob = \phi(\textbf{0}) + O(|D_n|^{-2/d}), \quad n\rightarrow \infty.
\end{equation*}
\item[(c)] Suppose that Assumptions \ref{assum:A} and \ref{assum:E}(ii) (for $m=d+1$) hold. Then, for a bounded function $\phi$,
\begin{equation*}
\left| \int_{\R^d} \phi(\ob) c_{h,n}(\ob) d\ob \right| = O(|D_n|^{-1/2}), \quad n\rightarrow \infty.
\end{equation*}
\item[(d)] For a bounded function $\phi(\ob)$, $\ob \in \R^d$, there exists $C \in (0,\infty)$ which depends only on $\phi(\cdot)$ such that 
\begin{equation*}
\int_{\R^d} |\phi(\ob) c_{h_1,n}(\ob) c_{h_2,n}(\ob+\ubb)| d\ob \leq C, \qquad n\in\N.
\end{equation*}
\item[(e)] For a bounded and compactly supported function $\phi(\ob)$, there exists $C \in (0,\infty)$ which depends only on $\phi(\cdot)$ such that 
\begin{equation*}
\int_{\R^d} |\phi(\ob) c_{h,n}(\ob)| d\ob \leq C, \quad n\in\N.
\end{equation*}
\end{itemize}

\end{lemma}
\textit{Proof}. (a)--(d) are due to \cite{p:mat-09}, Lemmas 1 and 2. The upper bound of (d) is $\sup_{\ob \in \R^d} |\phi(\ob)|$ which depends only on $\phi(\cdot)$.
To show (e), let $D \subset \R^d$ be the compact support of $\phi$. Then, by using Cauchy-Schwarz inequality and point (a),
\begin{equation*}
\int_{\R^d} |\phi(\ob) c_{h,n}(\ob)| d\ob \leq \left( \int_{D} \phi(\ob)^2 d\ob \right)^{1/2}, \quad n \in \N.
\end{equation*} Thus, we prove (e). All together, we get the desired results.
\hfill $\Box$

\section{Bounds for the cumulants of the DFT} \label{sec:cum-bounds}

In this section, we study the expressions and bounds of terms that are written in terms of the product of cumulants of the DFT. 
This section is essential to prove the asympotic orthogonality of the DFTs in Theorem \ref{thm:In-bias} and includes the limit of the variance of the integrated periodogram in Section \ref{sec:spec-mean}. Throughout the section, we let $C \in (0,\infty)$ be a generic constant that varies line by line. 

\subsection{Expressions of the covariance of the DFTs} \label{sec:cum2}

To begin with, we obtain the expressions for $\cov(J_{h,n}(\ob_1), J_{h,n}(\ob_2))$ in terms of the second-order measures. Note that these expressions above were also verified in \cite{p:raj-23}, Proposition IV.1 and Equation (47).

\begin{theorem} \label{thm:DFT2}
Let $X$ be a second-order stationary point process on $\R^d$ and let $h$ be the taper function such that $\sup_{\xb \in \R^d} h(\xb) <\infty$. Suppose that Assumption \ref{assum:C} holds for $\ell=2$.
Then, for $\ob_1, \ob_2 \in \R^d$, we have\begin{eqnarray}
&& \cov(J_{h,n}(\ob_1), J_{h,n}(\ob_2)) \nonumber \\ 
&&\quad= 
(2\pi)^{-d} H_{h,2}^{-1} |D_n|^{-1}  \lambda \int_{D_n} h(\xb/\aB)^2 e^{-i\xb^\top (\ob_1-\ob_2)} d\xb \nonumber \\
&&\quad~~ +
(2\pi)^{-d} H_{h,2}^{-1} |D_n|^{-1} 
\int_{D_n^2} h(\xb/\aB) h(\yb/\aB) e^{-i(\xb^\top \ob_1 - \yb^\top \ob_2)} \gamma_{2,\text{red}}(\xb-\yb) d\xb d\yb
\label{eq:DFT2-1}\\
&&\quad = (2\pi)^{-d} H_{h,2}^{-1} |D_n|^{-1} \int_{D_n^2} 
h(\xb/\aB) h(\yb/\aB)
e^{-i(\xb^\top \ob_1 - \yb^\top \ob_2)} C(\xb-\yb)d\xb d\yb.
\label{eq:DFT2-2}
\end{eqnarray}
 Suppose that Assumptions \ref{assum:B}(i) and \ref{assum:C}(for $\ell=2$) hold.
Then, for $\ob \in \R^d$, we have
\begin{equation}
 \var(J_{h,n}(\ob)) 
= \int_{\R^d} f(\xb) |c_{h,n}(\ob-\xb)|^2 d\xb.
\label{eq:DFT2-3}
\end{equation}
\end{theorem}
\textit{Proof}. We will only show (\ref{eq:DFT2-1}) and (\ref{eq:DFT2-2}) for the non-tapered case. A general case can be treated similarly. Since $J_n(\ob) = \mathcal{J}_n(\ob) - \lambda c_n(\ob)$ is centered DFT and $\lambda c_n(\ob)$ is deterministic, we have
\begin{equation*}
\cov(J_n(\ob_1), J_n(\ob_2))
 = \Ex[\mathcal{J}_n(\ob_1) \mathcal{J}_n(-\ob_2)] - \lambda^2 c_n(\ob_1) c_n(-\ob_2).
\end{equation*} By using (\ref{eq:mathcalDFT}),
\begin{equation*}
\mathcal{J}_n(\ob_1) \mathcal{J}_n(-\ob_2)
= (2\pi)^{-d} |D_n|^{-1} \bigg(  \sum_{\xb \in X \cap D_n} \exp(-i\xb^\top (\ob_1-\ob_2)) 
 + \sum_{\xb \neq \yb \in X \cap D_n}    
\exp(-i(\xb^\top \ob_1- \yb^\top \ob_2))
\bigg).
\end{equation*} Thus, by using (\ref{eq:lambda}) for both terms above, we have
\begin{eqnarray*}
\Ex[\mathcal{J}_n(\ob_1) \mathcal{J}_n(-\ob_2)]
&=& (2\pi)^{-d} |D_n|^{-1} \bigg\{  \lambda \int_{D_n} e^{-i\xb^\top (\ob_1 -\ob_2)} d\xb \\
&& + 
\int_{D_n}\int_{D_n} \exp(-i(\xb^\top \ob_1- \yb^\top \ob_2)) \lambda_{2,\text{red}}(\xb-\yb) d\xb d\yb \bigg\}.
\end{eqnarray*}
Next, by using that $\int_{D_n^2} \exp(-i(\xb^\top \ob_1- \yb^\top \ob_2)) d\xb d\yb = (2\pi)^{d} |D_n| c_n(\ob_1) c_n(-\ob_2)$ and $\gamma_{2,\text{red}}(\xb) = \lambda_{2,\text{red}}(\xb) - \lambda^2$,
 we have
\begin{eqnarray*}
&&\Ex[\mathcal{J}_n(\ob_1) \mathcal{J}_n(-\ob_2)] - \lambda^2 c_n(\ob_1) c_n(-\ob_2) = 
(2\pi)^{-d} |D_n|^{-1}  \lambda \int_{D_n} e^{-i\xb^\top (\ob_1 -\ob_2)} d\xb \\
&&~~ +
(2\pi)^{-d} |D_n|^{-1} 
\int_{D_n}\int_{D_n}
\exp(-i(\xb^\top \ob_1- \yb^\top \ob_2)) \gamma_{2,\text{red}} (\xb-\yb)d\xb d\yb. 
\end{eqnarray*}
Therefore, we show (\ref{eq:DFT2-1}) for $h(\xb) \equiv 1$ on $[-1/2,1/2]^d$.

(\ref{eq:DFT2-2}) can be easily seen by using the above identity and $C(\xb) = \lambda \delta(\xb) + \gamma_{2,\text{red}}(\xb)$.

Lastly, to show (\ref{eq:DFT2-3}), by substituting (\ref{eq:cov-spectral}) into (\ref{eq:DFT2-1}) for $\ob_1 =\ob_2=\ob$, we have
\begin{eqnarray*}
&& \var(J_{h,n}(\ob)) \\ 
&&~~= 
(2\pi)^{-d} H_{h,2}^{-1}|D_n|^{-1}  \lambda \int_{D_n}h(\xb/\aB)^2 d\xb 
+(2\pi)^{-d} H_{h,2}^{-1} |D_n|^{-1} 
\int_{D_n^2} h(\xb/\aB) h(\yb/\aB)
\\
&& ~~ \quad \times
 e^{-i(\xb^\top \ob - \yb^\top \ob)} \int_{\R^d} e^{i (\xb-\yb)^\top \tb}(f(\tb) - (2\pi)^{-d} \lambda ) d\tb d\xb d\yb.
\end{eqnarray*} Since $f(\ob) - (2\pi)^{-d} \lambda \in L^1(\R^d)$ due to Assumption \ref{assum:B}(i) and $\sup_{\xb \in \R^d} h(\xb) <\infty$, we can apply Fubini's theorem to interchange the summation above and get
\begin{eqnarray*}
&& \var(J_{h,n}(\ob)) \\
&&\quad = 
(2\pi)^{-d}H_{h,2}^{-1} |D_n|^{-1}  \lambda \int_{D_n}h(\xb/\aB)^2 d\xb  +
(2\pi)^{-d}H_{h,2}^{-1} |D_n|^{-1} \int_{\R^d} d\tb ( f(\tb) - (2\pi)^{-d} \lambda ) \\
&&\quad
~~~\times \int_{D_n} h(\xb/\aB)e^{-i\xb^\top (\ob-\tb)} d\xb  \int_{D_n} h(\yb/\aB) e^{ -i\yb^\top (\tb-\ob)}  d\yb \\
&&\quad = 
(2\pi)^{-d}  \lambda  + \int_{\R^d}  \{ f(\tb) - (2\pi)^{-d} \lambda \} |c_{h,n}(\ob-\tb)|^2 d\tb 
= \int_{\R^d}  f(\tb) |c_{h,n}(\ob-\tb)|^2 d\tb.
\end{eqnarray*} 
Here, we use (\ref{eq:Cn-h}) and $\int_{D_n} h(\xb/\aB)^2d\xb = H_{h,2} |D_n|$
on the second identity and Lemma \ref{lemma:Cn}(a) on the last identity. Thus, we get the desired results. 

All together, we prove the theorem. \hfill $\Box$

Now, we give an expression of $\cov(J_{h,n}(\ob_1), J_{h,n}(\ob_2))$ in terms of $H_{h,k}^{(n)}$ and $R_{h,g}^{(n)}$ in (\ref{eq:Hkn}) and (\ref{eq:Rhgn}), respectively.

\begin{lemma} \label{lemma:cov-exp}
Suppose that Assumption \ref{assum:C} holds for $\ell=2$. Then, for $\ob_1, \ob_2 \in \R^d$, 
\begin{eqnarray*}
&& \cov(J_{h,n}(\ob_1), J_{h,n}(\ob_2)) \\
&&~ = (2\pi)^{-d} H_{h,2}^{-1} |D_n|^{-1} \int_{\R^d} e^{-i\ubb^\top \ob_1} C(\ubb) 
\left(H_{h,2}^{(n)}(\ob_1-\ob_2) + R_{h,h}^{(n)}(\ubb,\ob_1-\ob_2) \right) d\ubb. 
\end{eqnarray*}
\end{lemma}
\textit{Proof}.
By using (\ref{eq:DFT2-2}) and using that $h(\cdot/\aB)$ has a support on $D_n$,
\begin{eqnarray*}
&& \cov(J_{h,n}(\ob_1), J_{h,n}(\ob_2)) \nonumber \\ 
&& ~~= (2\pi)^{-d} H_{h,2}^{-1} |D_n|^{-1} 
\int_{\R^{2d}} h(\xb/\aB) h(\yb/\aB) e^{-i(\xb^\top \ob_1 - \yb^\top \ob_2)} C(\xb-\yb) d\xb d\yb \nonumber  \\
&&~~ = (2\pi)^{-d} H_{h,2}^{-1} |D_n|^{-1} \int_{\R^d} d\ubb e^{-i\ubb^\top \ob_1} C(\ubb) \int_{\R^d} 
h((\ubb+\vbb)/\aB) h(\vbb/\aB) e^{-i\vbb^\top(\ob_1-\ob_2)} d\vbb  \nonumber  \\
&& ~~= (2\pi)^{-d} H_{h,2}^{-1} |D_n|^{-1} \int_{\R^d} d\ubb e^{-i\ubb^\top \ob_1}  C(\ubb)
\left( H_{h,2}^{(n)} (\ob_1 - \ob_2) +  R_{h,h}^{(n)}(\ubb,\ob_1-\ob_2)  \right).
\end{eqnarray*}
Here, we use change of variables $\ubb=\xb-\yb$ and $\vbb=\yb$ in the second identity. Thus, we get the desired result.
\hfill $\Box$

\vspace{0.5em}

Using the above lemma together with bound for $R_{h,h}^{(n)}(\cdot)$ in Theorem \ref{thm:Rhgn0}, we obtain the leading term of $\cov(J_n(\ob_1), J_n(\ob_2))$.

\begin{theorem} \label{thm:cov-exp}
Suppose that Assumptions \ref{assum:A}, \ref{assum:C} (for $\ell=2$) and \ref{assum:E}(i) hold. Let $f$ be the spectral density function
and let $H_{h,k}^{(n)}$ be defined as in (\ref{eq:Hkn}). Then, for $\ob_1, \ob_2 \in \R^d$,
\begin{equation} \label{eq:cov-exp}
\cov(J_{h,n}(\ob_1), J_{h,n}(\ob_2)) = |D_n|^{-1} H_{h,2}^{-1} f(\ob_1) H_{h,2}^{(n)}(\ob_1-\ob_2)
+o(1), \quad ~~ n\rightarrow \infty.
\end{equation} Here, the $o(1)$ error above is uniform over $\ob_1, \ob_2 \in \R^d$.
\end{theorem}
\textit{Proof}.
By using Lemma \ref{lemma:cov-exp} and $C(\cdot) = \lambda \delta (\cdot) + \gamma_{2,\text{red}}(\cdot)$, the first term in the expansion of 
$\cov(J_{h,n}(\ob_1), J_{h,n}(\ob_2))$ is
\begin{eqnarray*}
&& (2\pi)^{-d} \lambda |D_n|^{-1} H_{h,2}^{-1} H_{h,2}^{(n)}(\ob_1-\ob_2) 
+ |D_n|^{-1} H_{h,2}^{-1} H_{h,2}^{(n)}(\ob_1-\ob_2) \mathcal{F}^{-1}(\gamma_{2,\text{red}})(\ob_1) \\
&&~~= |D_n|^{-1} H_{h,2}^{-1} \left( (2\pi)^{-d} \lambda  +  \mathcal{F}^{-1}(\gamma_{2,\text{red}})(\ob_1) \right)
H_{h,2}^{(n)}(\ob_1-\ob_2) \\
&&~~= |D_n|^{-1} H_{h,2}^{-1} \cdot f(\ob_1) H_{h,2}^{(n)}(\ob_1-\ob_2).
\end{eqnarray*} Here, we use (\ref{eq:spectral1}) in the last identity. Similarly, the remainder term of the difference between $\cov(J_{h,n}(\ob_1), J_{h,n}(\ob_2))$ and $|D_n|^{-1} H_{h,2}^{-1} f(\ob_1) H_{h,2}^{(n)}(\ob_1-\ob_2)$ is bounded by 
\begin{equation}
C |D_n|^{-1} |R_{h,h}^{(n)}(\textbf{0}, \ob_1-\ob_2)|
+ C |D_n|^{-1} \int_{\R^d} |\gamma_{2,\text{red}}(\ubb)| |R_{h,h}^{(n)}(\ubb, \ob_1-\ob_2)| d\ubb.
\label{eq:cov-exp22}
\end{equation} By using Theorem \ref{thm:Rhgn0}, the first term above is $o(1)$ as $n \rightarrow \infty$ uniformly in $\ob_1, \ob_2 \in \R^d$. To bound the second term, by using  \ref{thm:Rhgn0} again, we have $|D_n|^{-1} |\gamma_{2,\text{red}}(\ubb)| |R_{h,h}^{(n)}(\ubb, \ob_1-\ob_2)| \leq C |\gamma_{2,\text{red}}(\ubb)| \in L^1(\R^d)$ and $\lim_{n \rightarrow \infty} |D_n|^{-1}|\gamma_{2,\text{red}}(\ubb)| |R_{h,h}^{(n)}(\ubb, \ob_1+\ob_2)| = 0$, $\ubb \in \R^d$. Therefore, by dominated convergence theorem, the second term above is $o(1)$ as $n \rightarrow \infty$ uniformly in $\ob_1, \ob_2 \in \R^d$. Thus, we get the desired result.
\hfill $\Box$

\subsection{Bounds on the terms involving covariances}
Let 
\begin{equation}
\begin{aligned} 
T_n(\tb_1, \tb_2) &=
|D_n|^{-1} \int_{\R^{2d}} 
h(\frac{\xb}{\aB}) h(\frac{\yb}{\aB}) h(\frac{\xb+\tb_1}{\aB}) h(\frac{\yb+\tb_2}{\aB}) \\
&~~\times
C(\xb-\yb) C(\tb_1 + \xb -\tb_2 -\yb) d\xb d\yb,~~\tb_1, \tb_2 \in \R^d.
\end{aligned}
\label{eq:Tn}
\end{equation} The term $T_n$ appears in the integrated form of $ \cov (J_{h,n}(\ob_1), J_{h,n}(\ob_2)) \cov (J_{h,n}(-\ob_1), J_{h,n}(-\ob_2))$ in the proof of Theorem \ref{thm:A1}. Below, we give an approximation of $T_n$. Let
\begin{equation} \label{eq:f-tilde}
\widetilde{f}(\ob) = f(\ob) -(2\pi)^{-d} \lambda, \qquad \ob \in \R^d,
\end{equation} and let
\begin{equation}
\begin{aligned} 
& R_n(\tb_1, \tb_2,\ob_1,\ob_2)  \\
& ~~= 
H_{h,2}^{(n)}(-\ob_1-\ob_2) R_{h,h}^{(n)}(\tb_2,\ob_1+\ob_2) +
H_{h,2}^{(n)}(\ob_1+\ob_2) R_{h,h}^{(n)}(\tb_1, -\ob_1-\ob_2) \\
&~~ ~~ +R_{h,h}^{(n)}(\tb_2,\ob_1+\ob_2)  R_{h,h}^{(n)}(\tb_1, -\ob_1-\ob_2),
\quad n \in \N,~~ \tb_1, \tb_2, \ob_1, \ob_2 \in \R^d.
\end{aligned}
\label{eq:Rn33}
\end{equation}

\begin{lemma} \label{lemma:C-expansion}
Let $\gamma_{2,\emph{red}}$ and $C(\cdot)$ be the reduced second-order cumulant intensity function and the
complete covariance function defined as in (\ref{eq:cov}). Suppose that Assumptions \ref{assum:A}, \ref{assum:C} (for $\ell=2$), and \ref{assum:B}(i) hold. Furthermore, the data taper $h$ is Lipschitz continuous on $[-1/2,1/2]^d$. Then, for $ \tb_1, \tb_2 \in D_n-D_n$,
\begin{eqnarray*}
T_n(\tb_1,\tb_2) &=& \lambda^2 \delta(\tb_1 - \tb_2) \left(
H_{h,4} + O(|D_n|^{-1/d})\|\tb_2\|\right)  \\
&&+ 2 \lambda \gamma_{2,\emph{red}}(\tb_1-\tb_2) \left( H_{h,4} + O(|D_n|^{-1/d})[\|\tb_1 \|+ \|\tb_2\|] \right)
\\
&&
+\int_{\R^{2d}} e^{i (\tb_1 - \tb_2)^\top \ob_2}
 \widetilde{f}(\ob_1) \widetilde{f}(\ob_2) \bigg( (2\pi)^{d} H_{h,4} F_{h^2,n}(\ob_1+\ob_2) \\
&&\quad
+|D_n|^{-1}  R_n(\tb_1, \tb_2,\ob_1,\ob_2)
 \bigg)d\ob_1 d\ob_2, \qquad n\rightarrow \infty,
\end{eqnarray*} where $F_{h^2,n}$ is the Fej\'{e}r kernel in (\ref{eq:fejer}) based on $h^2$.
\end{lemma}
\textit{Proof}. Recall $C(\xb) = \lambda \delta(\xb) + \gamma_{2,\emph{red}}(\xb)$. Substitute this to (\ref{eq:Tn}),
$T_n(\tb_1,\tb_2)$ can be decomposed into four terms. The first term is
\begin{eqnarray*}
&&|D_n|^{-1} \lambda^2 \int_{\R^{2d}} 
h(\frac{\xb}{\aB}) h(\frac{\yb}{\aB}) h(\frac{\xb+\tb_1}{\aB}) h(\frac{\yb+\tb_2}{\aB})
\delta(\xb-\yb) \delta(\tb_1 + \xb -\tb_2 -\yb) d\xb d\yb.
\end{eqnarray*} The above term is nonzero if and only if $\tb_1 = \tb_2$. Therefore, the first term is equivalent to
\begin{equation*}
|D_n|^{-1} \lambda^2 \delta(\tb_1 - \tb_2)\int_{\R^d} h(\frac{\yb}{\aB})^2 h(\frac{\yb+\tb_2}{\aB})^2 d\yb.
\end{equation*} By applying Lemma \ref{lemma:Hkn}, we have
\begin{equation*}
|D_n|^{-1} \lambda^2 \delta(\tb_1 - \tb_2)\int_{\R^d} h(\frac{\yb}{\aB})^2 h(\frac{\yb+\tb_2}{\aB})^2 d\yb =
 \lambda^2 \delta(\tb_1 - \tb_2)  \left( H_{h,4} + C \rho(\|\tb_2/\aB\|) \right), ~n\rightarrow \infty.
\end{equation*}  Here, we use $|D_n|^{-1}\int_{D_n} h(\xb/\aB)^4d\xb =\int_{[-1/2,1/2]^d} h(\xb)^4d\xb = H_{h,4}$ in the identity. Moreover, under Assumption \ref{assum:A}, it is easily seen that $\rho(\|\xb/\aB\|) \leq \|\xb/\aB\| = O(|D_n|^{-1/d}) \|\xb\|$ as $n\rightarrow \infty$. Therefore, the first term is $\lambda^2 \delta(\tb_1 - \tb_2)  \left( H_{h,4} + O(|D_n|^{-1/d}) \|\tb_2\| \right)$ as $n \rightarrow \infty$.

Similarly, the second and third terms are equal to
\begin{eqnarray*}
&&|D_n|^{-1} \lambda \int_{\R^{2d}} 
h(\frac{\xb}{\aB}) h(\frac{\yb}{\aB}) h(\frac{\xb+\tb_1}{\aB}) h(\frac{\yb+\tb_2}{\aB})
\delta(\xb-\yb) \gamma_{2,\emph{red}}(\tb_1 + \xb -\tb_2 -\yb) d\xb d\yb \\
&& = |D_n|^{-1} \lambda \gamma_{2,\emph{red}}(\tb_1-\tb_2) \int_{\R^d}  h(\frac{\yb}{\aB})^2 h(\frac{\yb+\tb_1}{\aB}) h(\frac{\yb+\tb_2}{\aB}) d\yb \\
&& = \lambda \gamma_{2,\emph{red}}(\tb_1-\tb_2) \left( H_{h,4} + O(|D_n|^{-1/d}) ( \|\tb_1 \|+ \|\tb_2\|) \right), \qquad n \rightarrow \infty.
\end{eqnarray*} Finally, by using (\ref{eq:cov-spectral}), the fourth term is equal to
\begin{eqnarray*}
&&|D_n|^{-1} \int_{\R^{2d}} 
h(\frac{\xb}{\aB}) h(\frac{\yb}{\aB}) h(\frac{\xb+\tb_1}{\aB}) h(\frac{\yb+\tb_2}{\aB}) 
\gamma_{2,\emph{red}}(\xb-\yb) \gamma_{2,\emph{red}}(\tb_1 + \xb -\tb_2 -\yb) d\xb d\yb \\
&& =|D_n|^{-1} \int_{\R^{2d}} h(\frac{\xb}{\aB}) h(\frac{\yb}{\aB}) h(\frac{\xb+\tb_1}{\aB}) h(\frac{\yb+\tb_2}{\aB})  \\
&&~~\times
\int_{\R^{2d}} \widetilde{f}(\ob_1) \widetilde{f}(\ob_2) e^{i(\xb-\yb)^{\top} (\ob_1 + \ob_2)} e^{i (\tb_1- \tb_2)^\top \ob_2}
d\ob_1 d\ob_2 d\xb d\yb.
\end{eqnarray*} Since $\widetilde{f} \in L^1(\R^d)$ due to Assumption \ref{assum:B}(i), we can apply Fubini's theorem and get
\begin{eqnarray*}
&& |D_n|^{-1}  \int_{\R^{2d}}  h(\frac{\xb}{\aB}) h(\frac{\yb}{\aB}) h(\frac{\xb+\tb_1}{\aB}) h(\frac{\yb+\tb_2}{\aB})  \\
&&~~\times
\int_{\R^{2d}} \widetilde{f}(\ob_1) \widetilde{f}(\ob_2) e^{i(\xb-\yb)^{\top} (\ob_1 + \ob_2)} e^{i (\tb_1- \tb_2)^\top \ob_2}
d\ob_1 d\ob_2 d\xb d\yb \\
&& = |D_n|^{-1} \int_{\R^{2d}}  e^{i (\tb_1 - \tb_2)^\top \ob_2}
 \widetilde{f}(\ob_1) \widetilde{f}(\ob_2)
\left( \int_{\R^d}  h(\frac{\yb}{\aB})  h(\frac{\yb+\tb_2}{\aB}) e^{-i\yb^{\top} (\ob_1+\ob_2)}  d\yb  \right) \\
&&~~\times \left( \int_{\R^d} h(\frac{\xb}{\aB})h(\frac{\xb+\tb_1}{\aB}) e^{i\xb^{\top} (\ob_1+\ob_2)}  d\xb  \right) d\ob_1 d\ob_2 \\
&&= \int_{\R^{2d}} e^{i (\tb_1 - \tb_2)^\top \ob_2}
 \widetilde{f}(\ob_1) \widetilde{f}(\ob_2) \left( |D_n|^{-1}  |H_{h,2}^{(n)}(\ob_1+\ob_2)|^2
+ |D_n|^{-1} R_n(\tb_1, \tb_2,\ob_1,\ob_2)
 \right)d\ob_1 d\ob_2,
\end{eqnarray*}  where $R_n(\tb_1, \tb_2, \ob_1, \ob_2)$ is defined as in (\ref{eq:Rn33}).

Next, from (\ref{eq:fejer}), the Fej\'{e}r kernel based on $h^2$ is
\begin{equation} \label{eq:Fh2n}
F_{h^2,n}(\ob) = (2\pi)^{-d} H_{h,4}^{-1} |D_n|^{-1} |H_{h,2}^{(n)}(\ob)|^2.
\end{equation}
Substitute this into the above, we have
\begin{eqnarray*}
&& |D_n|^{-1}  \int_{\R^{2d}} h(\frac{\xb}{\aB}) h(\frac{\yb}{\aB}) h(\frac{\xb+\tb_1}{\aB}) h(\frac{\yb+\tb_2}{\aB})  \\
&&~~\times
\int_{\R^{2d}} \widetilde{f}(\ob_1) \widetilde{f}(\ob_2) e^{i(\xb-\yb)^{\top} (\ob_1 + \ob_2)} e^{i (\tb_1- \tb_2)^\top \ob_2}
d\ob_1 d\ob_2 d\xb d\yb \\
&&= \int_{\R^{2d}} e^{i (\tb_1 - \tb_2)^\top \ob_2}
 \widetilde{f}(\ob_1) \widetilde{f}(\ob_2) \left( (2\pi)^{d} H_{h,4}  F_{h^2,n}(\ob_1+\ob_2)
+ |D_n|^{-1} R_n(\tb_1, \tb_2,\ob_1,\ob_2)
 \right)d\ob_1 d\ob_2.
\end{eqnarray*}
All together, we show the lemma.
\hfill $\Box$

\vspace{1em}

Next, we bound the term that are associated with $R_n(\tb_1, \tb_2,\ob_1,\ob_2)$ in (\ref{eq:Rn33}).

\begin{lemma} \label{lemma:Rn}
Suppose the same set of assumptions in Theorem \ref{thm:A1}(ii) holds. Let $\widehat{\phi}_M$ and $\widetilde{f}$ be defined as in (\ref{eq:phiM-FT}) and (\ref{eq:f-tilde}), respectively. Then,
\begin{eqnarray}
&&|D_n|^{-1} \int_{B_M^2} d\tb_1 d\tb_2 \widehat{\phi}_M (\tb_1) \widehat{\phi}_M(-\tb_2)
\int_{\R^{2d}} e^{i (\tb_1 - \tb_2)^\top \ob_2}
 \widetilde{f}(\ob_1) \widetilde{f}(\ob_2)  R_n(\tb_1, \tb_2,\ob_1,\ob_2) d\ob_1 d\ob_2 \nonumber  \\
&&~~= o(1), \quad n \rightarrow \infty.
\label{eq:A3-error-bound}
\end{eqnarray}
\end{lemma}
\textit{Proof}. Recall (\ref{eq:Rn33}). By using Theorem \ref{thm:Rhgn1}(ii), the first term in the expression of 
\\ $|D_n|^{-1/2}R_n(\tb_1, \tb_2,\ob_1,\ob_2)$ is bounded by
\begin{equation*}
 C|H_{h,2}^{(n)}(-\ob_1-\ob_2)| \sum_{i=0}^{m_d}
\sum_{\jb,\kb \in \Z^d} |h_{\jb}| |h_{\kb}|  \rho( \{ \| \kb\|+1\} \|\tb_2/\aB \|)^{1/2}  |c_{D_{n,i}(\tb_2)}
(\ob_1+\ob_2 - 2\pi(\jb+\kb)/\aB)|.
\end{equation*}
 Substitute the above into (\ref{eq:A3-error-bound}), the first term in the expansion of (\ref{eq:A3-error-bound}) is bounded by
\begin{eqnarray*}
&& C_{} \sum_{i=0}^{m_d} 
\sum_{\jb,\kb \in \Z^d}|h_{\jb}| |h_{\kb}|
\int_{B_M}  |\widehat{\phi}_M (\tb_1)| d\tb_1 \int_{B_M} |\widehat{\phi}_M(-\tb_2)|  \rho\left( ( \| \kb\|+1) \|\tb_2/\aB \|\right)^{1/2} d\tb_2 \\
&&~~\times
\int_{\R^{d}} 
 |\widetilde{f}(\ob_2)|
\int_{\R^{d}}  \left| \widetilde{f}(\ob_1)  c_{h^2,n}(-\ob_1-\ob_2)  c_{D_{n,i}(\tb_2)}(\ob_1+\ob_2 - 2\pi(\jb+\kb)/\aB) \right|  d\ob_1  d\ob_2.
\end{eqnarray*} Since $\sup_{\ob \in \R^d} |\widetilde{f}(\ob)| <\infty$ and $\widetilde{f} \in L^1(\R^d)$ due to Assumption \ref{assum:B}, we can apply Lemma \ref{lemma:Cn}(d) and get
\begin{equation*}
\int_{\R^{d}} 
 |\widetilde{f}(\ob_2)|
\int_{\R^{d}}  \left| \widetilde{f}(\ob_1)  c_{h^2,n}(-\ob_1-\ob_2)  c_{D_{n,i}(\tb_2)}(\ob_1+\ob_2 - 2\pi(\jb+\kb)/\aB) \right|  d\ob_1  d\ob_2 
 < \infty,
\end{equation*} uniformly over $\jb, \kb \in \Z^d$.
Thus, the above term is bounded by
\begin{eqnarray*}
C(m_{d}+1) \left( \sum_{\jb \in \Z^d} |h_{\jb}| \right) \left( \int_{B_M}  |\widehat{\phi}_M (\tb_1)| d\tb_1  \right)
\left(  \sum_{\kb \in \Z^d} |h_{\kb}| \int_{B_M} |\widehat{\phi}_M(-\tb_2)|  \rho\left( ( \| \kb\|+1) \|\tb_2/\aB \|\right)^{1/2} d\tb_2\right).
\end{eqnarray*}
We first note that $|h_{\kb}| \lim_{n\rightarrow \infty} \int_{B_M} |\widehat{\phi}_M(\tb)| 
\rho( \{ \|\kb\|+1\} \|\tb_2/\aB\|)^{1/2} d\tb_2 = 0$ due to the dominated convergence theorem. Moreover, since $\sum_{\kb \in \Z^d} |h_{\kb}| \int_{B_M} |\widehat{\phi}_M(\tb)| 
\rho( \{ \|\kb\|+1\} \|\tb_2/\aB\|)^{1/2} d\tb_2 < C \sum_{\kb \in \Z^d}|h_{\kb}|<\infty$, by applying dominated convergence theorem again, we show that
 the above term is $o(1)$ as $n\rightarrow \infty$.

Similarly, the second term in the decomposition of (\ref{eq:A3-error-bound}) is $o(1)$ as $n\rightarrow \infty$.

Lastly, we bound the third term. By using Theorem \ref{thm:Rhgn1}(ii) again, the third term in the decomposition of (\ref{eq:A3-error-bound}) is bounded by
\begin{eqnarray*}
&& C \sum_{p,q=0}^{m_d} 
\sum_{\jb_1,\kb_1, \jb_2, \kb_2 \in \Z^d} |h_{\jb_1}| |h_{\kb_1}| |h_{\jb_2}| |h_{\kb_2}|
\int_{B_M^2}  d\tb_1 d\tb_2 |\widehat{\phi}_M (\tb_1)| |\widehat{\phi}_M(-\tb_2)| 
\rho\left( ( \|\kb_1\|+1) \|\tb_1/\aB\|\right)^{1/2}
\\
&&~~\times
\rho\left( ( \| \kb_2\|+1) \|\tb_2/\aB \|\right)^{1/2} \int_{\R^{2d}} 
 |\widetilde{f}(\ob_1)| |\widetilde{f}(\ob_2)|
|c_{D_{n,p}(\tb_1)}(- 2\pi (\jb_1+\kb_1)/\aB - (\ob_1+\ob_2))| \\
&&~~ \times 
|c_{D_{n,q}(\tb_2)} (-2\pi(\jb_2+\kb_2)/\aB + (\ob_1+\ob_2) )| d\ob_1 d\ob_2.
\end{eqnarray*}
By using Lemma \ref{lemma:Cn}(d) and $\widetilde{f} \in L^1(\R^d)$, the integral $\int_{\R^{2d}} (\cdots) d\ob_1 d\ob_2$ is bounded above and the upper bound is depends only on $\widetilde{f}$. Therefore, the above is bounded by
\begin{equation*}
C (m_{d}+1)^2 
\left( \sum_{\jb \in \Z^d} |h_{\jb}| \right)^2
\left( \sum_{\kb \in \Z^d} |h_{\kb}|  \int_{B_M} |\widehat{\phi}_M(\tb)| 
\rho\left( ( \| \kb\|+1) \|\tb/\aB \|\right)^{1/2}
d\tb \right)^2.
\end{equation*}
By using a similar dominated convergence argument above, the third term is $o(1)$ as $n\rightarrow \infty$. 

All together, we show prove the lemma.
\hfill $\Box$

\vspace{0.5em}

Now, we are ready to compute the limit of $A_1$ in Section \ref{sec:A1-2}. Recall
\begin{equation*}
A_1 = 
|D_n|\int_{\R^{2d}} \phi_M(\ob_1) \phi_M(\ob_2) \cov (J_{h,n}(\ob_1), J_{h,n}(\ob_2)) \cov (J_{h,n}(-\ob_1), J_{h,n}(-\ob_2)) d\ob_1 d\ob_2.
\end{equation*}

\begin{theorem} \label{thm:limA1}
Suppose the same set of assumptions in Theorem \ref{thm:A1}(ii) holds. Then,
\begin{equation*}
\lim_{n\rightarrow \infty} A_{1} =(2\pi)^{d} (H_{h,4}/H_{h,2}^2)  \int_{\R^d} f(\ob)^2  \phi_M(\ob)^2 d\ob.
\end{equation*} 
\end{theorem}
\textit{Proof}.
First, by using (\ref{eq:DFT2-2}), we have
\begin{eqnarray*}
A_{1} &=& (2\pi)^{-2d} H_{h,2}^{-2} |D_n|^{-1} \int_{\R^{2d}}d\ob_1 d\ob_2 \phi_{M}(\ob_1) \phi_{M}(\ob_2) 
\int_{D_n^2}
 d\xb d\yb   h(\xb/\aB) h(\yb/\aB)  \\
&& \times  e^{-i(\xb^\top \ob_1 - \yb^\top \ob_2)} C(\xb-\yb)  \int_{D_n^2} d\ubb d\vbb h(\ubb/\aB) h(\vbb/\aB) e^{i(\ubb^\top \ob_1 - \vbb^\top \ob_2)}  C(\ubb-\vbb).
\end{eqnarray*} 
By using Cauchy-Schwarz inequality, the above is absolutely integrable, thus, we can interchange the summations and get
\begin{eqnarray*}
A_1 &=& (2\pi)^{-2d} H_{h,2}^{-2} |D_n|^{-1} \int_{D_n^2} d\xb d\yb h(\xb/\aB) h(\yb/\aB) C(\xb-\yb) \int_{D_n^2} d\ubb d\vbb h(\ubb/\aB) h(\vbb/\aB) C(\ubb-\vbb) \\
&&\times 
\int_{\R^d} \phi_{M}(\ob_1)  e^{i(\ubb-\xb)^\top \ob_1}  d\ob_1 
\int_{\R^d}  \phi_{M}(\ob_2) e^{i(\yb-\vbb)^\top \ob_2} d\ob_2  \\
 &=& (2\pi)^{-2d} H_{h,2}^{-2} |D_n|^{-1} \int_{D_n^4} 
 h(\xb/\aB) h(\yb/\aB)  h(\ubb/\aB) h(\vbb/\aB) \\
&& \times
C(\xb-\yb) C(\ubb-\vbb) \widehat{\phi}_{M}(\ubb-\xb) \widehat{\phi}_{M}(\yb-\vbb) d\xb d\yb d\ubb d\vbb \\
&=& (2\pi)^{-2d} H_{h,2}^{-2}  \int_{D_n-D_n} \int_{D_n-D_n} \widehat{\phi}_{M} (\tb_1) \widehat{\phi}_{M}(-\tb_2) 
T_n(\tb_1,\tb_2) d\tb_1 d\tb_2,
\end{eqnarray*} where $T_n(\tb_1,\tb_2)$ is from (\ref{eq:Tn}).
In the above, we use an inverse transform of (\ref{eq:phiM-FT}) in the second identity and 
the change of variables $\tb_1 = \ubb - \xb$ and $\tb_2 = \vbb-\yb$ in the last identity.

Next, recall $\widetilde{f}$ in (\ref{eq:f-tilde}).
 Then, by Assumption \ref{assum:B}(i), $\widetilde{f} \in L^1(\R^d)$ and has the Fourier representation as in (\ref{eq:cov-spectral}). We note that $B_M \subset D_n - D_n$ for large enough $n\in \N$ due to Assumption \ref{assum:A}. Thus, by using this fact together with Lemma \ref{lemma:C-expansion}, for large enough $n\in \N$, we have $A_1 = A_{11} + 2A_{12} + A_{13}$, where
\begin{eqnarray*}
A_{11} &=&
(2\pi)^{-2d} H_{h,2}^{-2} \lambda^2  \int_{B_M^2}\widehat{\phi}_{M} (\tb_1) \widehat{\phi}_{M}(-\tb_2) \delta(\tb_1-\tb_2)
\left(
H_{h,4} + O(|D_n|^{-1/d}) \|\tb_2\|
 \right)
 d\tb_1 d\tb_2, \\
 A_{12} &=&
(2\pi)^{-2d} H_{h,2}^{-2} \lambda \int_{B_M^2}\widehat{\phi}_{M} (\tb_1) \widehat{\phi}_{M}(-\tb_2)
\gamma_{2,\emph{red}}(\tb_1-\tb_2) \left[ H_{h,4} + O(|D_n|^{-1/d})( \|\tb_1\| + \|\tb_2\| ) \right]
  d\tb_1 d\tb_2,
\end{eqnarray*} as $n\rightarrow \infty$,
and
\begin{eqnarray*}
A_{13} &=&
(2\pi)^{-2d} H_{h,2}^{-2} \int_{B_M^2}\widehat{\phi}_{M} (\tb_1) \widehat{\phi}_{M}(-\tb_2)
\int_{\R^{2d}}   e^{i (\tb_1 - \tb_2)^\top \ob_2}
 \widetilde{f}(\ob_1) \widetilde{f}(\ob_2) \\
&&~~\times \left( (2\pi)^{d} H_{h,4} F_{h^2,n}(\ob_1+\ob_2)
+|D_n|^{-1}  R_n(\ob_1,\ob_2,\tb_1, \tb_2)
 \right)d\ob_1 d\ob_2 d\tb_1 d\tb_2, \quad n\rightarrow \infty,
\end{eqnarray*} where $F_{h^2,n}$ is Fej\'{e}r kernel based on $h^2$ defined as in (\ref{eq:Fh2n})
and $R_n(\ob_1,\ob_2,\tb_1, \tb_2)$ is the remainder term defined as in (\ref{eq:Rn33}).

\vspace{0.5em}

We bound each term above. The first term is 
\begin{equation}
\begin{aligned} 
A_{11} &= (H_{h,4}/H_{h,2}^2) R^2 \int_{B_M} |\widehat{\phi}_{M} (\tb_1)|^2 d\tb_1 +O(|D_n|^{-1/d})  \\
&= (2\pi)^d  (H_{h,4}/H_{h,2}^2) R^2 \int_{\R^d} |\phi_{M}(\ob)|^2 d\ob + o(1), \quad n\rightarrow \infty,
\end{aligned}
\label{eq:A11}
\end{equation} where $R = (2\pi)^{-d} \lambda$.
Here, the last identity is due to Plancherel theorem.
By using (\ref{eq:cov-spectral}) and (\ref{eq:phiM-FT}) together with Fubini's theorem, the second term is
\begin{equation}
\begin{aligned}
& A_{12} \\
&= (2\pi)^{-d} (H_{h,4}/H_{h,2}^2) R  \int_{B_M^2}  \widehat{\phi}_{M} (\tb_1) \widehat{\phi}_{M}(-\tb_2)
\int_{\R^d} \widetilde{f}(\ob) e^{i\ob^\top (\tb_1 -\tb_2)} d\ob d\tb_1 d\tb_2 + O(|D_n|^{-1/d}) \\
&= (2\pi)^{-d}  (H_{h,4}/H_{h,2}^2) R  \int_{\R^d} \widetilde{f}(\ob) \left( \int_{B_M} \widehat{\phi}_{M}(-\tb_2) e^{-i\ob^\top \tb_2} d\tb_2 \right)
\left( \int_{B_M} \widehat{\phi}_{M}(\tb_1) e^{i\ob^\top \tb_1} d\tb_1 \right)d\ob+ o(1) \\
&= (2\pi)^{d} (H_{h,4}/H_{h,2}^2) R \int_{\R^d} \widetilde{f}(\ob) \phi_{M}(\ob)^{2} d\ob + o(1), \quad n\rightarrow \infty.
\end{aligned}
\label{eq:A12}
\end{equation} Finally, by using Fubini's theorem and Lemmas \ref{lemma:Cn}(b) and \ref{lemma:Rn}, the third term is
\begin{equation}
\begin{aligned}
A_{13} &= (2\pi)^{-d} (H_{h,4}/H_{h,2}^{2})
\int_{\R^{2d}}   d\ob_1 d\ob_2 \widetilde{f}(\ob_1) \widetilde{f}(\ob_2) F_{h^2,n}(\ob_1+\ob_2) \\
&~~\times
\int_{B_M}\widehat{\phi}_{M} (\tb_1) e^{i\tb_1^\top \ob_2}    d\tb_1
\int_{B_M}\widehat{\phi}_{M} (-\tb_2) e^{-i\tb_2^\top \ob_2}   d \tb_2 
+ o(1) \\
&=
(2\pi)^{d}   (H_{h,4}/H_{h,2}^{2})\int_{\R^{d}}  \widetilde{f}(\ob_2)  \phi_M(-\ob_2)^2 \int_{\R^{d}}  \widetilde{f}(\ob_1)F_{h^2,n}(\ob_1+\ob_2)  d\ob_1 d\ob_2 + o(1)  \\
&=
(2\pi)^{d}  (H_{h,4}/H_{h,2}^{2}) \int_{\R^{d}}  \widetilde{f}(\ob_2)^2 \phi_M(\ob_2)^2 d\ob_2 + o(1), \qquad n \rightarrow \infty.
\end{aligned}
\label{eq:A13}
\end{equation} 
Combining (\ref{eq:A11})--(\ref{eq:A13}), we conclude
\begin{eqnarray*} 
\lim_{n\rightarrow \infty} A_1 &=&  (2\pi)^{d} (H_{h,4}/H_{h,2}^2) \int_{\R^d}
\left( \widetilde{f}(\ob)^2 + 2R \widetilde{f}(\ob) + R^2 \right) \phi_M(\ob)^2 d\ob \nonumber \\
&=& (2\pi)^{d} (H_{h,4}/H_{h,2}^2)  \int_{\R^d} f(\ob)^2  \phi_M(\ob)^2 d\ob.
\end{eqnarray*} Thus, we prove the theorem.
\hfill $\Box$

\subsection{Bounds on the terms involving higher order cumulants} \label{sec:cum4}

We give an expression of the complete fourth-order cumulant of the DFTs.
Through a simple combinatorial argument, the fourth-order complete reduced cumulant, denoted $\kappa_{4,\text{red}}(\xb, \yb,\zb)$, can be written as a sum of the 15 different reduced cumulant functions:
\begin{equation}
\begin{aligned}
\kappa_{4,\text{red}}(\xb, \yb,\zb) 
&= \gamma_{4,\text{red}}(\xb,\yb,\zb) + \bigg[
\gamma_{3,\text{red}} (\xb,\zb) \delta(\xb-\yb) + 
\gamma_{3,\text{red}} (\xb,\yb) \delta(\xb-\zb)  \\
& \qquad \qquad +
\gamma_{3,\text{red}} (\xb,\yb) \delta(\yb-\zb)
 +
\gamma_{3,\text{red}} (\xb-\zb,\yb-\zb) \left\{ \delta(\xb) + \delta(\yb) + \delta(\zb)\right\}
\bigg] \\
& + \bigg[ 
\gamma_{2,\text{red}} (\xb)  \left(  \delta(\yb) \delta(\zb) + \delta(\xb-\yb)\delta(\xb-\zb) + \delta(\xb-\yb)\delta(\zb) 
+ \delta(\xb-\zb)\delta(\yb) \right) \\
& \qquad +
\gamma_{2,\text{red}} (\yb) \left( \delta(\xb) \delta(\zb) + \delta(\xb)\delta(\yb-\zb) \right) + 
\gamma_{2,\text{red}} (\zb) \delta(\xb) \delta(\yb) \bigg] \\
& +\lambda \delta(\xb)\delta(\yb) \delta(\zb), \qquad \xb, \yb, \zb \in \R^d.
\end{aligned}
\label{eq:4th-cum}
\end{equation}

\begin{lemma} \label{lemma:kappa4}
Let $X$ be a fourth-order stationary point process on $\R^d$ and let
$\kappa_4$ be the fourth-order complete cumulant density function defined as in (\ref{eq:4th-cum}). Suppose that Assumption \ref{assum:C} holds for $\ell=4$. Then, for $\ob_1, \dots, \ob_4 \in \R^d$, 
\begin{eqnarray*}
&& \cum \left( J_{h,n}(\ob_1), J_{h,n}(\ob_2), J_{h,n}(\ob_3), J_{h,n}(\ob_4) \right)
=  (2\pi)^{-2d} H_{h,2}^{-2} |D_n|^{-2} 
  \\
&&~~ \times \int_{D_n^4}
\left( \prod_{j=1}^{4} h(\tb_j/\aB) \right)
 \exp(-i \sum_{j=1}^{4} \tb_j^\top \ob_j) \kappa_4(\tb_1-\tb_4, \tb_2-\tb_4, \tb_3-\tb_4) \prod_{j=1}^{4}d\tb_j.
\end{eqnarray*}
\end{lemma}
\textit{Proof}. The proof is similar to that of the proof of Theorem \ref{thm:DFT2}. We omit the details.
\hfill $\Box$

\vspace{0.5em}

Using the above expression, we calculate the limit of $A_3$ in Section \ref{sec:A1-2}.
Recall 
\begin{equation*}
A_3 = |D_n|\int_{\R^{2d}} \phi_M(\ob_1) \phi_M(\ob_2) \cum (J_{h,n}(\ob_1), J_{h,n}(-\ob_1), J_{h,n}(\ob_2), J_{h,n}(-\ob_2)) d\ob_1 d\ob_2.
\end{equation*}

\begin{theorem} \label{thm:limA3}
Suppose the same set of assumptions in Theorem \ref{thm:A1}(ii) holds. Then,
\begin{equation*}
\lim_{n\rightarrow \infty} A_{3} = (2\pi)^d (H_{h,4}/H_{h,2}^{2}) \int_{\R^{2d}} \phi_{M}(\blambda_1) \phi_M(\blambda_3) f_4(\blambda_1, -\blambda_1, \blambda_3) d\blambda_1 d\blambda_3,
\end{equation*} where $f_4$ is the fourth-order spectrum.
\end{theorem}
\textit{Proof}.
By using Lemma \ref{lemma:kappa4} and Fubini's theorem, $A_3$ can be written as
\begin{eqnarray*}
A_3 &=& (2\pi)^{-2d}H_{h,2}^{-2} |D_n|^{-1}
 \int_{D_n^4} 
\left( \prod_{j=1}^{4} h(\tb_j/\aB) \right)
\kappa_{4,\text{red}}(\tb_1-\tb_4, \tb_2-\tb_4, \tb_3-\tb_4) \\
&& \times
\left( \int_{\R^d} \phi_M(\ob_1) e^{i(\tb_2-\tb_1)^\top \ob_1} d\ob_1\right)
\left( \int_{\R^d} \phi_M(\ob_2) e^{i(\tb_4-\tb_3)^\top \ob_2} d\ob_2 \right)
 d\tb_1 d\tb_2 d\tb_3 d\tb_4 \\
&=& (2\pi)^{-2d} H_{h,2}^{-2}|D_n|^{-1} \int_{D_n^4} 
\left( \prod_{j=1}^{4} h(\tb_j/\aB) \right)
\kappa_{4,\text{red}}(\tb_1-\tb_4, \tb_2-\tb_4, \tb_3-\tb_4) \\
&&~~\times \widehat{\phi}_M(\tb_2-\tb_1)
\widehat{\phi}_M(\tb_4-\tb_3) d\tb_1 d\tb_2 d\tb_3 d\tb_4.
\end{eqnarray*} Let $\widetilde{\kappa}_{4,\text{red}}(\xb,\yb,\zb)= \kappa_{4,\text{red}}(\xb,\yb,\zb) - \lambda \delta(\xb) \delta(\yb) \delta(\zb)$. Then, from (\ref{eq:4th-cum}) and (\ref{eq:4th-spec}), the inverse Fourier transform of $\widetilde{\kappa}_{4,\text{red}}$ is
$\widetilde{f}_4(\xb,\yb,\zb) = f_4(\xb,\yb,\zb) - (2\pi)^{-3d} \lambda$. By using a similar decomposition as in Lemma \ref{lemma:C-expansion}, we have $A_{3} = A_{31} + A_{32}$, where
\begin{eqnarray*}
A_{31} &=&  (2\pi)^{-2d} H_{h,2}^{-2} |D_n|^{-1} \lambda \int_{D_n^4} 
\left( \prod_{j=1}^{4} h(\tb_j/\aB) \right) \delta(\tb_1-\tb_4) \delta(\tb_2-\tb_4) \delta(\tb_3-\tb_4)\\
&&~~\times \widehat{\phi}_M(\tb_2-\tb_1)
\widehat{\phi}_M(\tb_4-\tb_3) d\tb_1  d\tb_2 d\tb_3 d\tb_4  \\
&=& (2\pi)^{-2d}  H_{h,2}^{-2}  |D_n|^{-1} \lambda  \widehat{\phi}_M(\textbf{0})^2 \int_{D_n} h(\tb/\aB)^4 d\tb
=(2\pi)^{-2d} (H_{h,4}/H_{h,2}^2) \lambda\widehat{\phi}_M(\textbf{0})^2
\end{eqnarray*} and
\begin{eqnarray*}
A_{32} &=& (2\pi)^{-2d} H_{h,2}^{-2} |D_n|^{-1} \int_{D_n^4} 
\left( \prod_{j=1}^{4} h(\tb_j/\aB) \right)
\widetilde{\kappa}_{4,\text{red}}(\tb_1-\tb_4, \tb_2-\tb_4, \tb_3-\tb_4) \\
&&~~\times \widehat{\phi}_M(\tb_2-\tb_1)
\widehat{\phi}_M(\tb_4-\tb_3) d\tb_1 d\tb_2 d\tb_3 d\tb_4 \\
&=& (2\pi)^{-2d} H_{h,2}^{-2} |D_n|^{-1} \int_{D_n-D_n} \int_{D_n-D_n}  d\ubb d\vbb \widehat{\phi}_M(\ubb) \widehat{\phi}_M(\vbb)  \int_{\R^{2d}} 
h(\tb_1/\aB) h( (\tb_1+\ubb)/\aB)
\\
&&~~\times
 h(\tb_3/\aB) h( (\tb_3+\vbb)/\aB)
\widetilde{\kappa}_{4,\text{red}}(\tb_1-\tb_3-\vbb, \tb_1-\tb_3 + \ubb - \vbb, -\vbb)
 d\tb_1 d\tb_3.
\end{eqnarray*} Here, we use change of variables $\ubb = \tb_2-\tb_1$ and $\vbb = \tb_4-\tb_3$ in the second identity above. To obtain an expression of $A_{31}$, by using (\ref{eq:phi-FT}),
$\widehat{\phi}_M(\textbf{0})^2 = \widehat{\phi}(\textbf{0})^2 = \int_{D^2} \phi(\ob_1)\phi(\ob_2) d\ob_1 d\ob_2$. Therefore,
\begin{equation}
A_{31} = (2\pi)^{-2d} (H_{h,4}/H_{h,2}^2) \lambda \int_{D^2} \phi(\ob_1)\phi(\ob_2) d\ob_1 d\ob_2.
\label{eq:A31} 
\end{equation}
Obtaining an expression for $\lim_{n\rightarrow \infty} A_{32}$ is similar to that in
deriving an expression for $\lim_{n\rightarrow \infty} A_{13}$ above.
By using similar techniques as in Lemmas \ref{lemma:C-expansion} and \ref{lemma:Rn}, it can be shown that for large $n\in \N$,
\begin{eqnarray*}
A_{32} &=& (2\pi)^{-d} (H_{h,4}/H_{h,2}^{2}) \int_{B_M^2} d\ubb d\vbb \widehat{\phi}_M(\ubb) \widehat{\phi}_M(\vbb) \int_{\R^{3d}} \widetilde{f}_4(\blambda_1, \blambda_2, \blambda_3) F_{h^2,n}(\blambda_1 + \blambda_2) \\
&&~~\times \exp\left( i (-\blambda_1^\top \vbb + \blambda_2^\top (\ubb-\vbb) - \blambda_3^\top \vbb)\right) d\blambda_1 d\blambda_2 d\blambda_3  + o(1) \\
&=& (2\pi)^{-d} (H_{h,4}/H_{h,2}^{2}) 
\int_{B_M^2} d\ubb d\vbb \widehat{\phi}_M(\ubb) \widehat{\phi}_M(\vbb)
\int_{\R^{2d}} \widetilde{f}_4(\blambda_1, -\blambda_1, \blambda_3) 
e^{-i(\blambda_1^\top \ubb + \blambda_3^\top \vbb)} d\blambda_1 d\blambda_3 + o(1)
 \\
&=& (2\pi)^{d} (H_{h,4}/H_{h,2}^{2}) \int_{\R^{2d}} \widetilde{f}_4(\blambda_1, -\blambda_1, \blambda_3)  
\phi_M(\blambda_1) \phi_M(\blambda_3) d\blambda_1 d\blambda_3+ o(1),\quad n\rightarrow \infty.
\end{eqnarray*} 
Here, we use Lemma \ref{lemma:Cn}(b) in the second identity and
Fubini's theorem and (\ref{eq:phiM-FT}) in the last identity.
Therefore,
\begin{equation}
\lim_{n\rightarrow \infty} A_{32} = (2\pi)^d (H_{h,4}/H_{h,2}^{2}) \int_{\R^{2d}} \phi_{M}(\blambda_1) \phi_M(\blambda_3) \widetilde{f}_4(\blambda_1, -\blambda_1, \blambda_3) d\blambda_1 d\blambda_3.
\label{eq:A32} 
\end{equation} 
Combining (\ref{eq:A31}) and (\ref{eq:A32}), we have
\begin{equation*}
\lim_{n\rightarrow \infty} A_{3} = (2\pi)^d (H_{h,4}/H_{h,2}^{2}) \int_{\R^{2d}} \phi_{M}(\blambda_1) \phi_M(\blambda_3) f_4(\blambda_1, -\blambda_1, \blambda_3) d\blambda_1 d\blambda_3.
\end{equation*} Thus, we obtain the limit of $A_3$.
\hfill $\Box$

\vspace{0.5em}

To end this section, we obtain the bounds for the general higher order cumulant term.
\begin{lemma} \label{lemma:k4-bound}
Suppose that Assumptions \ref{assum:A} and \ref{assum:C} (for some $\ell \in \{2, 3, \dots \}$) hold. Let
\begin{equation*}
\mathcal{Z}_{j,n} = |D_n|^{-\alpha_j} \sum_{\xb \in X \cap D_n} g_{j,n}(\xb), \quad j \in \{1, \dots, \ell\},
\end{equation*} where $\alpha_j \in [0,\infty)$ and $g_{j,n}(\xb)$ is a bounded function on $\R^d$ uniformly in $n \in \N$ and $\xb \in \R^d$. 
Then, we have
\begin{equation*}
\big| \cum(\mathcal{Z}_{1,n} , \dots, \mathcal{Z}_{j,n})
\big| = O(|D_n|^{-\sum_{k=1}^{j} \alpha_k + 1}), \qquad j \in \{2, \dots, \ell\}.
\end{equation*}
\end{lemma}
\textit{Proof}. We will only show the above for $(j,\ell) = (4,4)$ under fourth-order stationarity. 
The general cases are treated similarly (see the statement after Assumption \ref{assum:C}). Let $\alpha = \sum_{k=1}^{4} \alpha_k$.
By generalizing Lemma \ref{lemma:kappa4}, we have
\begin{equation*}
\cum(\mathcal{Z}_{1,n} , \dots, \mathcal{Z}_{4,n}) 
= |D_n|^{-\alpha} \int_{D_n^4}\prod_{j=1}^{4} g_{j,n}(\tb_j) \times \kappa_{4,\text{red}}(\tb_1-\tb_4, \tb_2-\tb_4, \tb_3-\tb_4) d\tb_1 d\tb_2 d\tb_3 d\tb_4,
\end{equation*} where $\kappa_{4,\text{red}}(\cdot, \cdot, \cdot)$ is defined as in (\ref{eq:4th-cum}). 
Let $\sup_{n \in \N}\sup_{\xb \in \R^d} |g_{j,n}(\xb)| < C_j <\infty$, $j \in \{1,2,3,4\}$, and let $C = \max\{C_1, \dots, C_4\}$. Then, we have
\begin{eqnarray*}
|\cum(\mathcal{Z}_{1,n} , \dots, \mathcal{Z}_{4,n})  | &\leq& C^4 |D_n|^{-\alpha} \int_{D_n^4} |\kappa_{4,\text{red}}(\tb_1-\tb_4, \tb_2-\tb_4, \tb_3-\tb_4)| d\tb_1 d\tb_2 d\tb_3 d\tb_4 \\
&\leq& C^4 |D_n|^{-\alpha} \int_{D_n} \int_{D_n-\tb_4} \int_{D_n-\tb_4} \int_{D_n-\tb_4}
|\kappa_{4,\text{red}}(\xb, \yb,\zb)| d\xb d\yb d\zb d\tb_4 \\
&\leq& C^4 |D_n|^{-\alpha}  \left( \int_{D_n} d\tb_4 \right) \left( \int_{\R^{3d}} |\kappa_{4,\text{red}}(\xb, \yb,\zb)| d\xb d\yb d\zb\right)
= O(|D_n|^{-\alpha+1}).
\end{eqnarray*} Here, we use change of variables $\xb = \tb_1-\tb_4$, $\yb= \tb_2-\tb_4$, and $\zb = \tb_3 - \tb_4$ in the second inequality, and the last identity is due to absolute integrability of $\kappa_{4,\text{red}} (\cdot, \cdot, \cdot)$ under Assumption \ref{assum:C} for $\ell=4$. Thus, we prove the lemma for $(j,\ell) =(4,4)$ under fourth-order stationarity.
\hfill $\Box$


\section{Asymptotic equivalence of the periodograms}  \label{sec:some-cumulants}

In this section, we obtain some cumulant bounds that appear in the first and second moments of the different $\widehat{I}_{h,n}(\ob) - I_{h,n}(\ob)$. These bounds are mainly used to prove an asymptotic equivalence between the feasible and infeasible integrated periodograms (see Theorem \ref{thm:A2}). However, the results derived in this section may also be of independent interest in periodogram-based methods for spatial point process. Throughout the section, we let $C \in (0,\infty)$ be a generic constant that varies line by line. 
Recall (\ref{eq:In-feas-h}):
\begin{equation} 
\begin{aligned}
& \widehat{I}_{h,n}(\ob) - I_{h,n}(\ob) \\
&~~= - (\widehat{\lambda}_{h,n} - \lambda) \left( c_{h,n}(\ob)  J_{h,n}(-\ob) + c_{h,n}(-\ob)  J_{h,n}(\ob) \right) + (\widehat{\lambda}_{h,n} - \lambda)^2 |c_{h,n}(\ob)|^2 \\
&~~= R_1(\ob) + R_2(\ob).
\end{aligned}
\label{eq:In-diff}
\end{equation} 
In the following lemma, we bound the cumulants that are made of $\widehat{\lambda}_{h,n}$ and 
\begin{equation} \label{eq:mathcalK}
K_{h,n}(\ob) = c_{h,n}(\ob) J_{h,n}(-\ob) + c_{h,n}(-\ob)  J_{h,n}(\ob), \quad \ob\in \R^d.
\end{equation}

\begin{lemma} \label{lemma:cumulants}
Suppose that Assumptions \ref{assum:C}(for $\ell=2$) and \ref{assum:E}(i) hold.
Then, we obtain the following two bounds:
\begin{itemize}
\item[(a)] For $\ob \in \R^d$,
$|\cum( \widehat{\lambda}_{h,n}, K_{h,n}(\ob))| \leq C|D_n|^{-1} |c_{h,n}(\ob)||c_{h^2,n}(\ob)| + |D_n|^{-1/2} |c_{h,n}(\ob)|o(1)$ as $n\rightarrow \infty$, where $o(1)$ error is uniform over $\ob \in \R^d$.


\item[(b)] $\cum(\widehat{\lambda}_{h,n}, \widehat{\lambda}_{h,n}) = \var(\widehat{\lambda}_{h,n}) \leq C |D_n|^{-1}$.
\end{itemize}
If we further assume Assumption \ref{assum:C} holds for $\ell=4$. Then, we have
\begin{itemize}
\item[(c)] For $\ob_1 ,\ob_2 \in \R^d$, 
$|\cum(\widehat{\lambda}_{h,n}, K_{h,n}(\ob_1), 
\widehat{\lambda}_{h,n},  K_{h,n}(\ob_2) )| \leq C |c_{h,n}(\ob_1)| |c_{h,n}(\ob_2)| |D_n|^{-2}
$.
\item[(d)] $|\cum(\widehat{\lambda}_{h,n}, \widehat{\lambda}_{h,n}, \widehat{\lambda}_{h,n}, \widehat{\lambda}_{h,n})| \leq C |D_n|^{-3}$.

\item[(e)] $\cum((\widehat{\lambda}_{h,n}-\lambda)^2,(\widehat{\lambda}_{h,n}-\lambda)^2)
= \var((\widehat{\lambda}_{h,n}-\lambda)^2) \leq C|D_n|^{-2}$.
\end{itemize}

\end{lemma}
\textit{Proof}. 
Recall $\widehat{\lambda}_{h,n} = H_{h,1}^{-1}|D_n|^{-1} \sum_{\xb \in X \cap D_n} h(\xb/\aB)$. (b) and (d) are straightforward due to Lemma \ref{lemma:k4-bound}.
To show (e), we note that
\begin{equation*}
\cum((\widehat{\lambda}_{h,n}-\lambda)^2,(\widehat{\lambda}_{h,n}-\lambda)^2)
= 2 \var(\widehat{\lambda}_{h,n})^2 + 
\cum(\widehat{\lambda}_{h,n}, \widehat{\lambda}_{h,n}, \widehat{\lambda}_{h,n}, \widehat{\lambda}_{h,n}).
\end{equation*} Thus, (e) follows from (b) and (d).

Next, we will show (a). We first note that $\widehat{\lambda}_{h,n} = (2\pi)^{d/2} H_{h,2}^{1/2}H_{h,1}^{-1} |D_n|^{-1/2} \mathcal{J}_{h,n}(\textbf{0})$. Therefore, by using Theorem \ref{thm:cov-exp},
\begin{eqnarray*}
\cum(\widehat{\lambda}_{h,n}, c_{h,n}(\ob) J_{h,n}(-\ob)) &=& 
(2\pi)^{d/2} H_{h,2}^{1/2}H_{h,1}^{-1} |D_n|^{-1/2}  c_{h,n}(\ob) \cum(J_{h,n}(\textbf{0}),  J_{h,n}(-\ob))
 \\
&=& C |D_n|^{-3/2}  c_{h,n}(\ob) f(\textbf{0}) H_{h,2}^{(n)}(-\ob) + |D_n|^{-1/2}  c_{h,n}(\ob) o(1),~~n \rightarrow \infty,
\end{eqnarray*}  where $f$ is the spectral density
and $o(1)$ error above is uniform over $\ob \in \R^d$. Since $f(\textbf{0}) <\infty$ under Assumption \ref{assum:C}(for $\ell=2$) and by using (\ref{eq:Cn-h}), we have
\begin{equation*}
| \cum(\widehat{\lambda}_{h,n}, c_{h,n}(\ob) J_{h,n}(-\ob)) | 
\leq C |D_n|^{-1} | c_{h,n}(\ob)| |c_{h^2,n}(\ob)|  + |D_n|^{-1/2}| c_{h,n}(\ob)| o(1), \quad n \rightarrow \infty.
\end{equation*}
Similarly, we have $| \cum(\widehat{\lambda}_{h,n}, c_{h,n}(-\ob) J_{h,n}(\ob))| \leq C |D_n|^{-1} | c_{h,n}(\ob)| |c_{h^2,n}(\ob)|  + |D_n|^{-1/2}| c_{h,n}(\ob)| o(1)$ as $n\rightarrow \infty$. Thus, we show (a).

Lastly, (c) is straightforward due to Lemma \ref{lemma:k4-bound}.
All together, we get the desired results.
\hfill $\Box$

\vspace{0.5em} 

Using the bounds derived above, we obtain the bounds for the first-order moments for $R_1(\ob)$ and $R_2(\ob)$ and the second-order moment for $R_2(\ob)$.

\begin{theorem} \label{thm:In-moments}
Suppose that Assumptions \ref{assum:C}(for $\ell=2$) and \ref{assum:E}(i) hold.
Let $R_1(\cdot)$ and $R_2(\cdot)$ be defined as in (\ref{eq:In-diff}). Then, for $\ob \in \R^d$, we have
\begin{eqnarray} \label{eq:ExR1}
\left| \Ex[R_1(\ob)] \right| &\leq& C |D_n|^{-1} |c_{h,n}(\ob)| |c_{h^2,n}(\ob)| + |D_n|^{-1/2} |c_{h,n}(\ob)| o(1) \\
\text{and} \quad \left| \Ex[R_2(\ob)] \right| &\leq& C |D_n|^{-1} |c_{h,n}(\ob)|^2, \quad n \rightarrow \infty,
\label{eq:ExR2}
\end{eqnarray} where $o(1)$ error above is uniform over $\ob \in \R^d$.

If we further assume Assumption \ref{assum:C} for $\ell=4$. Then, 
for $\ob_1, \ob_2 \in \R^d$,
\begin{equation} \label{eq:varR2}
|\cov(R_2(\ob_1), R_2(\ob_2))| \leq C|D_n|^{-2} |c_{h,n}(\ob_1)|^2 |c_{h,n}(\ob_2)|^2. 
\end{equation}
\end{theorem}
\textit{Proof}. To show (\ref{eq:ExR1}), we use Lemma \ref{lemma:cumulants}(a) and get
\begin{equation*}
\left| \Ex[R_1(\ob)] \right| = |\cum(\widehat{\lambda}_n, K_{h,n}(\ob))|
\leq C |D_n|^{-1} |c_{h,n}(\ob)| |c_{h^2,n}(\ob)| + |D_n|^{-1/2} |c_{h,n}(\ob)| o(1)
\end{equation*} as $n \rightarrow \infty$. Thus, we show (\ref{eq:ExR1}).

To show (\ref{eq:ExR2}), by Lemma \ref{lemma:cumulants}(b),
\begin{equation*}
|\Ex[R_2(\ob)]| = |c_{h,n}(\ob)|^2 \cum(\widehat{\lambda}_n, \widehat{\lambda}_n) \leq
C |D_n|^{-1}|c_{h,n}(\ob)|^{2},~~\ob\in\R^d.
\end{equation*} Thus, we show (\ref{eq:ExR2}).

To show (\ref{eq:varR2}), by Lemma \ref{lemma:cumulants}(e), we have
\begin{eqnarray*}
&& \cov(R_2(\ob_1), R_2(\ob_2)) \\
&&~~= |c_{h,n}(\ob_1)|^2 |c_{h,n}(\ob_2)|^2 \var((\widehat{\lambda}_n-\lambda)^2)
\leq C|D_n|^{-2} |c_{h,n}(\ob_1)|^2 |c_{h,n}(\ob_2)|^2,~~\ob_1, \ob_2 \in \R^d.
\end{eqnarray*} Thus, we show (\ref{eq:varR2}). All together, we prove the theorem.
\hfill $\Box$

\vspace{0.5em}

Now, let 
\begin{equation*}
S_i = |D_n|^{1/2}\int_{D}\phi(\ob) R_{i}(\ob) d\ob, \quad i \in \{1,2\},
\end{equation*} where $D$ is a compact region on $\R^d$ and $\phi$ is a symmetric continuous function on $D$.
In the following theorem, we show that $S_1$ and $S_2$ are asympototically negligible.

\begin{theorem} \label{thm:S12}
Suppose that Assumptions \ref{assum:A}, \ref{assum:C} (for $\ell=4$), and \ref{assum:E}(i) hold. Then,
\begin{equation*}
S_1, S_2  \Lcon 0, \quad n \rightarrow \infty,
\end{equation*} where $\Lcon$ denotes convergence in $L_2$.
\end{theorem}
\textit{Proof}.
We first calculate the expectations. By using (\ref{eq:ExR2}) and Lemma \ref{lemma:Cn}(d), we have
\begin{equation}
\begin{aligned}
|\Ex [S_2] | &\leq
|D_n|^{1/2} \int_{D} |\phi(\ob)| |\Ex[R_2(\ob)]| d\ob \\
&\leq C|D_n|^{-1/2} \int_{D} |\phi(\ob)| |c_{h,n}(\ob)|^2 d\ob = O(|D_n|^{-1/2}), ~~n\rightarrow \infty.
\end{aligned}
\label{eq:ES2}
\end{equation}  Similarly, by using (\ref{eq:ExR1}) and Lemma \ref{lemma:Cn}(d),(e), we have
\begin{equation}
\begin{aligned}
|\Ex [S_1] | &\leq
|D_n|^{1/2} \int_{D} |\phi(\ob)| |\Ex[R_1(\ob)]| d\ob \\
&\leq C|D_n|^{-1/2} \int_{D} |\phi(\ob)| |c_{h,n}(\ob)| |c_{h^2,n}(\ob)| d\ob  
+o(1) \int_{D} |\phi(\ob)| |c_{h,n}(\ob)|d\ob  \\
&= o(1), ~~n\rightarrow \infty.
\end{aligned}
\label{eq:ES1}
\end{equation}

Next, we calculate the variances.  By using (\ref{eq:varR2}) and Lemma \ref{lemma:Cn}(d), $\var(S_2)$ is bounded by
\begin{equation}
\begin{aligned} 
\var(S_2) &\leq |D_n| \int_{D^2} |\phi(\ob_1)| |\phi(\ob_2)| |\cov(R_2(\ob_1), R_2(\ob_2))| d\ob_1 d\ob_2 \\
 &\leq C|D_n|^{-1} \left( \int_{D} |\phi(\ob_1)| |c_{h,n}(\ob)|^2 d\ob_1\right)^2 = O(|D_n|^{-1})
, \quad n\rightarrow \infty.
\end{aligned}
\label{eq:varS2}
\end{equation}
To bound $\var(S_1)$, we need more sophisticated calculations.
By using indecomposable partitions, for $\ob_1, \ob_2 \in \R^d$, we have
\begin{eqnarray*}
&&\cov(R_1(\ob_1), R_1(\ob_2)) \\
&&=  \cum( (\widehat{\lambda}_{h,n}-\lambda) K_{h,n}(\ob_1),
(\widehat{\lambda}_{h,n}-\lambda) K_{h,n}(\ob_2)) \\
&&=\var(\widehat{\lambda}_{h,n}) \cum( K_{h,n}(\ob_1), K_{h,n}(\ob_2))
+\cum(\widehat{\lambda}_{h,n}, K_{h,n}(\ob_1)) \cum(\widehat{\lambda}_{h,n}, K_{h,n}(\ob_2))
 \\
&& ~~
+ \cum(\widehat{\lambda}_{h,n}, K_{h,n}(\ob_1), \widehat{\lambda}_{h,n}, K_{h,n}(\ob_2)).
\end{eqnarray*} Thus, we have $\var(S_1) = L_1 + L_2 + L_3$, where
\begin{eqnarray*}
L_1 &=& |D_n| \var(\widehat{\lambda}_{h,n}) \int_{D^2} \phi(\ob_1)\phi(\ob_2) 
\cum(K_{h,n}(\ob_1), K_{h,n}(\ob_2)) d\ob_1d\ob_2,\\
L_2 &=& |D_n| \left( \int_{D} \phi(\ob)  \cum(\widehat{\lambda}_{h,n}, K_{h,n}(\ob)) d\ob \right)^2,~~ \text{and} \\
L_3 &=& |D_n| \int_{D^2} \phi(\ob_1)\phi(\ob_2) 
\cum(\widehat{\lambda}_{h,n}, K_{h,n}(\ob_1), \widehat{\lambda}_{h,n}, K_{h,n}(\ob_2)) d\ob_1d\ob_2.
\end{eqnarray*}

We will bound each term above. First, since $L_2 = \Ex[S_1]^2$, we have
\begin{equation} \label{eq:L2bound}
L_2 = |E[S_1]|^2 = o(1),~~n\rightarrow \infty.
\end{equation}
By using Lemmas \ref{lemma:Cn}(e) and \ref{lemma:cumulants}(c), $L_3$ is bounded by
\begin{equation}
L_3 \leq C |D_n|^{-1} \left( \int_{D} |\phi(\ob)| |c_{h,n}(\ob)| d\ob \right)^2 
= O(|D_n|^{-1}), \quad n\rightarrow \infty.
\label{eq:L3bound}
\end{equation}
To bound $L_1$, we only focus on the $c_{h,n}(\ob_1) c_{h,n}(\ob_2)\cum(J_{h,n}(-\ob_1), J_{h,n}(-\ob_2))$ term in the expansion of 
$\cum(K_{h,n}(\ob_1), K_{h,n}(\ob_2))$ and other three terms are treated similarly. By using Lemma \ref{lemma:cumulants}(b)
and Theorem \ref{thm:cov-exp}, (a part of) $L_1$ is bounded by
\begin{equation*}
C |D_n|^{-1} \int_{D^2} \phi(\ob_1)\phi(\ob_2) c_{h,n}(\ob_1) c_{h,n}(\ob_2) 
\left( H_{h,2}^{(n)}(-\ob_1+\ob_2) + o(1) \right)
d\ob_1 d\ob_2, \quad n \rightarrow \infty,
\end{equation*} where $o(1)$ error is uniform over $\ob_1, \ob_2 \in \R^d$.
By using Lemma \ref{lemma:Cn}(e), the second term above is $o(|D_n|^{-1})$ as $n\rightarrow \infty$. Moreover, the first term is bounded by
\begin{eqnarray*}
&& C |D_n|^{-1/2} \int_{D^2} |\phi(\ob_1)\phi(\ob_2) c_{h,n}(\ob_1) c_{h,n}(\ob_2) | |c_{h^2,n}(-\ob_1 + \ob_2)| d\ob_1 d\ob_2 \\
&& \quad  = C |D_n|^{-1/2} \int_{D} d \ob_1 |\phi(\ob_1) c_{h,n}(\ob_1)| \int_{D} |\phi(\ob_2)c_{h,n}(\ob_2) c_{h^2,n}(-\ob_1 + \ob_2)| d\ob_2 \\
&& \quad  \leq C |D_n|^{-1/2} \int_{D} d \ob_1 |\phi(\ob_1) c_{h,n}(\ob_1)| = O(|D_n|^{-1/2}), \quad n\rightarrow \infty.
\end{eqnarray*}
Here, the inequality is due to Lemma \ref{lemma:Cn}(d) and the second identity is due to Lemma \ref{lemma:Cn}(e).
All together, we conclude that
\begin{equation} \label{eq:L1bound}
L_1 = o(1), \qquad n\rightarrow \infty.
\end{equation} Combining (\ref{eq:L2bound}), (\ref{eq:L3bound}), and (\ref{eq:L1bound}), we conclude
\begin{equation}\label{eq:varS1}
\var(S_1) \leq L_1 +L_2+L_3 = o(1), \qquad n\rightarrow \infty.
\end{equation}

Combining (\ref{eq:ES1}) and (\ref{eq:varS1}), we have $S_1 \Lcon 0$ as $n \rightarrow \infty$ and 
 (\ref{eq:ES2}) and (\ref{eq:varS2}) yield $S_2 \Lcon 0$ as $n \rightarrow \infty$. Thus, we get the desired results.
\hfill $\Box$

\vspace{0.5em}

Lastly, recall the theoretical counterpart of the kernel spectral density estimator $\widetilde{f}_{n,b}(\ob)$ in (\ref{eq:KSD-t}).
As a consequence of the above theorem, we obtain the probabilistic bound for the difference $\widetilde{f}_{n,b}(\ob)-\widehat{f}_{n,b}(\ob)$.

\begin{corollary} \label{coro:KDE-asym}
Suppose that Assumptions \ref{assum:A}, \ref{assum:C} (for $\ell=4$), and \ref{assum:E}(i) hold. 
Moreover, the bandwidth $b = b(n)$ is such that $\lim_{n\rightarrow \infty} b(n) = 0$, then
\begin{equation} \label{eq:KSD-equiv2}
\sqrt{|D_n| b^d}(\widetilde{f}_{n,b}(\ob)  - \widehat{f}_{n,b}(\ob)) \Lcon 0, \quad n \rightarrow \infty.
\end{equation}
\end{corollary}
\textit{Proof}. Let $\phi_b(\xb) = W_b(\ob-\xb) = b^{-d} W(b^{-1}(\ob-\xb))$, $\xb \in \R^d$. Then,
$\sqrt{|D_n|b^d} (\widehat{f}_{n,b}(\ob) - \widetilde{f}_{n,b}(\ob))= Q_1 + Q_2$, where 
\begin{equation*}
Q_i = \sqrt{|D_n| b^d} \int_{\R^d} \phi_b(\xb) R_i(\xb) d\xb, \quad i \in \{1,2\}.
\end{equation*} By simple calculation, we have
\begin{equation*}
\int_{\R^d} \phi_b(\xb) d\xb = \int_{\R^d} W(\xb) d\xb = 1 \quad \text{and} \quad
\int_{\R^d} \phi_b(\xb)^2 d\xb = b^{-d} \int_{\R^d} W(\xb)^2 d\xb = C b^{-d}. 
\end{equation*} Therefore, by using similar techniques to bound the first- and second-order moments of $S_2$ in Theorem \ref{thm:S12}, we have
\begin{eqnarray*}
\Ex[Q_2] =O(|D_n|^{-1/2}b^{d/2})=o(1) \quad \text{and} \quad 
\var(Q_2) = O(|D_n|^{-1}b^{d}) = o(1),~~n\rightarrow \infty.
\end{eqnarray*} To bound the expectation of $Q_1$, by using a similar argument as in (\ref{eq:ES1}), we have
\begin{eqnarray*}
\Ex[Q_1] &\leq& C |D_n|^{1/2} b^{d/2} \left( |D_n|^{-1}O(1)+ |D_n|^{-1/2}o(1) \int_{\R^d} |\phi_b(\xb)| |c_{h,n}(\ob)|d\xb \right) \\
&\leq& C |D_n|^{-1/2}b^{d/2}+ o(1) = o(1), \quad n \rightarrow \infty.
\end{eqnarray*} Here, we use a modification of Lemma \ref{lemma:Cn}(e) 
\begin{equation*}
b^{d/2}\int_{\R^d} |\phi_b(\xb)| |c_{h,n}(\ob)|d\xb \leq C \left( b^{d} \int \phi_b(\xb)^{2} d\xb \right)^{1/2} = O(1), \quad n \rightarrow \infty
\end{equation*}
in the second inequality above. Similarly, by using an expansion of $\var(S_1)$ in the proof of Theorem \ref{thm:S12}, one can easily seen that $\var(Q_1) = o(1)$ as $n\rightarrow \infty$. Therefore, both $Q_1$ and $Q_2$ converges to zero in $L_2$ as $n \rightarrow \infty$. Thus, we prove the theorem.
\hfill $\Box$

\section{Verification of the $\alpha$-mixing CLT for the integrated periodogram} \label{sec:mixing}

In this section, we will provide greater details the CLT for $\widetilde{G}_{h,n}(\phi)$ defined as in (\ref{eq:Gtilde}).
With loss of generality, we assume that $|D_n|$ grows proportional to $n^{d}$ as $n \rightarrow \infty$. Thus, $A_1, \cdots, A_d$ increases with order $n$. Next, let $\beta, \gamma \in (0,1)$ be chosen such that $2d /\varepsilon<\beta <\gamma < 1$, where $\varepsilon>2d$ is from Assumption \ref{assum:D}(ii).. For $n \in \N$, let
\begin{equation*}
A_n = \{\kb: \kb \in n^{\gamma}\Z^d~~\text{and}~~ D_n^{(\kb)} = \kb + [-(n^{\gamma}-n^{\beta})/2, (n^{\gamma}-n^{\beta})/2]^d \subset D_n\}.
\end{equation*}  Thus, $\widetilde{D}_n = \bigcup_{\kb \in A_n}D_n^{(\kb)}$ is a disjoint union that is included in $D_n$. Let $k_n = |A_n|$ and let $E_n = D_n \backslash \widetilde{D}_n$. Then, it can be easily seen that
\begin{equation} \label{eq:sub-block-vol}
\lim_{n\rightarrow \infty} \frac{|\widetilde{D}_n|}{|D_n|}
=\lim_{n\rightarrow \infty} \left( 1- \frac{|E_n|}{|D_n|}\right) = 1.
\end{equation}

\subsection{Linearization of the integrated periodogram}  \label{sec:mixing11}

Now, decompose the DFT using sub-blocks. Let
\begin{eqnarray*}
\mathcal{J}_{h,n}^{1}(\ob) &=& 
(2\pi)^{-d/2} H_{h,2}^{-1/2}  |\widetilde{D}_n|^{-1/2} \sum_{\xb \in X \cap \widetilde{D}_n}
h(\xb / \aB) \exp(-i \xb^\top \ob) \\
\text{and} \quad 
\mathcal{J}_{h,n}^{2}(\ob) &=& 
(2\pi)^{-d/2} H_{h,2}^{-1/2}  |E_n|^{-1/2} \sum_{\xb \in X \cap E_n}
h(\xb / \aB) \exp(-i \xb^\top \ob).
\end{eqnarray*} 
Let $J_{h,n}^{i}(\ob) = \mathcal{J}_{h,n}^{i}(\ob) - \Ex[\mathcal{J}_{h,n}^{i}(\ob)]$, $i \in \{1,2\}$, be the centered DFTs.
Then, since $D_n$ is a disjoint union of $\widetilde{D}_n$ and $E_n$, we have
\begin{equation} \label{eq:Job-12}
J_{h,n}(\ob) = \frac{|\widetilde{D}_n|^{1/2}}{|D_n|^{1/2}} J_{h,n}^{1}(\ob) + 
 \frac{|E_n|^{1/2}}{|D_n|^{1/2}} J_{h,n}^{2}(\ob), \quad n\in \N.
\end{equation} Furthermore, since $\widetilde{D}_n = \bigcup_{\kb \in A_n}D_n^{(\kb)}$ is a disjoint union and $|\widetilde{D}_n| = k_n |D_n^{(\kb)}|$.
 $J_{h,n}^{1}(\ob)$ can be written as
\begin{equation} \label{eq:J1-decomp} 
J_{h,n}^{1}(\ob) = k_n^{-1/2} \sum_{\kb \in A_n} J_{h,n}^{(\kb)}(\ob)
= k_n^{-1/2} \sum_{\kb \in A_n} \left( \mathcal{J}_{h,n}^{(\kb)}(\ob) - \Ex[\mathcal{J}_{h,n}^{(\kb)}(\ob)] \right),
\end{equation}
where
\begin{equation} \label{eq:J-kb}
\mathcal{J}_{h,n}^{(\kb)}(\ob) = 
(2\pi)^{-d/2} H_{h,2}^{-1/2} | D_n^{(\kb)}|^{-1/2} \sum_{\xb \in X \cap D_n^{(\kb)}}
h(\xb / \aB) \exp(-i \xb^\top \ob), \quad \kb \in A_n.
\end{equation} 

Recall $\widetilde{G}_{h,n}(\phi)$ and $V_{h,n}(\phi)$ from (\ref{eq:Gtilde}) and (\ref{eq:Vhn}), respectively. Our aim in this section is to show that $\widetilde{G}_{h,n}(\phi)-V_{h,n}(\phi)$ is asymptotically negligible. To do so, we introduce an intermediate statistic. For $n\in \N$, let
\begin{equation} \label{eq:Shn}
S_{h,n}(\phi) = |\widetilde{D}_n|^{1/2} \int_{D} \phi(\ob) \left( 
|J^{1}_{h,n}(\ob)|^2 - \Ex[|J^{1}_{h,n}(\ob)|^2]
\right) d\ob.
\end{equation} In the next two theorems, we will show that $\widetilde{G}_{h,n}(\phi) - S_{h,n}(\phi)$ and $S_{h,n}(\phi) - V_{h,n}(\phi)$ are asymptotically negligible.

\begin{theorem} \label{thm:Shn}
Suppose that Assumptions \ref{assum:A}, \ref{assum:C} (for $\ell=4$), and \ref{assum:B} hold. Suppose further the data taper $h$ is either constant on $[-1/2,1/2]^d$ or satisfies Assumption \ref{assum:E}(for $m=d+1$). Then,
\begin{equation*}
\widetilde{G}_{h,n}(\phi) - S_{h,n}(\phi) = o_p(1), \quad n \rightarrow \infty.
\end{equation*}
\end{theorem}
\textit{Proof}. Since both $\widetilde{G}_{h,n}(\phi)$ and $S_{h,n}(\phi)$ are centered, we will only require to show that \\ $\lim_{n \rightarrow \infty} \var( \widetilde{G}_{h,n}(\phi) - S_{h,n}(\phi)) = 0$. Let $
\widetilde{J}_{h,n}^{1}(\ob) = 
\frac{|\widetilde{D}_n|^{1/2}}{|D_n|^{1/2}} J_{h,n}^{1}(\ob)$ and $
\widetilde{J}_{h,n}^{2}(\ob) = 
\frac{|E_n|^{1/2}}{|D_n|^{1/2}} J_{h,n}^{2}(\ob)$. Then, we have
\begin{eqnarray}
&&\var \left(\widetilde{G}_{h,n}(\phi) - \frac{|\widetilde{D}_n|^{1/2}}{|D_n|^{1/2}} S_{h,n}(\phi)\right)=
|D_n| \var \left( \int_{D} \phi(\ob) \left( |J_{h,n}(\ob)|^2 - |\widetilde{J}_{h,n}^{1}(\ob)|^2\right) d\ob
\right)  \nonumber \\
&&~=
|D_n| \int_{D^2} \phi(\ob_1)\phi(\ob_2) \cov \left( |J_{h,n}(\ob_1)|^2 -|\widetilde{J}_{h,n}^{1}(\ob_1)|^2, 
|J_{h,n}(\ob_2)|^2 -|\widetilde{J}_{h,n}^{1}(\ob_2)|^2 \right) d\ob_1 d\ob_2. \nonumber \\
\label{eq:var-term00}
\end{eqnarray}
From (\ref{eq:Job-12}), we have
\begin{eqnarray*}
|J_{h,n}(\ob)|^2 -|\widetilde{J}_{h,n}^{1}(\ob)|^2 
&=& J_{h,n}(\ob) \left( J_{h,n}(-\ob) - \widetilde{J}_{h,n}^{1}(-\ob) \right)+
\left(J_{h,n}(\ob) - \widetilde{J}_{h,n}^{1}(\ob)\right) \widetilde{J}_{h,n}^{1}(-\ob) \\
&=&J_{h,n}(\ob)\widetilde{J}_{h,n}^{2}(-\ob) + \widetilde{J}_{h,n}^{2}(\ob) \widetilde{J}_{h,n}^{1}(-\ob), \quad \ob \in \R^d.
\end{eqnarray*}
Therefore, the variance term (\ref{eq:var-term00}) can be decomposed into four terms. We will only focus on the term
\begin{equation} \label{eq:var-term11}
|D_n| \int_{D^2} \phi(\ob_1)\phi(\ob_2) \cov \left( J_{h,n}(\ob_1)\widetilde{J}_{h,n}^{2}(-\ob_1), 
J_{h,n}(\ob_2)\widetilde{J}_{h,n}^{2}(-\ob_2) \right) d\ob_1 d\ob_2
\end{equation} and other three terms can be treated similarly. Using indecomposable partition, the above term is $B_1 + B_2 +B_3$, where
\begin{eqnarray*}
B_1 &=& |D_n| \int_{D^2} \phi(\ob_1)\phi(\ob_2) \cum( J_{h,n}(\ob_1), J_{h,n}(-\ob_2)) \cum( \widetilde{J}^2_{h,n}(-\ob_1), \widetilde{J}^2_{h,n}(\ob_2))
d\ob_1 d\ob_2, \\
B_2 &=& |D_n| \int_{D^2} \phi(\ob_1)\phi(\ob_2) \cum( J_{h,n}(\ob_1), \widetilde{J}^2_{h,n}(\ob_2))
\cum( \widetilde{J}^2_{h,n}(-\ob_1), J_{h,n}(-\ob_2))
d\ob_1 d\ob_2, \\
B_3 &=& |D_n| \int_{D^2} \phi(\ob_1)\phi(\ob_2) \cum( J_{h,n}(\ob_1), J_{h,n}(-\ob_2),
\widetilde{J}^2_{h,n}(-\ob_1), \widetilde{J}^2_{h,n}(\ob_2))
d\ob_1 d\ob_2.
\end{eqnarray*}
By using a similar techniques to show Lemma \ref{lemma:k4-bound}, we have
\begin{equation*}
\left| \cum( J_{h,n}(\ob_1), J_{h,n}(-\ob_2),
\widetilde{J}^2_{h,n}(-\ob_1), \widetilde{J}^2_{h,n}(\ob_2)) \right| \leq C\frac{|E_n|}{|D_n|^2}, \quad
\ob_1, \ob_2 \in \R^d.
\end{equation*} Therefore,
\begin{equation*}
|B_3| \leq C \frac{|E_n|}{|D_n|} \left( \int_{D} |\phi(\ob)| \right)^2 = o(1), \quad n \rightarrow \infty.
\end{equation*} Here, the second identity is due to (\ref{eq:sub-block-vol}). Now, let 
\begin{equation*}
h_{E_n}(\xb/\aB) = h(\xb/\aB) I(\xb \in E_n), \quad n \in \N, \quad \xb \in D_n
\end{equation*} be a truncated taper function. Then, we have
\begin{equation*}
\frac{|E_n|^{1/2}}{|D_n|^{1/2}}\mathcal{J}^2_{h,n}(\ob) = (2\pi)^{-d/2} H_{h,2}^{-1/2}  |D_n|^{-1/2} \sum_{\xb \in X \cap D_n}
h_{E_n}(\xb / \aB) \exp(-i \xb^\top \ob), \quad \ob \in \R^d.
\end{equation*} Therefore, $\widetilde{J}^2_{h,n}(\ob)$ can be viewed as a (centered) DFT on $D_n$ with taper function $h_{E_n}$. Thus, from the modification of the results of Theorem \ref{thm:A1}(ii), we have
\begin{equation*}
|B_1| \leq C |D_n|^{-1} \int_{D_n} h(\xb/\aB)^2  h_{E_n}(\xb/\aB)^2 d\xb 
= C  |D_n|^{-1} \int_{D_n} h(\xb/\aB)^4  I(\xb \in E_n) d\xb \leq C^\prime \frac{|E_n|}{|D_n|}.
\end{equation*} Therefore, by using (\ref{eq:sub-block-vol}), we conclude $B_1 = o(1)$ as $n\rightarrow \infty$. Similarly, we can show $B_2 = o(1)$ as $n\rightarrow \infty$. All together, the term in (\ref{eq:var-term11}) limits to zero as $n\rightarrow \infty$. By using similar technique, we can bound other three terms in the decomposition in (\ref{eq:var-term00}). Therefore, we have
\begin{equation*}
\var \left(\widetilde{G}_{h,n}(\phi) - \frac{|\widetilde{D}_n|^{1/2}}{|D_n|^{1/2}} S_{h,n}(\phi)\right)
= o(1), \quad n \rightarrow \infty.
\end{equation*}
Lastly, since $\var(\widetilde{G}_{h,n}(\phi))$ is finite, so does $\var(S_{h,n}(\phi))$. Therefore, by using (\ref{eq:sub-block-vol}), we show $\lim_{n \rightarrow \infty} \var( \widetilde{G}_{h,n}(\phi) - S_{h,n}(\phi)) = 0$. Thus, we get the desired results.
\hfill $\Box$

\vspace{0.5em}

Now, we study the asymptotic equivalence of $S_{h,n}(\phi) - V_{h,n}(\phi)$. Recall (\ref{eq:J1-decomp}) and (\ref{eq:Shn}). We have
\begin{eqnarray*}
S_{h,n}(\phi) &=& |\widetilde{D}_n|^{1/2} \int_{D} \phi(\ob) \left(
|J^{1}_{h,n}(\ob)|^2 - \Ex[|J^{1}_{h,n}(\ob)|^2]
\right) d\ob \\
&=& k_n^{-1/2} \sum_{\jb, \kb \in A_n} |D_n^{(\kb)}|^{1/2} \int_{D} \phi(\ob) \left( 
J^{(\jb)}_{h,n}(\ob) J^{(\kb)}_{h,n}(-\ob) - \Ex[J^{(\jb)}_{h,n}(\ob) J^{(\kb)}_{h,n}(-\ob)]
\right) d\ob \\
&=& V_{h,n}(\phi) + T_{h,n}(\phi),
\end{eqnarray*} where
\begin{equation}\label{eq:Thn}
T_{h,n}(\phi)
= k_n^{-1/2} \sum_{\jb \neq \kb \in A_n} |D_n^{(\kb)}|^{1/2} \int_{D} \phi(\ob) \left(
J^{(\jb)}_{h,n}(\ob) J^{(\kb)}_{h,n}(-\ob) - \Ex[J^{(\jb)}_{h,n}(\ob) J^{(\kb)}_{h,n}(-\ob)]
\right) d\ob.
\end{equation}
To show that $T_{h,n}(\phi)$ is asymptotically negligible, we require the covariance inequality below. For $\jb, \kb \in n^{\gamma}\Z$, let
\begin{equation*}
k(\jb, \kb) = \|\jb - \kb\|_{\infty} - n^{\gamma} + n^{\beta} \in \{n^{\beta}, n^{\beta} + n^{\gamma}, n^{\beta} + 2 n^{\gamma}, \dots \}.
\end{equation*}

\begin{lemma} \label{eq:mixing-cov-ineq}
Suppose that Assumption \ref{assum:C} for $\ell=8$ holds. Let $\jb, \kb, \ellb \in A_n$ be the pairwise disjoint points and let $w_n = |D_n^{(\kb)}|$ be the common volume of the sub-blocks. Then, for $\ob_1, \ob_2 \in \R^d$,
\begin{eqnarray}
&& \left| \cov (J^{(\jb)}_{h,n}(\ob_1) J^{(\jb)}_{h,n}(-\ob_1), J^{(\kb)}_{h,n}(\ob_2) J^{(\ellb)}_{h,n}(-\ob_2)) \right| 
\leq C  \alpha_{w_n, 2w_n} \left(k(\jb, \kb) \wedge k(\jb, \ellb)\right)^{1/2}, \label{eq:Jmixing1} \\
&& \left| \cov (J^{(\jb)}_{h,n}(\ob_1) J^{(\jb)}_{h,n}(-\ob_1), J^{(\jb)}_{h,n}(\ob_2) J^{(\kb)}_{h,n}(-\ob_2)) \right|  
 \leq C \alpha_{w_n, w_n} \left( k(\jb, \kb)\right)^{1/2},
\label{eq:Jmixing2}
\end{eqnarray}
where $\alpha_{p,q} (\cdot)$ is the $\alpha$-mixing coefficient defined as in (\ref{eq:alpha}).
\end{lemma}
\textit{Proof}.
We first show (\ref{eq:Jmixing1}). 
For the brevity, we write $J^{(\jb)}_{h,n}(\ob_1) = J^{(\jb)}(\ob_1), \cdots$.
Note that $D_n^{(\jb)}$ and $D_n^{(\kb)} \cup D_n^{(\ellb)}$ are disjoint and $d(D_n^{(\jb)}, D_n^{(\kb)} \cup D_n^{(\ellb)}) = k(\jb, \kb) \wedge k(\jb, \ellb)$. Therefore, by using well-known covariance inequality (cf. \cite{b:dou-94}, Theorem 3(1)), we have
\begin{eqnarray*}
&& \left| \cov (J^{(\jb)}(\ob_1) J^{(\jb)}(-\ob_1) , J^{(\kb)}(\ob_2) J^{(\ellb)}(-\ob_2)) \right|  \\
&& ~~\leq 8  \alpha_{w_n, 2w_n} \left(k(\jb, \kb) \wedge k(\jb, \ellb)\right)^{1/2} \|J^{(\jb)}(\ob_1) J^{(\jb)}(-\ob_1)\|_{\Ex,4}
\|J^{(\kb)}(\ob_2) J^{(\ellb)}(-\ob_2))\|_{\Ex,4},
\end{eqnarray*} where $\|X\|_{\Ex,4} = \{\Ex|X|^4\}^{1/4}$. Then, by using $\|X\|_{\Ex,4}^{4} = \kappa_4(X) + 3 \var(X)^2$ and Lemma \ref{lemma:k4-bound}, we show $\|J^{(\jb)}(\ob) J^{(\jb)}(-\ob)\|_{\Ex,4} = O(1)$ as $n \rightarrow \infty$, provided Assumption \ref{assum:C} for $\ell=8$ holds. Substitute this to above, we show (\ref{eq:Jmixing1}).

To show (\ref{eq:Jmixing2}), we first note that for a centered random variable $X,Y,Z,W$, $\cov(XY,ZW) = \cov(XY\overline{Z},W) - \cov(XY) \overline{\cov(ZW)}$. Apply this identity to the left hand side of (\ref{eq:Jmixing2}), we get
\begin{eqnarray*}
&& \left| \cov (J^{(\jb)}(\ob_1) J^{(\jb)}(-\ob_1), J^{(\jb)}(\ob_2) J^{(\kb)}(-\ob_2)) \right| 
\leq \left| \cov (J^{(\jb)}(\ob_1) J^{(\jb)}(-\ob_1) J^{(\jb)}_{}(-\ob_2), J^{(\kb)}_{}(-\ob_2))\right|
 \nonumber \\
&& +
\left| \cov (J^{(\jb)}(\ob_1) J^{(\jb)}(-\ob_1))\right| \cdot \left|\cov( J^{(\jb)}_{}(-\ob_2), J^{(\kb)}_{}(\ob_2))\right|.
\end{eqnarray*}
By using similar arguments above, the second term above is bounded by $C \alpha_{w_n, w_n} \left(k(\jb, \kb)\right)^{1/2}$. To bound for the first term, we first note that $J^{(\jb)}(\ob_1) J^{(\jb)}(-\ob_1) = I^{(\jb)}(\ob) \in \R$ and by using H\"{o}lder's inequality, we have
\begin{eqnarray*}
\Ex | I^{(\jb)}(\ob_1) J^{(\jb)}(-\ob_2)|^{8/3} &=&
\Ex [ I^{(\jb)}(\ob_1)^{8/3} I^{(\jb)}(\ob_2) |J^{(\jb)}(-\ob_2)|^{2/3}] \\
&\leq& \|  I^{(\jb)}(\ob_2) |J^{(\jb)}(-\ob_2)|^{2/3}\|_{\Ex, 3} \|  I^{(\jb)}(\ob_1)^{8/3}\|_{\Ex, 3/2} \\
&\leq& \| I^{(\jb)}(\ob_2)^4\|_{\Ex,1}^{1/3} \| I^{(\jb)}(\ob_1)^4\|_{\Ex,1}^{2/3} = O(1), \quad n\rightarrow \infty.
\end{eqnarray*}
Therefore, by using covariance inequality and the H\"{o}lder's inequality, we have
\begin{eqnarray*}
&& \left| \cov (J^{(\jb)}(\ob_1) J^{(\jb)}(-\ob_1) J^{(\jb)}_{}(-\ob_2), J^{(\kb)}_{}(-\ob_2))\right| \\
&&~~ \leq 8  \alpha_{w_n, w_n} \left( k(\jb, \kb)\right)^{1/2} \left\| 
I^{(\jb)}(\ob_1) J^{(\jb)}_{}(-\ob_2)
\right\|_{\Ex,8/3} \| J^{(\kb)}_{}(-\ob_2)\|_{\Ex,8} \leq
C\alpha_{w_n, w_n} \left( k(\jb, \kb)\right)^{1/2}.
\end{eqnarray*}
Thus, we show (\ref{eq:Jmixing2}). 
All together, we get the desired results.
\hfill $\Box$

\vspace{0.5em}

Now, we are ready to prove the theorem below.

\begin{theorem} \label{thm:Thn}
Suppose that Assumptions \ref{assum:A}, \ref{assum:C} (for $\ell=8$), \ref{assum:D}(ii), \ref{assum:B}, and \ref{assum:F} hold. Suppose further the data taper $h$ is either constant on $[-1/2,1/2]^d$ or satisfies Assumption \ref{assum:E}(for $m=d+1$). Then,
\begin{equation*}
S_{h,n}(\phi) - V_{h,n}(\phi) = T_{h,n}(\phi) = o_p(1), \quad n \rightarrow \infty.
\end{equation*}
\end{theorem}
\textit{Proof}. By using a similar argument as in the proof of Theorem \ref{thm:Shn}, it is enough to show $\lim_{n\rightarrow \infty} \var(T_{h,n}(\phi)) =0$. Note
\begin{equation}  \label{eq:var-bound00}
 \var(T_{h,n}(\phi)) = \left(  \var(S_{h,n}(\phi)) -  \var(V_{h,n}(\phi)) \right) - 2 \cov (T_{h,n}(\phi), V_{h,n}(\phi)).
\end{equation} We will bound each term above. To bound the first term, by using Theorems \ref{thm:A1}(ii), \ref{thm:Shn}, and \ref{thm:var-Vhn} (below), we have
\begin{equation} \label{eq:var-bound11}
\lim_{n\rightarrow \infty} \var(S_{h,n}(\phi))  = \lim_{n\rightarrow \infty} \var(V_{h,n}(\phi))
= (2\pi)^{d} (H_{h,4} / H_{h,2}^2) ( \Omega_1 + \Omega_2),
\end{equation} where $\Omega_1$ and $\Omega_2$ are defined as in (\ref{eq:Omegas}). Therefore, the first term in (\ref{eq:var-bound00}) is $o(1)$ as $n\rightarrow \infty$.
To bound the second term, we will use a $\alpha$-mixing condition. For $q \in \N$, let $(A_n)^{q,\neq}$ be a set of $q$ disjoints points in $A_n$. Then, 
\begin{eqnarray*}
&& \cov (T_{h,n}(\phi), V_{h,n}(\phi)) = k_n^{-1} \sum_{\jb \in A_n} \sum_{(\kb,\ellb) \in (A_n)^{2,\neq}} \cov (\widetilde{G}_{h,n}^{(\jb,\jb)}(\phi), \widetilde{G}_{h,n}^{(\kb,\ellb)}(\phi)) \\
&& =
k_n^{-1}\sum_{(\jb, \kb,\ellb) \in (A_n)^{3,\neq}} \cov (\widetilde{G}_{h,n}^{(\jb,\jb)}(\phi), \widetilde{G}_{h,n}^{(\kb,\ellb)}(\phi)) 
+ k_n^{-1}\sum_{(\jb, \kb) \in (A_n)^{2,\neq}} \cov (\widetilde{G}_{h,n}^{(\jb,\jb)}(\phi), \widetilde{G}_{h,n}^{(\kb,\jb)}(\phi)) \\
&&~~ +k_n^{-1}\sum_{(\jb, \kb) \in (A_n)^{2,\neq}} \cov (\widetilde{G}_{h,n}^{(\jb,\jb)}(\phi), \widetilde{G}_{h,n}^{(\jb,\kb)}(\phi)),
\end{eqnarray*} where
\begin{equation*} 
\widetilde{G}_{h,n}^{(\kb,\ellb)}(\phi) =
|D_n^{(\kb)}|^{1/2} \int_{D} \phi(\ob) \left( 
J^{(\kb)}_{h,n}(\ob) J^{(\ellb)}_{h,n}(-\ob) - \Ex[J^{(\kb)}_{h,n}(\ob) J^{(\ellb)}_{h,n}(-\ob)]
\right) d\ob, \quad  \kb, \ellb \in A_n.
\end{equation*} We will bound each term in the expansion of $\cov (T_{h,n}(\phi), V_{h,n}(\phi))$. 
By using (\ref{eq:Jmixing1}), the first term is bounded by
\begin{eqnarray*}
&& k_n^{-1} \left| \sum_{(\jb, \kb,\ellb) \in (A_n)^{3,\neq}} \cov (\widetilde{G}_{h,n}^{(\jb,\jb)}(\phi), \widetilde{G}_{h,n}^{(\kb,\ellb)}(\phi)) \right| \\
&&~~\leq  k_n^{-1} w_n \sum_{(\jb, \kb,\ellb) \in (A_n)^{3,\neq}} \int_{D^2} |\phi(\ob_1) \phi(\ob_2)|
\left| \cov (J^{(\jb)}_{h,n}(\ob_1) J^{(\jb)}_{h,n}(-\ob_1), J^{(\kb)}_{h,n}(\ob_2) J^{(\ellb)}_{h,n}(-\ob_2)) \right| d\ob_1 d\ob_2 \\
&&~~\leq C k_n^{-1} w_n \sum_{(\jb, \kb,\ellb) \in (A_n)^{3,\neq}}\alpha_{w_n, 2w_n} \left(k(\jb, \kb) \wedge k(\jb, \ellb)\right)^{1/2}.
\end{eqnarray*} Let $n^{-\gamma} \min (\|\jb-\kb\|_\infty, \|\jb-\ellb\|_\infty) = m \in \{1,2, \dots\}$. Then, for fixed $\jb \in A_n$ and $m \in \N$, the number of disjoint pairs $(\kb, \ellb) \in (A_n)^{2,\neq}$ that satisfies $n^{-\gamma} \min (\|\jb-\kb\|_\infty, \|\jb-\ellb\|_\infty) = m$ is upper bounded by $Ck_n m^{d-1}$ for some constant $C > 0$. Therefore, by using Assumption \ref{assum:D}(ii) the right hand side above is bounded by
\begin{eqnarray*}
&& C_1 k_n^{-1} w_n \sum_{(\jb, \kb,\ellb) \in (A_n)^{3,\neq}}\alpha_{w_n, 2w_n} \left(k(\jb, \kb) \wedge k(\jb, \ellb)\right)^{1/2} \\
&&~~\leq C_2  w_n^{} \sum_{\jb \in A_n} \sum_{m=1}^{\infty} m^{d-1} \cdot  w_n^{1/2} (n^{\gamma}(m-1) + n^{\beta})^{-(d+\varepsilon)/2} \\
&&~~\leq C_3 k_n w_n^{3/2} \left( n^{-\beta(d+\varepsilon)/2} + \sum_{m=1}^{\infty} m^{d-1} (n^{\gamma} m)^{-(d+\varepsilon)/2}
\right) \\
&&~~ \leq C_4 n^{d}n^{\gamma d/2} \left( n^{-\beta (d+\varepsilon)/2} + n^{-\gamma (d+\varepsilon)/2} \right)= o(1), \quad n\rightarrow \infty.
\end{eqnarray*} 
Here, we use Assumption \ref{assum:D}(ii) on the first inequality,
$(m+1)^{d-1} \leq 2 m^{d-1}$ and $n^{\gamma}(m-1)+n^{\beta} > n^{\gamma}(m-1)$ on the second inequality, 
$k_n w_n \leq |D_n| \leq C n^{d}$, $w_n^{1/2} \leq n^{\gamma d/2}$, and $\sum_{m=1}^{\infty} m^{d-1-(d+\varepsilon)/2} < \infty$ on the third inequality,
and $\beta, \gamma > 2d/\varepsilon$ on the identity. Thus, we have
\begin{equation} \label{eq:G-mixing-bound1}
k_n^{-1}\sum_{(\jb, \kb,\ellb) \in (A_n)^{3,\neq}} \cov (\widetilde{G}_{h,n}^{(\jb,\jb)}(\phi), \widetilde{G}_{h,n}^{(\kb,\ellb)}(\phi))
= o(1), \quad n\rightarrow \infty.
\end{equation} 
To bound the second term, we use (\ref{eq:Jmixing2}) and Assumption \ref{assum:D}(ii) and get
\begin{eqnarray*}
&& k_n^{-1} \left| \sum_{(\jb, \kb) \in (A_n)^{2,\neq}} \cov (\widetilde{G}_{h,n}^{(\jb,\jb)}(\phi), \widetilde{G}_{h,n}^{(\kb,\jb)}(\phi))\right|
\leq C_1 k_n^{-1} w_n \sum_{(\jb, \kb) \in (A_n)^{2,\neq}} 
\alpha_{w_n, w_n} \left(k(\jb, \kb)\right)^{1/2} \\
&& \quad \leq C_2 w_n^{3/2} \sum_{m=1}^{\infty} m^{d-1} (n^{\gamma}(m-1) + n^{\beta})^{-(d+\varepsilon)/2}
\leq C_3 n^{3\gamma d / 2} n^{-\gamma(d+\varepsilon)/2} = 
= o(1), \quad n\rightarrow \infty.
\end{eqnarray*} Similarly, the third term is bounded by
\begin{equation*}
k_n^{-1}\sum_{(\jb, \kb) \in (A_n)^{2,\neq}} \cov (\widetilde{G}_{h,n}^{(\jb,\jb)}(\phi), \widetilde{G}_{h,n}^{(\jb,\kb)}(\phi))
= o(1), \quad n\rightarrow \infty.
\end{equation*} Substitute these results into the expansion of $\cov (T_{h,n}(\phi), V_{h,n}(\phi))$, we have
\begin{equation} \label{eq:var-bound22}
\lim_{n\rightarrow \infty} \cov (T_{h,n}(\phi), V_{h,n}(\phi)) = 0.
\end{equation} Substituting (\ref{eq:var-bound11}) and (\ref{eq:var-bound22}) into (\ref{eq:var-bound00}), we show $\lim_{n\rightarrow \infty} \var(T_{h,n}(\phi)) = 0$ as $n \rightarrow \infty$. Thus, we prove the theorem.
\hfill $\Box$

Below theorem is immediately follows from Theorems \ref{thm:Shn} and \ref{thm:Thn}.

\begin{theorem} \label{thm:Vhn}
Suppose that Assumptions \ref{assum:A}, \ref{assum:C} (for $\ell=8$), \ref{assum:D}(ii), \ref{assum:B}, and \ref{assum:F} hold. Suppose further the data taper $h$ is either constant on $[-1/2,1/2]^d$ or satisfies Assumption \ref{assum:E}(for $m=d+1$). Then,
\begin{equation*}
\widetilde{G}_{h,n}(\phi) - V_{h,n}(\phi)  = o_p(1), \quad n \rightarrow \infty.
\end{equation*}
\end{theorem}

\subsection{CLT for $V_{h,n}(\phi)$}  \label{sec:mixing22}

Recall $V_{h,n}(\phi) = k_n^{-1/2} \sum_{\kb \in A_n} \widetilde{G}_{h,n}^{(\kb)}(\phi)$, where
\begin{equation*} 
\widetilde{G}_{h,n}^{(\kb)}(\phi) 
=| D_n^{(\kb)}|^{1/2} \int_{D} \phi(\ob) \left( 
|J_{h,n}^{(\kb)}(\ob)|^2 - \Ex[ |J_{h,n}^{(\kb)}(\ob)|^2]
\right) d\ob, \quad n \in \N, \quad \kb \in A_n.
\end{equation*}

In this section, we prove the CLT for $V_{h,n}(\phi)$. First, we calculate the asymptotic variance of $V_{h,n}(\phi)$.

\begin{theorem} \label{thm:var-Vhn}
Suppose that Assumptions \ref{assum:A}, \ref{assum:C} (for $\ell=4$), \ref{assum:D}(ii), \ref{assum:B}, and \ref{assum:F} hold. 
Suppose further the data taper $h$ is either constant on $[-1/2,1/2]^d$ or satisfies Assumption \ref{assum:E}(for $m=d+1$).
Then,
\begin{equation*}
\lim_{n\rightarrow \infty} \var(V_{h,n}(\phi)) = (2\pi)^{d} (H_{h,4} / H_{h,2}^2) (\Omega_1 + \Omega_2),
\end{equation*} where $\Omega_1$ and $\Omega_2$ are defined as in (\ref{eq:Omegas}).
\end{theorem}
\textit{Proof}. For $\kb \in A_n$, let $\widehat{G}_{h,n}^{(\kb)}(\phi)$ be an independent copy of $\widetilde{G}_{h,n}^{(\kb)}(\phi)$ and let 
\begin{equation} \label{eq:Vhat}
\widehat{V}_{h,n}(\phi) = k_n^{-1/2} \sum_{\kb \in A_n} \widehat{G}_{h,n}^{(\kb)}(\phi).
\end{equation}
 Then, by using a standard telescoping sum argument (cf. \cite{p:paw-09}), one can easily shown that $V_{h,n}(\phi) - \widehat{V}_{h,n}(\phi) = o_p(1)$ as $n\rightarrow \infty$. Therefore, 
\begin{equation*}
\lim_{n\rightarrow \infty} \var(V_{h,n}(\phi))
= \lim_{n\rightarrow \infty} \var(\widehat{V}_{h,n}(\phi))
= \lim_{n \rightarrow \infty} k_n^{-1} \sum_{\kb \in A_n} \var(\widetilde{G}_{h,n}^{(\kb)}(\phi)).
\end{equation*} Now, we focus on the variance of $\widehat{G}_{h,n}^{(\kb)}(\phi)$. Since $D_n^{(\kb)}$ also has a rectangle form that satisfies Assumption \ref{assum:A}, from the modification of the proof of the asymptotic variance of $|D_n|^{1/2}A_{h,n}(\phi)$ in Theorem \ref{thm:A1}(ii), one can show that
\begin{eqnarray*}
\var(\widetilde{G}_{h,n}^{(\kb)}(\phi)) \approx (2\pi)^{d} |D_n^{(\kb)}|^{-1} 
\left( 
\frac{\int_{D_n^{(\kb)}} \{h(\xb/\aB)\}^4 d\xb}{H_{h,2}^2}\right)
(\Omega_1 + \Omega_2), \quad \kb \in A_n.
\end{eqnarray*} Therefore, we have
\begin{eqnarray*}
\lim_{n\rightarrow \infty} k_n^{-1}  \sum_{\kb \in A_n} \var(\widetilde{G}_{h,n}^{(\kb)}(\phi)) 
&=& (2\pi)^{d} \lim_{n\rightarrow \infty} k_n^{-1} |D_n^{(\kb)}|^{-1}  H_{h,2}^{-2} \left( \int_{\widetilde{D}_n}\{h(\xb/\aB)\}^4 d\xb \right)(\Omega_1 + \Omega_2) \\
&=& (2\pi)^{d} \frac{H_{h,4}}{H_{h,2}^{2}}(\Omega_1 + \Omega_2).
\end{eqnarray*} Here, the last identity is due to (\ref{eq:sub-block-vol}), $k_n |D_n^{(\kb)}| = |\widetilde{D}_n|$, and $\int_{D_n} h(\xb/\aB)^4d\xb = |D_n| H_{h,4}$.
Thus, we get the desired result.
\hfill $\Box$

\vspace{0.5em}

Now, we are ready to prove the CLT for $V_{h,n}(\phi)$.

\begin{theorem} \label{thm:CLT-Vhn}
Suppose that Assumptions \ref{assum:A}, \ref{assum:C} (for $\ell=8$), \ref{assum:D}(ii), \ref{assum:B}, and \ref{assum:F} hold. 
Suppose further the data taper $h$ is either constant on $[-1/2,1/2]^d$ or satisfies Assumption \ref{assum:E}(for $m=d+1$).
\begin{equation*}
V_{h,n}(\phi) \Dcon \mathcal{N}\left(0,(2\pi)^{d} (H_{h,4} / H_{h,2}^2) (\Omega_1 + \Omega_2) \right), \quad n \rightarrow \infty.
\end{equation*}
\end{theorem}
\textit{Proof}. Recall $\widehat{V}_{h,n}(\phi)$ from (\ref{eq:Vhat}) has the same asymptotic distribution with $V_{h,n}(\phi)$. To show the CLT for $\widehat{V}_{h,n}(\phi)$, we only need to check the Lyapunov condition: for some $\delta > 0$,
\begin{equation} \label{eq:V-lyapunov}
\lim_{n\rightarrow \infty} k_n^{-(2+\delta)/2} \sum_{\kb \in A_n} \Ex [|\widetilde{G}_{h,n}^{(\kb)}(\phi)|^{2+\delta}] =0.
\end{equation} We will show the above for $\delta=2$, provided Assumption \ref{assum:C} for $\ell=8$. To show (\ref{eq:V-lyapunov}), it is enough to show
\begin{equation} \label{eq:G4-bound}
\sup_{n \in \N}\Ex[ |\widetilde{G}_{h,n}(\phi)|^{4}] <\infty.
\end{equation} This is because, once we show (\ref{eq:G4-bound}), one can show
\begin{equation*}
k_n^{-2} \sum_{\kb \in A_n} \Ex [|\widetilde{G}_{h,n}^{(\kb)}(\phi)|^{4}] \leq C k_n^{-1} \rightarrow 0, \quad n\rightarrow \infty.
\end{equation*} Thus, (\ref{eq:V-lyapunov}) holds for $\delta=2$.
Using (\ref{eq:Ex4}), we have
\begin{equation*}
\sup_{n \in \N}\Ex[ |\widetilde{G}_{h,n}(\phi)|^{4}] \leq \sup_{n \in \N} \kappa_4(\widetilde{G}_{h,n}(\phi))
+ 3  \left( \sup_{n \in \N} \var(\widetilde{G}_{h,n}(\phi)) \right)^2.
\end{equation*} From Theorem \ref{thm:A1}(ii), we have
\begin{equation} \label{eq:varG}
\lim_{n\rightarrow \infty}
\var(\widetilde{G}_{h,n}(\phi)) = (2\pi)^d (H_{h,4}/H_{h,2}^2) (\Omega_1 + \Omega_2).
\end{equation} Therefore, $\sup_{n \in \N} \var(\widetilde{G}_{h,n}(\phi)) < \infty$. To bound the fourth-order cumulant term, we note that
\begin{eqnarray}
\kappa_4(\widetilde{G}_{h,n}(\phi)) &=&
|D_n|^2 \kappa_4 \left( \int_{D} \phi(\ob) I_{h,n}(\ob)\right) 
= |D_n|^2 \int_{D^4} \left( \prod_{i=1}^{4} \phi(\ob_i) \right)
\nonumber \\
&&~~\times 
\cum(I_{h,n}(\ob_1),I_{h,n}(\ob_2),I_{h,n}(\ob_3),I_{h,n}(\ob_4)) d\ob_1 d\ob_2 d\ob_3 d\ob_4.
 \label{eq:kappa4-G}
\end{eqnarray}  Now, we will evaluate the cumulant term above. Note that
\begin{eqnarray*}
&& \cum(I_{h,n}(\ob_1),I_{h,n}(\ob_2),I_{h,n}(\ob_3),I_{h,n}(\ob_4)) \\ 
&&~= \cum \bigg(J_{h,n}(\ob_1) J_{h,n}(-\ob_1),
J_{h,n}(\ob_2) J_{h,n}(-\ob_2), 
J_{h,n}(\ob_3) J_{h,n}(-\ob_3),
J_{h,n}(\ob_4) J_{h,n}(-\ob_4)
\bigg).
\end{eqnarray*}
Using indecomposable partitions, the above joint cumulant can be written as sum of product of cumulants of form $\cum(J_{h,n}(\pm \ob_{i_1}), \dots, J_{h,n}(\pm \ob_{i_k}))$, where $k \in \{2, \dots, 8\}$ and $i_j \in \{1,2,3,4\}$ for $j \in \{1, \dots, k\}$. By using an argument in Lemma \ref{lemma:k4-bound}, the leading term (which has the largest order) is a product of four joint cumulants of order two. An example of such term is
\begin{equation}
\begin{aligned}
& \cum(J_{h,n}(\ob_1), J_{h,n}(\ob_2)) \cum(J_{h,n}(-\ob_1), J_{h,n}(\ob_3)) \\
&~~\times \cum(J_{h,n}(-\ob_2), J_{h,n}(-\ob_4)) \cum(J_{h,n}(-\ob_3), J_{h,n}(\ob_4)).
\end{aligned}
\label{eq:cum-ex}
\end{equation} 

Now, we will bound one of the terms in (\ref{eq:kappa4-G}) that is associated with the above cumulant products.
Let $\phi(\ob)=0$ outside the domain $D$. By using Lemma \ref{lemma:cov-exp}, we have
\begin{eqnarray}
&& \cum(J_{h,n}(\ob_1), J_{h,n}(\ob_2))  \nonumber  \\
&& = |D_n|^{-1} H_{h,2}^{-1} f(\ob_1) H_{h,2}^{(n)}(\ob_1+\ob_2)
+ C |D_n|^{-1} \int_{\R^d} e^{-i\ubb^\top \ob_1} C(\ubb) R_{h,h}^{(n)}(\ubb,\ob_1+\ob_2),
\label{eq:cum2-exp11}
\end{eqnarray} where $H_{h,2}^{(n)}(\cdot)$ and $R_{h,h}^{(n)}(\cdot,\cdot)$ are defined as in (\ref{eq:Hkn}) and (\ref{eq:Rhgn}). Therefore, the integral term in the decomposition of (\ref{eq:kappa4-G}) that is associated with (\ref{eq:cum-ex}) has 16 terms. We will only bound the two representative terms. Other 14 terms will be bounded in the similar way. The first representative term is 
\begin{eqnarray*}
&& H_{h,2}^{-4} |D_n|^{-2} \int_{\R^{4d}}
\prod_{i=1}^{4} \phi(\ob_i) \times f(\ob_1) f(-\ob_1) f(-\ob_2) f(\ob_3) \\
&&~~ \times H_{h,2}^{(n)}(\ob_1+\ob_2) H_{h,2}^{(n)}(-\ob_1+\ob_3) H_{h,2}^{(n)}(-\ob_2-\ob_4)H_{h,2}^{(n)}(-\ob_3+\ob_4)
d\ob_1 d\ob_2 d\ob_3 d\ob_4 \\
&& = C \int_{\R^{4d}}
\psi_{}(\tb_1,\tb_2,\tb_3, \tb_4)  \times c_{h^2,n}(\tb_1) c_{h^2,n}(\tb_2) c_{h^2,n}(\tb_3)c_{h^2,n}(-(\tb_1+\tb_2+\tb_3))
d\tb_1 d\tb_2 d\tb_3 d\tb_4,
\end{eqnarray*} where for $\tb_1, \dots, \tb_4 \in \R^d$.
\begin{eqnarray*}
\psi_{}(\tb_1,\tb_2,\tb_3, \tb_4)
&=& \alpha \times \phi(\tb_1+\tb_3+\tb_4) \phi(-\tb_3-\tb_4)\phi(\tb_1+\tb_2+\tb_3+\tb_4) \phi(\tb_4) \\
&&\quad \times
f(\tb_1+\tb_3+\tb_4) f(-\tb_1-\tb_3-\tb_4) f(-\tb_3-\tb_4) f(\tb_1+\tb_2+\tb_3+\tb_4).
\end{eqnarray*}
 Here,
we use change of variables $\tb_1 = \ob_1+\ob_2$, $\tb_2=-\ob_1+\ob_3$, $\tb_3 = -\ob_2-\ob_4$, and $\tb_4 = \ob_4$ in the identity above and $\alpha$ in $\psi_{}(\cdot, \cdot, \cdot, \cdot)$ is the Jacobian determinant. Since $\phi(\cdot)$ is bounded and has a compact support
and $\sup_{\ob \in \R^d} f(\ob)<\infty$, it is easily seen that $|\psi_{}(\cdot, \cdot, \cdot, \cdot)|$ is also bounded and has a compact support. Then, by iteratively applying Lemma \ref{lemma:Cn}(c),(d), and (e), we have
\begin{equation*}
\left| \int_{\R^{4d}}
\psi_{}(\tb_1,\tb_2,\tb_3, \tb_4)  c_{h^2,n}(\tb_1) c_{h^2,n}(\tb_2) c_{h^2,n}(\tb_3)c_{h^2,n}(-(\tb_1+\tb_2+\tb_3))
d\tb_1 d\tb_2 d\tb_3 d\tb_4\right| = O(|D_n|^{-1/2})
\end{equation*} as $n\rightarrow \infty$. 

The second representative term is 
\begin{eqnarray*}
&& C^4 |D_n|^{-2} \int_{\R^{4d}}
\prod_{i=1}^{4} \phi(\ob_i) \int_{\R^{4d}} \prod_{i=1}^{4} C(\ubb_i) \times
e^{-i(\ubb_1^\top \ob_1 - \ubb_2^\top \ob_1 - \ubb_3^\top \ob_2 + \ubb_4^\top \ob_3) } \\
&&~~\times R_{h,h}^{(n)}(\ubb_1, \ob_1+\ob_2)R_{h,h}^{(n)}(\ubb_2, -\ob_1+\ob_3)
R_{h,h}^{(n)}(\ubb_3, -\ob_2-\ob_4)R_{h,h}^{(n)}(\ubb_4, -\ob_3+\ob_4).
\end{eqnarray*} To bound the above term, we require a sharp bound for $R_{h,h}^{(n)}$.
Let $\{h_{\jb}\}_{\jb \in \Z^d}$ be the Fourier coefficients of $h$ that satisfies (\ref{eq:H-fourier}). Then, by using Theorem \ref{thm:Rhgn1}(ii) together with an inequality $\rho(|x|) \leq 1$
and change of variables $\tb_1 = \ob_1+\ob_2$, $\tb_2=-\ob_1+\ob_3$, $\tb_3 = -\ob_2-\ob_4$, and $\tb_4 = \ob_4$,
 the above is bounded by
\begin{eqnarray*}
&& C \sum_{p_1, \cdots, p_4 = 0}^{m_d} 
\sum_{\jb_1, \dots, \jb_4, \kb_1, \dots, \kb_4 \in \Z^d} \left( \prod_{i=1}^{4} |h_{\jb_i} h_{\kb_i}| \right)
 \int_{\R^{4d}} 
 d\ubb_1 d\ubb_2 d\ubb_3 d\ubb_4
\prod_{i=1}^{4} |C(\ubb_i)|  
\\
&& ~~ \times \int_{\R^{4d}} 
|\widetilde{\psi}_{}(\tb_1,\tb_2, \tb_3, \tb_4)| \times 
\bigg| c_{D_{n,p_1}(\ubb_1)}(\tb_1 -2\pi(\jb_1+\kb_1)/\aB) 
\\
&& ~~\times
c_{D_{n,p_2}(\ubb_2)}(\tb_2 -2\pi(\jb_2+\kb_2)/\aB) 
c_{D_{n,p_3}(\ubb_3)}(\tb_3 -2\pi(\jb_3+\kb_3)/\aB)  \\
&&~~\times
c_{D_{n,p_4}(\ubb_4)}(-(\tb_1+\tb_2+\ob_3)+\ob_4-2\pi(\jb_4+\kb_4)/\aB) \bigg| d\tb_1 d\tb_2 d\tb_3 d\tb_4,
\end{eqnarray*} where $\widetilde{\psi}_{}(\tb_1,\tb_2, \tb_3, \tb_4)$ is a bounded function with compact support.

By using Lemma \ref{lemma:Cn}(d),(e)  the term in integral with respect to $d\tb_1 d\tb_2 d\tb_3 d\tb_4$ is uniformly bounded above. With this observation together with $\sum_{j \in \Z^d} |h_{\jb}|$ and $C(\ubb) \in L^{1}(\R^d)$, the above term is $O(1)$ as $n\rightarrow \infty$.
 Therefore, we conclude that the decomposition of (\ref{eq:kappa4-G})  associated with (\ref{eq:cum-ex}) is $O(1)$ as $n \rightarrow \infty$. 

Similarly, all other terms in the indecomposable partition are $O(1)$ as $n\rightarrow \infty$. Thus,  $\sup_{n \in \N}|\kappa_4(\widetilde{G}_{h,n}(\phi))| <\infty$ and this shows (\ref{eq:G4-bound}). 

Lastly, once we verify (\ref{eq:G4-bound}), the Lyapunov condition in (\ref{eq:V-lyapunov}) is also true, thus, combining this with Theorem \ref{thm:var-Vhn}, we obtain the desired results. 
\hfill $\Box$

\section{Estimation of the asymptotic variance} \label{sec:subsampling}
Recall the integrated periodgram $\widehat{A}_{h,n}(\phi)$ in (\ref{eq:Aphi}). By Theorem \ref{thm:A1}, the $(1-\alpha)$ ($\alpha \in (0,1)$) confidence interval of the spectral mean $A(\phi)$ of form (\ref{eq:Aphi}) is
\begin{equation*}
\widehat{A}_{h,n}(\phi) \pm \frac{z_{1-\alpha/2}}{|D_n|^{1/2}} (2\pi)^{d/2} (H_{h,4}^{1/2}/H_{h,2}) \sqrt{\Omega_1 + \Omega_2},
\end{equation*} where $z_{1-\alpha/2}$ is the $(1-\alpha/2)$-th quantile of the standard normal random variable and
$\Omega_1$ and $\Omega_2$ are defined as in (\ref{eq:Omegas}). The quantity $\Omega_1 + \Omega_2$ is in terms of the unknown spectral density function and complete fourth-order spectral density function. In this section, we sketch the procedure to estimate the asymptotic variance  $\lim_{n \rightarrow \infty} |D_n| \var( \widehat{A}_{h,n}(\phi)) = (2\pi)^{d} (H_{h,4}/H_{h,2}^2) (\Omega_1 + \Omega_2)$ by the mean of subsampling.

To ease the presentation, we assume that $|D_n|$ grows proportional to $n^{d}$ as $n \rightarrow \infty$. Thus the side lengthes $A_1, \cdots, A_d$ increases proportional to order of $n$. For $n \in \N$, let $0<a_1 < a_2 <\dots$ be a sequence of increasing numbers such that $a_n = o(n)$ as $n \rightarrow \infty$. Therefore, we have $\lim_{n\rightarrow \infty} a_n / A_{i}(n) = 0$ for any $i \in \{1, \dots, d\}$. Now, let
\begin{equation*}
T_{n} = \{\kb: \kb \in \Z^d~~\text{and}~~ B_{n}^{(\kb)} = \kb + [-a_n/2, a_n/2]^d \subset D_n\}, \quad n \in \N.
\end{equation*} Unlike the subrectangle $D_n^{(\kb)}$ in Appendix \ref{sec:mixing}, $\{B_{n}^{(\kb)}\}_{\kb \in T_n}$ are the subrectangles of $D_n$ that can be overlapped. For $n \in \N$ and $\kb \in T_n$, we define subsampling analogous of the DFT by
\begin{equation*}
\mathbb{J}_{h,n}^{(\kb)}(\ob) = 
(2\pi)^{-d/2} H_{h,2}^{-1/2} | B_{n}^{(\kb)}|^{-1/2} \sum_{\xb \in X \cap B_{n}^{(\kb)} }
h(a_n^{-1} (\xb-\kb)) \exp(-i \xb^\top \ob), \quad \ob \in \R^d.
\end{equation*} Note that the above definition is slightly different from $\mathcal{J}_{h,n}^{(\kb)}(\ob)$ of (\ref{eq:J-kb}) since $\mathbb{J}_{h,n}^{(\kb)}(\cdot)$ alters the data taper function for each subretangle. By simple calculation, we have $\Ex[\mathbb{J}_{h,n}^{(\kb)}(\ob)] = \lambda c_{h,n}^{(\kb)}(\ob)$, where for $n \in \N$, $\kb \in T_n$, and $\ob \in \R^d$,
\begin{equation*}
c_{h,n}^{(\kb)}(\ob) = (2\pi)^{-d/2} H_{h,2}^{-1/2} | B_{n}^{(\kb)}|^{-1/2} \exp(i\kb^\top \ob) \int_{[-a_n/2,a_n/2]^d} h(a_n^{-1}\xb)\exp(-i\xb^\top \ob) d\xb.
\end{equation*} Therefore, the subsample version of the integrated periodogram is given by
\begin{equation*}
\widehat{A}_{h,n}^{(\kb)}(\phi) = \int_{D} \phi(\ob) \left|
\mathbb{J}_{h,n}^{(\kb)}(\ob) - \widehat{\lambda}_{h,n} c_{h,n}^{(\kb)}(\ob)
\right|^2 d\ob, \quad n \in \N, \quad \kb \in T_n,
\end{equation*} where $\widehat{\lambda}_{h,n}$ is an unbiased tapered estimator of the first-order intensity function.
In practice, one can approximate $\widehat{A}_{h,n}^{(\kb)}(\phi)$ using Riemann sum as outlined in Section \ref{sec:practice}. Now, our subsampling estimator of the asymptotic variance of $|D_n| \var( \widehat{A}_{h,n}(\phi))$ is
\begin{equation*}
\zeta_{n} = \frac{\alpha_n^d}{|T_n|} \sum_{\kb \in T_n} \left\{ \widehat{A}_{h,n}^{(\kb)}(\phi) - |T_n|^{-1} \sum_{\jb \in T_n} \widehat{A}_{h,n}^{(\kb)}(\phi) 
\right\}^2, \quad n \in \N,
\end{equation*} where $\alpha_n^d$ is the common volume of $B_n^{(\kb)}$. Under appropriate moment and mixing conditions such as conditions $(\mathcal{S}1)$--$(\mathcal{S}6$) in \cite{p:bis-19} (page 1174), one may expect that $\zeta_{n}$ is a consistent estimator of the asymptotic variance $(2\pi)^{d} (H_{h,4}/H_{h,2}^2) (\Omega_1 + \Omega_2)$. A rigourous proof of the sampling properties of $\zeta_{n}$ will be reported in future research.

\section{Additional simulation results} \label{sec:add-sims}

\subsection{Computation and illustration of the kernel spectral density estimator} \label{sec:add-sim}

In this section, we describe the computation of the kernel spectral density estimator and provide illustrations. Recall $\widehat{f}_{n,b}(\ob)$ in (\ref{eq:KSD}). Since $\widehat{f}_{n,b}(\ob)$ has an integral form, Riemann sum approximation of $\widehat{f}_{n,b}(\ob)$ is
\begin{eqnarray} 
\widehat{f}_{n,b}^{(R)}(\ob_{})
&=& \frac{\sum_{\kb \in \Z^d} W_{b}(\ob_{} - \ob_{\kb,A}) \widehat{I}_{h,n}(\ob_{\kb,A})}{\sum_{\kb \in \Z^d} W_{b}(\ob_{} - \ob_{\kb,A})}
, \quad \ob \in \R^d,
\label{eq:KSD-sum}
\end{eqnarray} where $A>0$ is the side length of the observation domain $D_n = [-A/2,A/2]^d$ and $\ob_{\kb,A} = 2\pi \kb / A$ for $\kb \in \Z^d$. The summation above is a finite sum due to the fact that $W_b(\cdot)$ has a support $[-b/2,b/2]^d$. For a selection of the kernel function, we choose a triangular kernel $W(\xb) = W(x_1)W(x_2)$ for $\xb = (x_1, x_2)^\top \in \R^2$ where $W(x) = 2\max\{ 1-2|x|, 0\}$. The bandwidth $b \in (0,\infty)$ is set at $b = |D_n|^{-1/6}$ which is an optimal rate in the sense of mean-squared error criterion (see \cite{p:dsy-24}, Section 6.1 for details).

In top panels of Figure \ref{fig:ksde} below, we calculate $\widehat{f}_{n,b}^{(R)}(\ob_{})$ of each periodogram that are computed in the bottom panels of Figure \ref{fig:motiv}. In the middle panels, we evaluate the absolute biases of $\widehat{I}_{h,n}$ and in the bottom panels, we evaluate the absolute biases of $\widehat{f}_{n,b}^{(R)}$. 

\begin{figure}[ht!]
\centering

\textbf{Kernel spectral density estimator (KSDE)}

\includegraphics[width=0.9\textwidth]{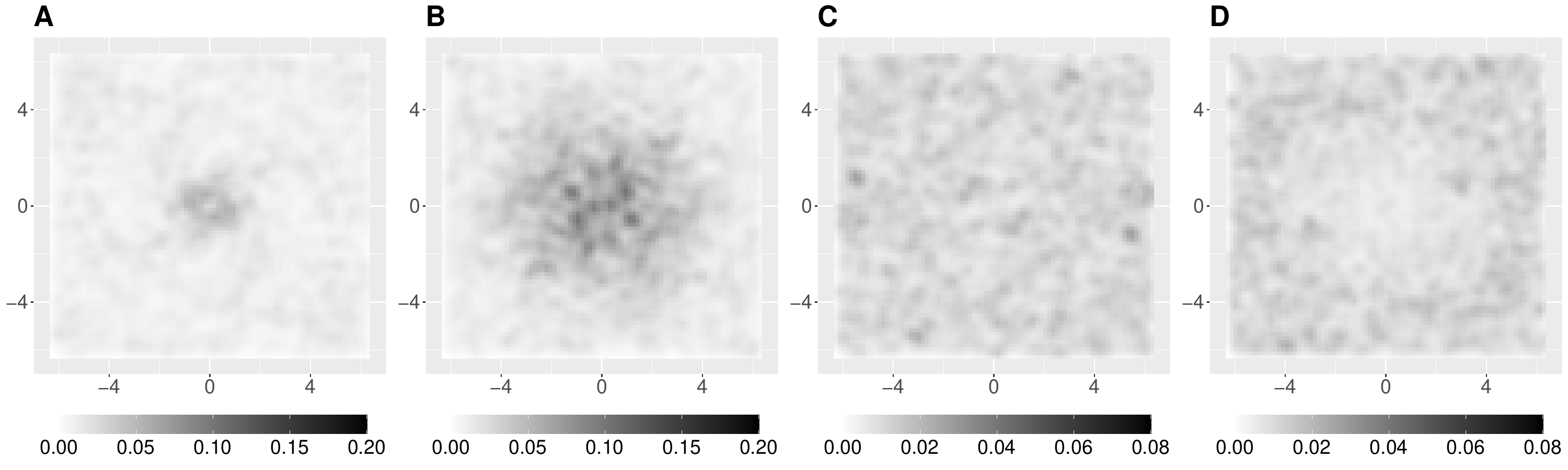}

\textbf{Absolute bias of the periodogram}

\includegraphics[width=0.9\textwidth]{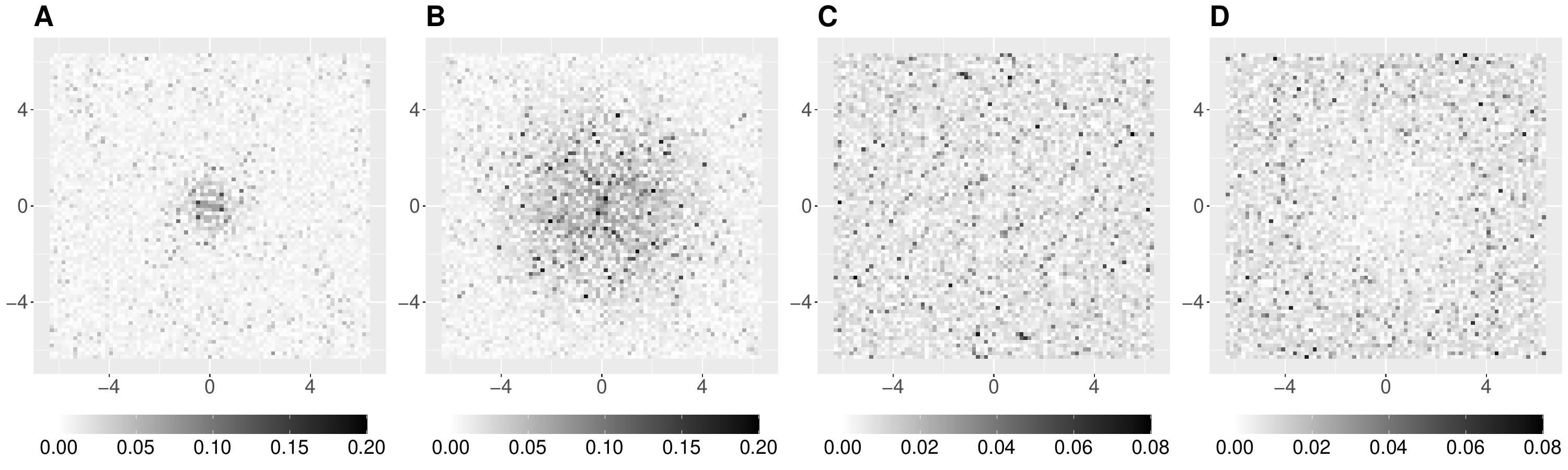}

\textbf{Absolute bias of the KSDE}

\includegraphics[width=0.9\textwidth]{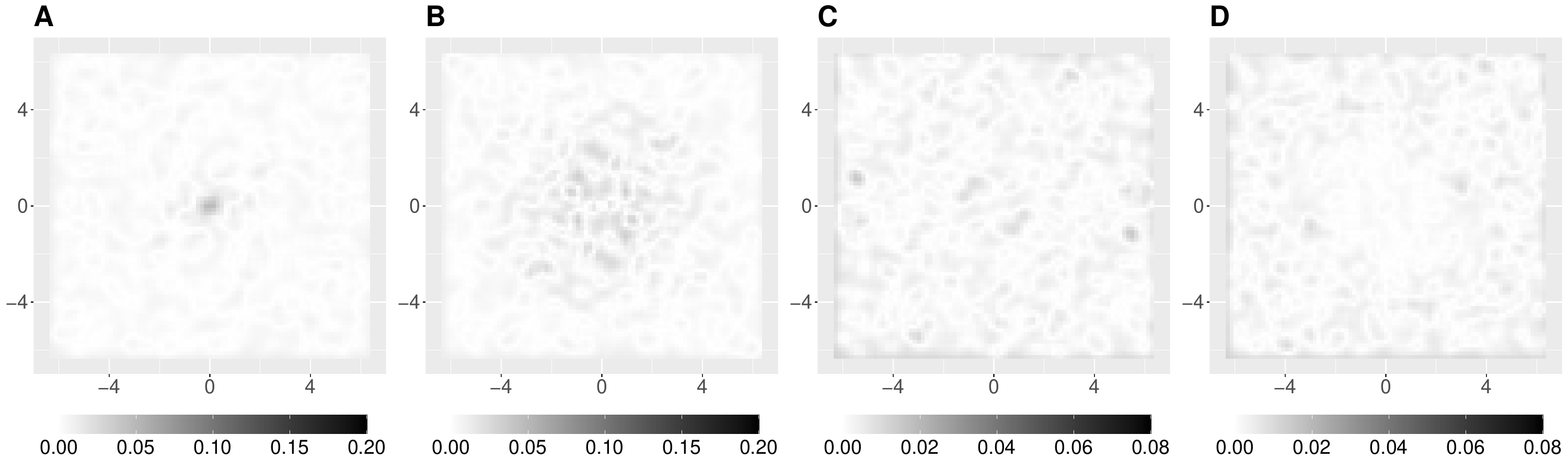}
\caption{\textit{ Top: Kernel spectral density estimators as in (\ref{eq:KSD-sum}) of the periodograms that are computed in the bottom panels of Figure \ref{fig:motiv}. Middle: $|\widehat{I}_{h,n}(\ob) - f(\ob)|$ for each model. Bottom: $\widehat{f}_{n,b}^{(R)}(\ob) - f(\ob)|$ for each model.
}}
\label{fig:ksde}
\end{figure}

For all models, the absolute bias of the periodograms (middle panels) have similar patterns to those of the corresponding spectral density functions. This can be explained by using Theorem \ref{thm:In-bias} that $|\widehat{I}_{h,n}(\ob) - f(\ob)| \approx \var(\widehat{I}_{h,n}(\ob))^{1/2} = f(\ob)$ for $\ob \in \R^d \backslash \{\textbf{0}\}$. Therefore, the middle panels indicate that the periodogram is inconsistent. However, the absolute bias of the smoothed periodograms (bottom panels) are nearly zero across all frequencies and all models. This solidifies the thoeretical results in Thoerem \ref{thm:KSD} which states that $\widehat{f}_{n,b}(\ob)$ is a consistent estimator of $f(\ob)$ for all $\ob \in \R^d$.

\subsection{Computation of the periodograms} \label{eq:per-compute}
In this section, we discuss an implementation of computing periodograms to evaluate the discretized Whittle likelihood in (\ref{eq:discrete-Whittle}). For the simplicity, we assume $d=2$. The cases when $d=1$ or $d \in \{3,4, \dots\}$ can be treated similarly.

Let the observation window be $D_n = [-A_1/2,A_1/2]\times [-A_2/2, A_2/2]$ for $A_1, A_2 \in (0,\infty)$. Suppose that the prespecified domain $D \subset \R^2$ has a rectangle form centered at the origin, thus the gridded version of $D$, denotes $D_{\text{grid}} = \{2\pi \kb / \Omega: \kb \in \Z^2,2\pi \kb / \Omega \in D\}$, also forms a rectangular grid. Let this retangular grid can be written as $D_{\text{grid}} = \{(2\pi k_1 / \Omega ,2\pi k_2 / \Omega): |k_i| \leq a_i,~~ k_i \in \Z\}$ for some $a_1, a_2 \in \N$. Suppose further that data taper function is separable, i.e., $h(\xb) = h_{1}(x_1) h_{2}(x_2)$ for some $h_1, h_2$ and let 
\begin{equation} \label{eq:ui}
u_i(\omega, A) = H_{i}^{(n)}(\omega) = \int_{-A/2}^{A/2} h_{i}(x/A) \exp(-ix\omega) dx, \quad i \in \{1,2\}.
\end{equation}
We will assume that $u_i(\omega, A)$ has a closed form expression, thus, there is no additional computational burden to approximate the integral in $u_i$.

Now, we will discuss an efficient way to compute $\{\widehat{I}_{h,n}(\ob): \ob \in D_{\text{grid}} \}$ based on the observed point pattern $\{\xb_j= (x_{j,1}, x_{j,2})^\top: 1\leq j \leq m\}$ in $D_n$. From its definition, $\widehat{I}_{h,n}(\ob) = |\mathcal{J}_{h,n}(\ob) - \widehat{\lambda}_{h,n} c_{h,n}(\ob)|^2$ where $\mathcal{J}_{h,n}(\ob)$ and $c_{h,n}(\ob)$ are as in (\ref{eq:mathcalDFT-h}) and (\ref{eq:Cn-h}), respectively. Therefore, we will compute $\mathcal{J}_{h,n}(\ob)$ and $c_{h,n}(\ob)$ separately. Since $h$ is separable, the tapered DFT can be written as $\mathcal{J}_{h,n}(2\pi \kb / \Omega) = C \sum_{j=1}^{m} v_1(x_{j,1}, k_1) v_2(x_{j,2}, k_2)$, where
$C = (2\pi)^{-d/2} H_{h,2}^{-1/2} |D_n|^{-1/2}$ and 
$v_i(x,k) = h_j(x/A_i) \exp(-ix_{}(2\pi k/\Omega))$, $i \in \{1,2\}$. Then, the matrix form of $\{\mathcal{J}_{h,n}(\ob): \ob \in D_{\text{grid}} \}$ is equal to $C V_1^\top V_2$, where for $i \in \{1,2\}$,
\begin{equation*}
V_{i} = [\vbb_{i}(-a_i) | \vbb_{i}(-a_i+1) | \cdots | \vbb_{i}(a_i)], \quad \text{where} \quad
\vbb_{i}(k) = ( v_i(x_{i,1},k), \dots, v_i(x_{i,m},k))^\top.
\end{equation*}

Next, we calculate $c_{h,n}(2\pi \kb / \Omega)$. Again using separability of $h$, we have
\begin{equation*}
c_{h,n}(2\pi \kb / \Omega) = C H_{h,1}^{(n)}(2\pi \kb / \Omega) = C u_{1}(2\pi k_1/\Omega, A_1) u_{2}(2\pi k_2/\Omega, A_2).
\end{equation*} Here, $C$ is the same constant as above and $u_1, u_2$ are as in (\ref{eq:ui}). Therefore, a matrix form of $\{c_{h,n}(\ob): \ob \in D_{\text{grid}}\}$ is $C U_1 U_2^\top$, where
\begin{equation*}
U_i = (u_{i}(2\pi (-a_i)/\Omega, A_1), \dots, u_{i}(2\pi (a_i)/\Omega, A_1)^\top \in \C^{2a_i+1}.
\end{equation*} These give algorithms for the fast computation of the periodogram on grid.

\subsection{Additional Figures} \label{sec:add-fig}

In this section, we provide supplementary figures for the simulation results in Section \ref{sec:sim}.


\begin{figure}[ht!]
\centering
\includegraphics[width=0.9\textwidth]{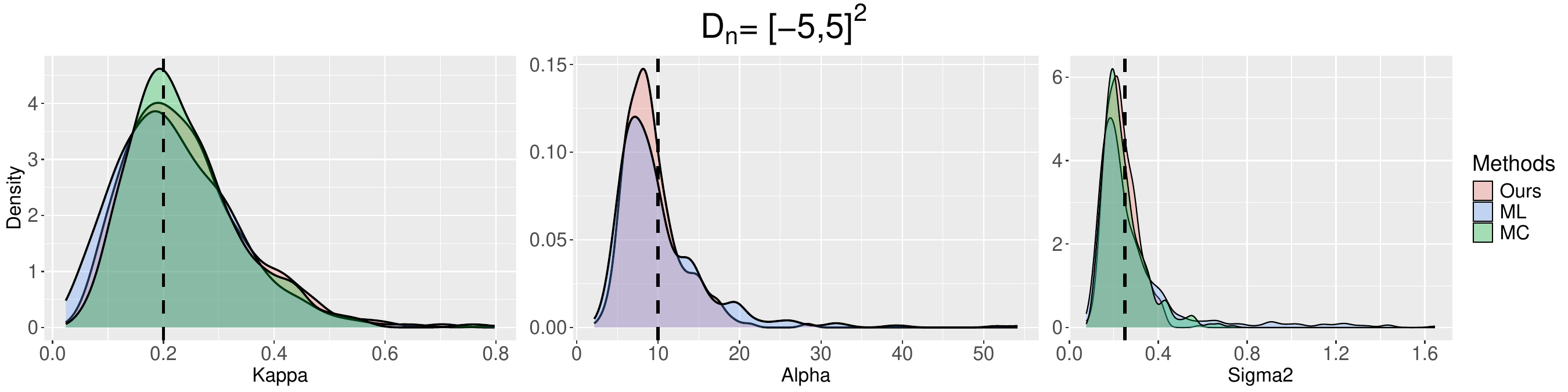}
\includegraphics[width=0.9\textwidth]{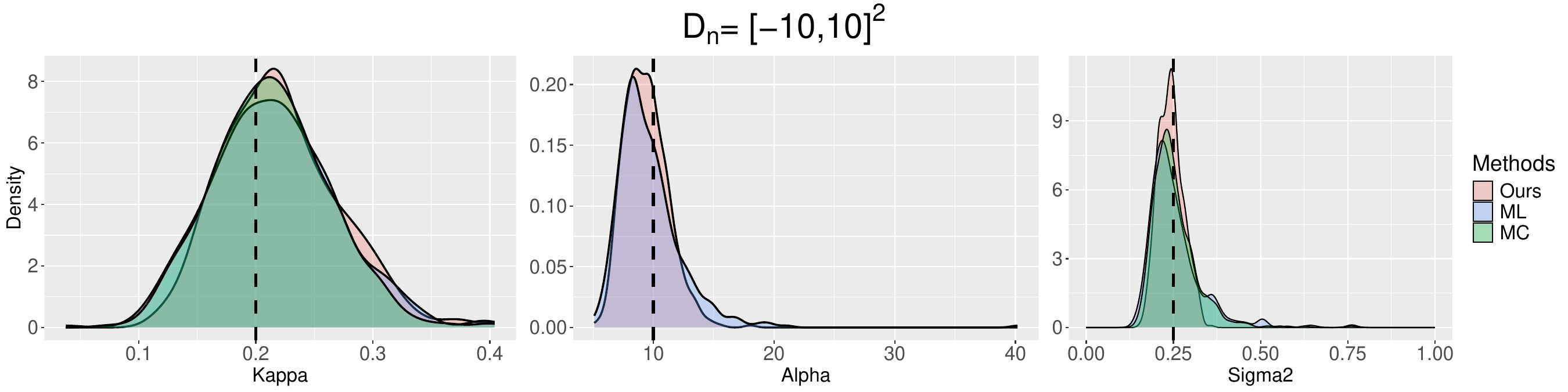}
\includegraphics[width=0.9\textwidth]{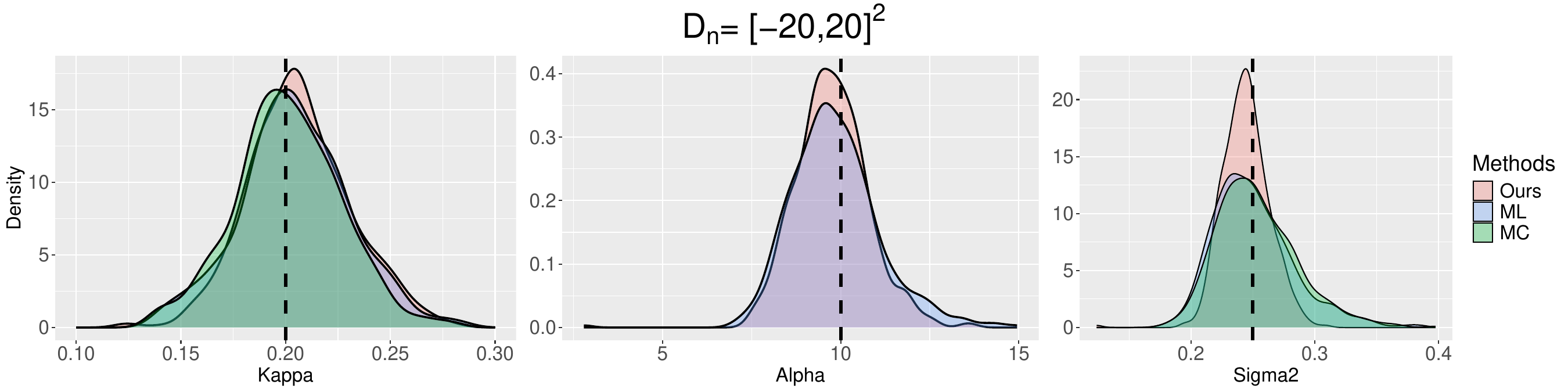}
\caption{\textit{
Densities of the estimated parameters for the correctly specified Thomas clustering process model as in Section \ref{sec:CS}. Each row refers to different observation domains. Vertical dashed lines refer to the true parameter values.
}}
\label{fig:TCP-density}
\end{figure}

\begin{figure}[ht!]
\centering
\includegraphics[width=0.9\textwidth]{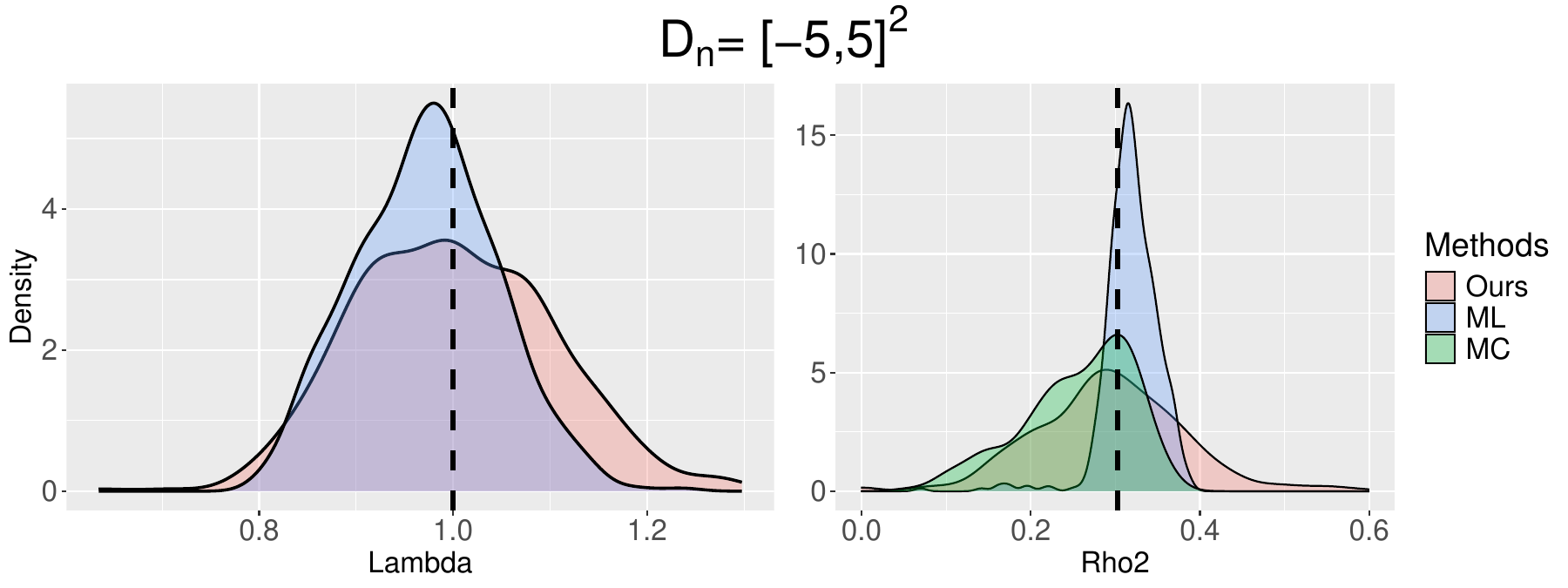}
\includegraphics[width=0.9\textwidth]{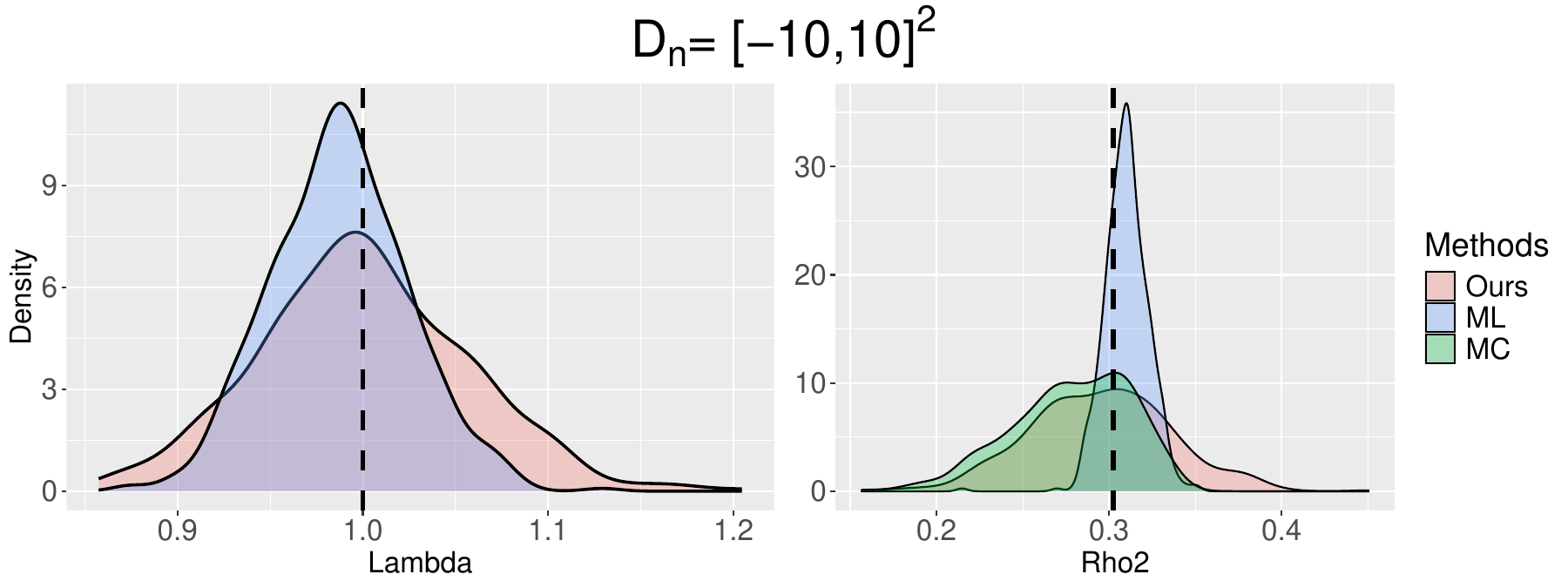}
\includegraphics[width=0.9\textwidth]{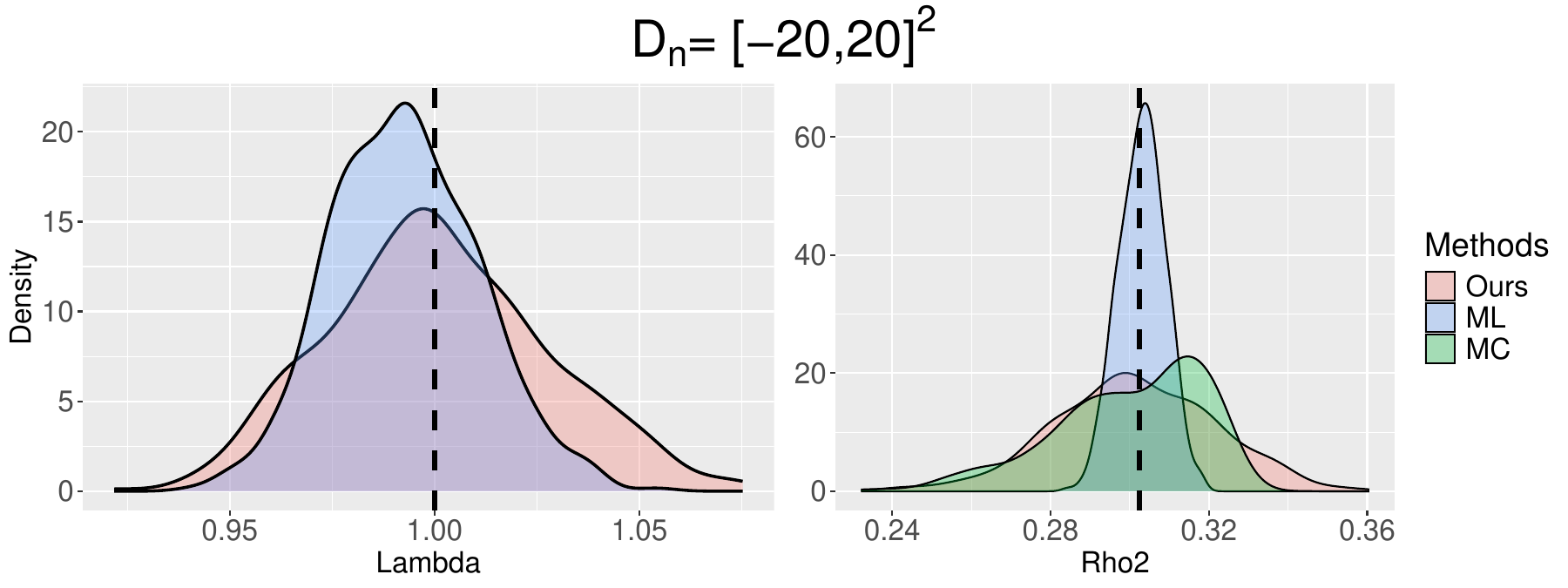}
\caption{\textit{
Densities of the estimated parameters for the correctly specified determinantal point process with Gaussian kernel as in Section \ref{sec:CS}. Each row refers to different observation domains. Vertical dashed lines refer to the true parameter values. 
}}
\label{fig:GDPP-density}
\end{figure}


\begin{figure}[ht!]
\centering
\includegraphics[width=\textwidth]{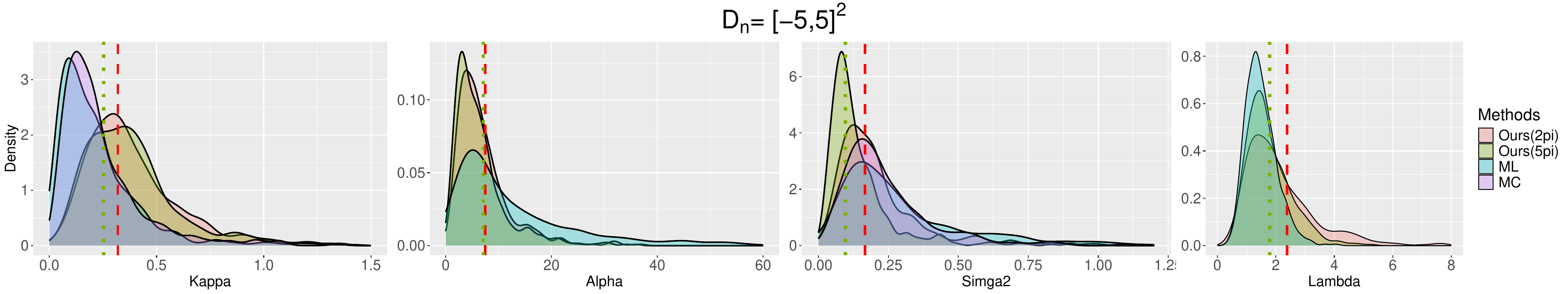}
\includegraphics[width=\textwidth]{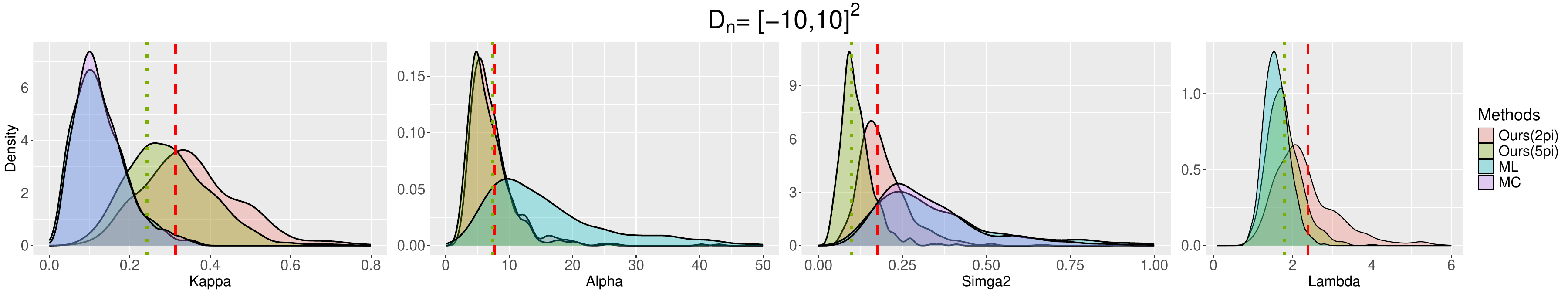}
\includegraphics[width=\textwidth]{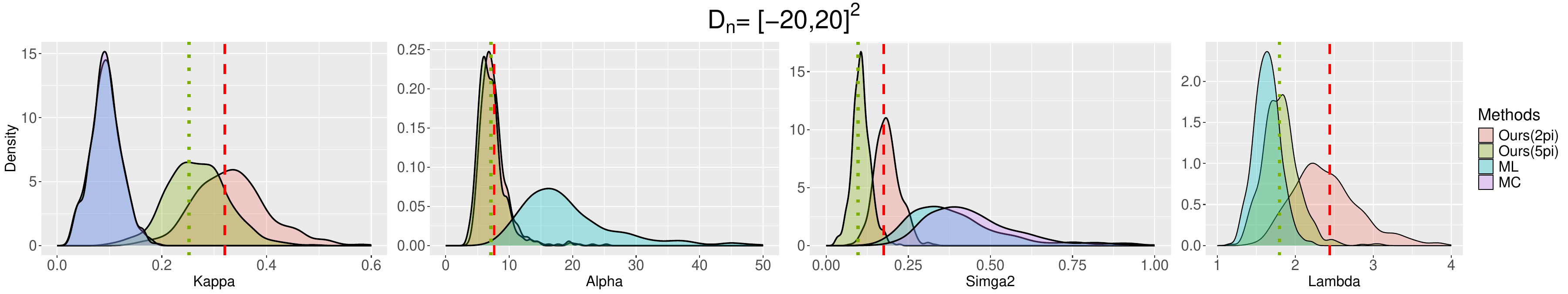}
\caption{\textit{
Densities of the estimated parameters for the misspecified LGCP fitting with the full TCP model described as in Section \ref{sec:MS}. Each row refers to different observation domains. Dashed (green) and dotted (red) vertial lines refer to the best fitting TCP parameters as in (\ref{eq:mathcalL-R}) evaluated on $D_{2\pi}$ and $D_{5\pi}$, respectively.
}}
\label{fig:LGCP-density}
\end{figure}

\pagebreak

\subsection{Additional simulations} \label{sec:TCP-mis2}

As discussed in Section \ref{sec:MS}, in case the model misspecifies the true spatial point pattern, the best fitting model may not always accurately estimates the true first-intensity. In this section, we provide a potential remedy to overcome this issue by fitting the ''reduced'' model.

For our simulation, we generate the same LGCP model on $\R^2$ as in Section \ref{sec:MS} and fit the Thomas clustering process (TCP) models with parameters $(\kappa, \alpha, \sigma^2)^\top$ as in Section \ref{sec:ex2-NS}. However, for each simulated point pattern, we constraint the parameter $\alpha = \widehat{\lambda} / \kappa$, where $\widehat{\lambda}$ is a nonparameteric unbiased estimator of the ''ture'' first-order intensity. We denote this model with constraint as the ''reduced'' TCP model. The reduced TCP model has two free parameters $\betaa = (\kappa, \sigma^2)^\top$ and the estimated first-order intensity for the fitted reduced TCP model is $\widehat{\lambda}$. Therefore, the reduced TCP model corrected estimates the true first-order intensity. 

For each simulation, we fit the reduced TCP model using three estimation methods: discrete version of our estimator as described in Section \ref{sec:practice}, maximum likelihood-based method using the log-Palm likelihood(ML; \cite{p:tan-08}), and the minimum contrast method (MC). When evaluating our estimator, we follow the guidelines as in Section \ref{sec:practice} and consider the two prespecified domains $D_{2\pi} = \{ \ob \in \R^2: 0.1\pi \leq \|\ob\|_{\infty} \leq 2\pi \}$ and $D_{5\pi} = \{ \ob \in \R^2: 0.1\pi \leq \|\ob\|_{\infty} \leq 5\pi \}$. 

Now, we consider the best fitting reduced TCP model. The (discretized) Whittle likelihood of the reduced model is
\begin{equation} \label{eq:L-R2}
L^{(R)}(\betaa) = \sum_{\ob_{\kb,A} \in D} \left( \frac{\widehat{I}_{h,n}(\ob_{\kb,A})}{f_{\btheta(\betaa)}^{(TCP)}(\ob_{\kb,A})} + \log f_{\btheta(\betaa)}^{(TCP)}(\ob_{\kb,A}) \right).
\end{equation} Here, $\btheta(\betaa) = (\kappa,  \widehat{\lambda}/ \kappa, \sigma^2)^\top$, $A \in \{10,20,40\}$ is the side length of the observation window, and $D \in \{D_{2\pi}, D_{5\pi}\}$. Then, we report $\widehat{\betaa}^{(R)} = \arg \min_{} L^{(R)}(\betaa)$. Since $\widehat{\lambda}$ varies by simulations, the best fitting reduced TCP parameters also varies by simulations. However, we note that under mild conditions, $\widehat{\lambda}$ consistently estimates the true first-order intensity $\lambda^{(true)}$, the ''ideal'' best reduced TCP parameters are $\betaa_0(D,A) = \arg \min \mathcal{L}^{(R)}(\betaa)$, where
\begin{equation} \label{eq:mathcalL-R2}
\mathcal{L}^{(R)}(\betaa) = \sum_{\ob_{\kb,A} \in D}\left( \frac{f(\ob_{\kb,A})}{f_{\widetilde{\btheta}(\betaa)}^{(TCP)}(\ob_{\kb,A})} + \log f_{\widetilde{\btheta}(\betaa)}^{(TCP)}(\ob_{\kb,A}) \right),
\end{equation} where $f$ is the true spectral density function as in (\ref{eq:fLGCP}) and $\widetilde{\btheta}(\betaa) = (\kappa,  \lambda^{(true)}/ \kappa, \sigma^2)^\top$.

Table \ref{tab:fTCP-mis2} summarizes parameter estimation results. The results are also illustrated in Figure \ref{fig:LGCP-density-2}. We note that as the observation domain increases, our estimators (for $D_{2\pi}$ and $D_{5\pi}$) tend to converge to the corresponding (ideal) best fitting reduced TCP parameters. Whereas, the standard errors of the ML and MC estimator for $\sigma^2$ does not seem to significantly decrease to zero even for the sample points for few thousands (corresponds to the observation domain $D_n = [-20,20]^2$). Moreover, there is no clear evidence in Table \ref{tab:fTCP-mis2} and Figure \ref{fig:LGCP-density-2} that the parameter estimates for the ML and MC converge to some fixed non- diverging or non-shrinking parameters.

\begin{table}[h]
    \centering
  \begin{tabular}{cc|cc|cccccc}
\multirow{2}{*}{Window} & \multirow{2}{*}{Par.} & \multicolumn{2}{c|}{Best Par.} &
\multicolumn{4}{c}{Method} \\
\cline{3-8}
& & $D_{2\pi}$ & $D_{5\pi}$ & Ours($D_{2\pi}$) & Ours($D_{5\pi}$)  & ML & MC \\ \hline \hline
 \multirow{3}{*}{$[-5,5]^2$}& $\kappa$ &  0.22    & 0.25  & 0.38(0.25)  & 0.42(0.35)   & 0.28(0.34) & 0.25(0.28) \\  
				     & $\sigma^2$& 0.09   &  0.08 & 0.13(0.12) & 0.12(0.13)  & 0.19(0.15) & 0.24(0.17) \\
				     & Time(sec) &  ---   & ---         & 0.14 & 0.71  & 0.30 & 0.06 \\ 
\cline{1-8}
 \multirow{3}{*}{$[-10,10]^2$}& $\kappa$ &  0.21 &  0.24 & 0.27(0.09)    & 0.30(0.11)   & 0.13(0.06) & 0.13(0.06) \\  
				     & $\sigma^2$    & 0.09  &  0.08 & 0.10(0.04) &  0.09(0.04)  & 0.34(0.17) & 0.36(0.25) \\ 
				     & Time(sec) &   ---  & ---            & 0.47         &  2.68           &  0.91 & 0.16 \\ 
\cline{1-8}
 \multirow{3}{*}{$[-20,20]^2$}& $\kappa$ &   0.21 & 0.24   & 0.23(0.05)  & 0.26(0.06) & 0.09(0.03) & 0.09(0.03) \\  
				     & $\sigma^2$ & 0.09 & 0.08      &0.10(0.02)   & 0.08(0.03) & 0.41(0.19) & 0.46(0.19) \\ 
				     & Time(sec) &  ---   & ---            & 1.76          & 11.66        & 13.21        & 1.38  \\ \hline
\cline{1-8}
\end{tabular} 
\caption{\textit{The mean and the standard errors (in the parentheses) of the estimated parameters for the misspecified LGCP fitting with the reduced TCP model. The best fitting parameters are calculated by minimizing $\mathcal{L}^{(R)}(\btheta)$ in (\ref{eq:mathcalL-R2}).
When evaluating our estimator, we use two different prespecified domains: $D_{2\pi}$ and $D_{5\pi}$. The time is calculated as an averaged computational time (using a parallel computing in R on a desktop computer with an i7-10700 Intel CPU) of each method per one simulation from 500 independent replications.
}}
\label{tab:fTCP-mis2}
\end{table}

\begin{figure}[ht!]
\centering
\includegraphics[width=0.9\textwidth]{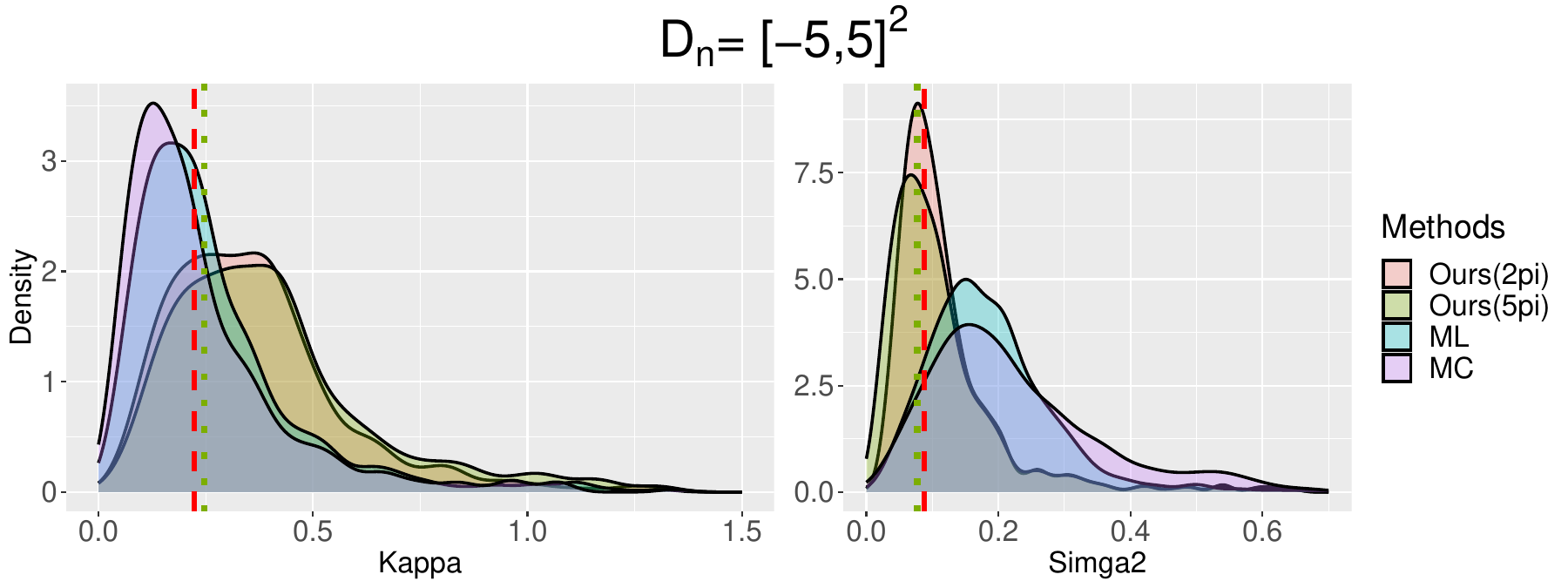}
\includegraphics[width=0.9\textwidth]{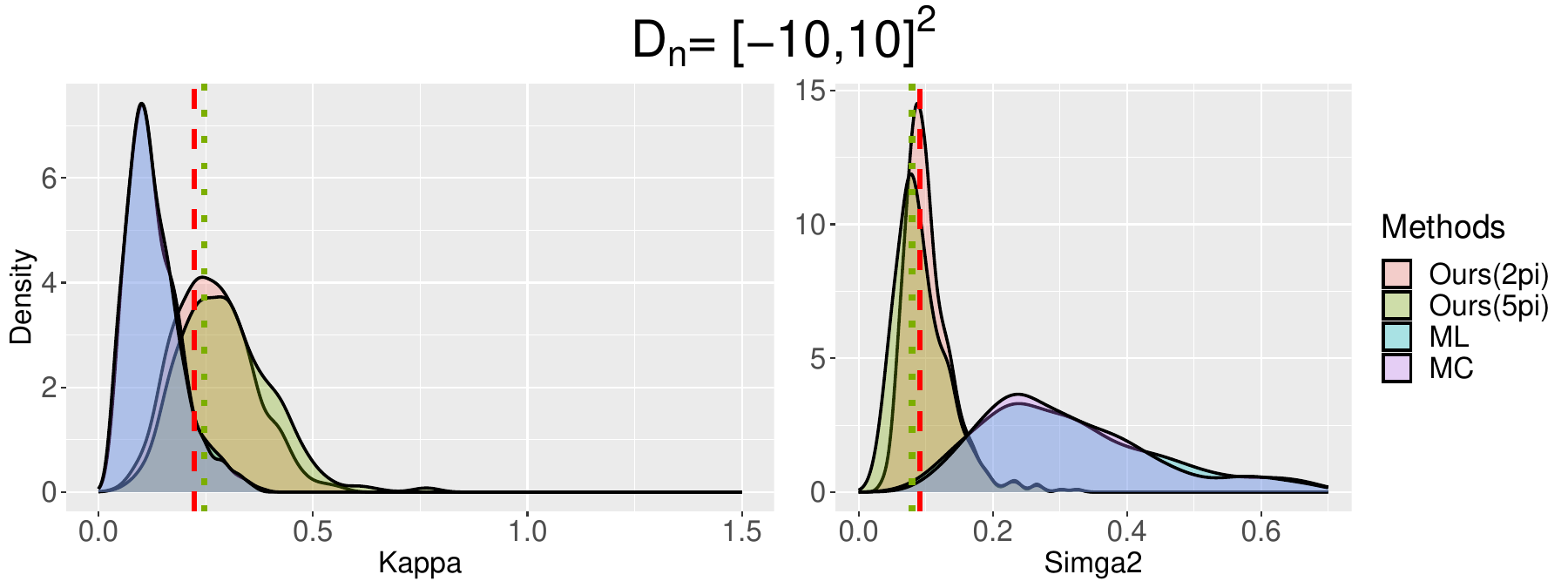}
\includegraphics[width=0.9\textwidth]{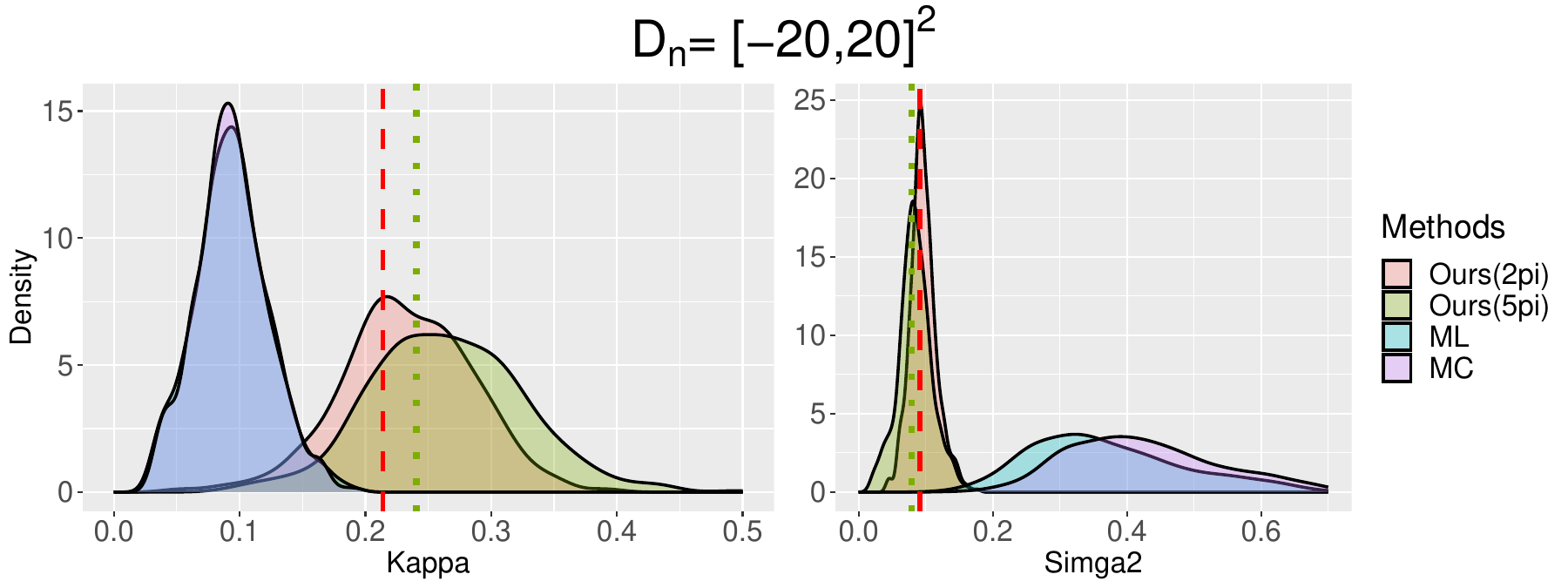}
\caption{\textit{
Densities of the estimated parameters for the misspecified LGCP fitting with the reduced TCP model described as in Section \ref{sec:TCP-mis2}. Each row refers to different observation domains. Dashed (green) and dotted (red) vertial lines refer to the best fitting TCP parameters as in (\ref{eq:mathcalL-R2}) evaluated on $D_{2\pi}$ and $D_{5\pi}$, respectively.
}}
\label{fig:LGCP-density-2}
\end{figure}



\section{Spectral methods for nonstationary point processes} \label{sec:IRS}

\subsection{A new DFT for the intensity reweighted process}
Recall the $n$th-order intensity function $\lambda_n$ as in (\ref{eq:lambda}). In this section, we do not presuppose the (second-order) stationarity of the point process $X$. Instead, we let $X$ be a simple second-order intensity reweighted stationary (SOIRS) point process on $\R^d$ (\cite{p:bad-00}). That is, there exists $\ell_2: \R^{d} \rightarrow \R$ such that
\begin{equation} \label{eq:gn}
\frac{\gamma_2(\xb_1,\xb_2)}{\lambda_1(\xb_1) \lambda_1 (\xb_2)} =  \ell_{2} (\xb_1-\xb_2), \quad \xb_1, \xb_2 \in \R^d,
\end{equation} where $\lambda_1(\cdot)$ is the first-order intensity that does not need to be a constant.
The second-order stationary point processes fit within this framework by setting $\ell_2 =  \lambda^{-2} \gamma_{2,\text{red}}$, where $\lambda$ is the constant first-order intensity and $\gamma_{2,\text{red}}$ is the reduced second-order cumulant intensity function.

To leverage the Fourier methods developed for the stationary case, we consider a slight variant of the ordinary DFT defined in (\ref{eq:mathcalDFT-h}). The analogous large sample results for the ordinary DFT under the SOIRS framework are similar, with greater details provided in \cite{p:dsy-24}.

\begin{definition}[Intensity reweighted DFT] \label{def:IRDFT}
Let $X$ be an SOIRS spatial point process on $D_n$ ($n\in \N$) of form (\ref{eq:Dn}). Then, the intensity reweighted DFT (IR-DFT) with the data taper $h$ is defined as 
\begin{equation} \label{eq:mathbbJn-lambda}
\mathcal{J}_{h,n}^{(IR)}(\ob; \lambda_1) = (2\pi)^{-d/2} H_{h, 2}^{-1/2} |D_n|^{-1/2} \sum_{\xb \in X_{} \cap D_n}
\frac{h(\xb/\aB)}{\lambda_1(\xb)}\exp(-i\xb^\top \ob),  \quad \ob \in \R^d.
\end{equation} 
\end{definition}


Before investigating the theoretical properties of the IR-DFT, we draw a comparison between the IR-DFT and the ordinary DFT.
Firstly, unlike the ordinary DFT, $\mathcal{J}_{h,n}^{(IR)}(\ob; \lambda_1)$ is contingent on the underlying unknown first-order intensity function. Secondly, under stationarity, $\mathcal{J}_{h,n}^{(IR)}(\ob;\lambda_1)$ and $\mathcal{J}_{h,n}(\ob)$ in (\ref{eq:mathcalDFT-h}) are related through $\mathcal{J}_{h,n}^{(IR)}(\ob; \lambda) = \lambda^{-1} \mathcal{J}_{h,n}(\ob)$, where $\lambda$ is the constant first-order intensity.
Lastly, by using (\ref{eq:lambda}), we have
$\Ex[\mathcal{J}_{h,n}^{(IR)}(\ob; \lambda_1)] = c_{h,n}(\ob)$,  where $c_{h,n}(\cdot)$ is the bias factor as defined in (\ref{eq:Cn-h}). Therefore, the expectation of the IR-DFT is a deterministic function depends solely on the data taper $h$ and the domain $D_n$.
 
By using the above bias expression, we now can define the theoretical centered IR-DFT and IR-periodogram respectively as
\begin{equation}
J_{h,n}^{(IR)}(\ob; \lambda_1) =
\mathcal{J}_{h,n}^{(IR)}(\ob; \lambda_1) - 
c_{h,n}(\ob)  \text{~~and~~}
I_{h,n}^{(IR)}(\ob; \lambda_1) = |J_{h,n}^{(IR)}(\ob; \lambda_1)|^2,  \quad \ob \in \R^d.
 \label{In-lambda}
\end{equation} 

\subsection{Asymptotic properties of the IR-DFT and IR-periodogram} \label{appen:IRDFT}

In this section, we study asymptotic properties for the IR-DFT and IR-periodogram. To do so, we adopt a different asymptotic framework compared to the stationary case. This is because if we only rely on Assumption \ref{assum:A} as our asymptotic setup for general SOIRS processes, there is no gain in information of $\lambda_1(\xb)$ at the fixed location $\xb \in \R^d$ as the domain $D_n$ increases. Therefore, in a similar spirit to \cite{p:dah-97}, we consider an infill-type asymptotic framework for the first-order intensity function below. For a domain $W \in \R^d$, we use the notation $X_{W}$ to indicate the observations of $X$ are confined within $W$. 

\begin{assumption} \label{assum:infill}
Let $X_{D_n}$ ($n \in \N$) be a sequence of SOIRS processes defined on the increasing domain $\{D_n\}$ of form (\ref{eq:Dn}). Let $\lambda_{1,n}(\cdot)$ and $\gamma_{2,n}(\cdot,\cdot)$ be the first- and second-order cumulant intensity functions of $X_{D_n}$, respectively. Then, the following structural assumptions on  $\lambda_{1,n}(\cdot)$ and $\gamma_{2,n}(\cdot,\cdot)$ hold:
\begin{itemize}
\item[(i)] For $n \in \N$, $\lambda_{1,n}(\cdot)$ is a strictly positive function on $D_n$ and there exists non-negative function $\lambda(\xb)$, $\xb \in \R^d$, with a compact support on $[-1/2,1/2]^d$, such that
\begin{equation}  \label{eq:lambda-infill}
\lambda_{1,n}(\xb) = \lambda(\xb/\aB), \quad n\in \N, \quad \xb \in D_n.
\end{equation} 

\item[(ii)] For $n \in \N$ and $\xb, \yb \in D_n$, $\gamma_{2,n}(\xb,\yb)/( \lambda_{1,n}(\xb) \lambda_{1,n}(\yb)) = \ell_2(\xb-\yb)$ where $\ell_2:\R^d \rightarrow \R$ does not depend on $n$.
\end{itemize}
\end{assumption}
Under Assumption \ref{assum:infill}(i), the IR-DFT can be written as
\begin{equation} \label{eq:JIR-22}
\mathcal{J}_{h,n}^{(IR)}(\ob; \lambda)  = (2\pi)^{-d/2} H_{h, 2}^{-1/2} |D_n|^{-1/2} \sum_{\xb \in X_{D_n}}
\frac{h(\xb/\aB)}{\lambda(\xb/\aB)}\exp(-i\xb^\top \ob), \quad \ob \in \R^d.
\end{equation} 
Here, we use the notation $\lambda$ instead of $\lambda_1$ to emphasize the asymptotic framework as in Assumption \ref{assum:infill}(i). The centered IR-DFT and IR-periodogram, denoted by $J_{h,n}^{(IR)}(\ob; \lambda)$ and $I_{h,n}^{(IR)}(\ob; \lambda)$, respectively, can be defined similarly.

Theorem \ref{thm:asympDFT-IRS} below addresses the asymptotic uncorrelatedness of the IR-DFTs.

\begin{theorem}[Asymptotic uncorrelatedness of the IR-DFT] \label{thm:asympDFT-IRS}
Let $X_{D_n}$ ($n \in \N$) be a sequence SOIRS point processes that satisfy Assumption \ref{assum:infill}. Suppose that Assumptions \ref{assum:A}, \ref{assum:C} (for $\ell=2$), and Assumption \ref{assum:E}(i) hold. Furthermore, $\lambda(\cdot)$ from (\ref{eq:lambda-infill}) is strictly positive and continuous on $[-1/2,1/2]^d$.
 Let $\{\ob_{1,n}\}$ and $\{\ob_{2,n}\}$ be two asymptotic distant sequencies on $\R^d$. Then,
\begin{equation} \label{eq:lim-DFT-IRS}
\lim_{n\rightarrow \infty} \cov(J_{h,n}^{(IR)}(\ob_{1,n};\lambda), J_{h,n}^{(IR)}(\ob_{2,n};\lambda)) = 0.
\end{equation}
If we further assume $\lim_{n \rightarrow \infty} \ob_{1,n} = \ob \in \R^d$, then
\begin{equation} \label{eq:lim-DFT2-IRS}
 \lim_{n\rightarrow \infty} \var(J_{h,n}^{(IR)}(\ob_{1,n};\lambda)) 
= \lim_{n \rightarrow \infty} \Ex[I_{h,n}^{(IR)}(\ob; \lambda)] =
 (2\pi)^{-d} \frac{H_{h^2/\lambda,1}}{ H_{h,2}} + \mathcal{F}^{-1}(\ell_2)(\ob). 
\end{equation}
\end{theorem}
\begin{proof} 
To prove the theorem, we first start with the expression of the covariance of the IR-DFT. The proof of lemma below is almost identical to that of the proof of Lemma \ref{thm:DFT2} so we omit the details.
\begin{lemma} \label{coro:DFT2-IRS}
Let $X_{D_n}$ ($n \in \N$) be a sequence of SOIRS spatial point processes that satisfy Assumption \ref{assum:infill}
and let $h$ be data taper such that $\sup_{\ob \in \R^d} h(\xb) <\infty$. 
Suppose that Assumption \ref{assum:C} for $\ell=2$ holds. Then, 
\begin{equation}
\begin{aligned} 
&\cov(J_{h,n}^{(IR)}(\ob_1; \lambda), J_{h,n}^{(IR)}(\ob_2; \lambda)) = (2\pi)^{-d} H_{h,2}^{-1} |D_n|^{-1}
\bigg(
 H_{h^2/\lambda,1}^{(n)}(\ob_1-\ob_2) \\
&\quad + 
\int_{D_n^2} h(\xb/\aB) h(\yb/\aB) e^{-i(\xb^\top \ob_1 - \yb^\top \ob_2)} \ell_{2}(\xb-\yb) d\xb d\yb\bigg),
\quad \ob_1, \ob_2 \in \R^d.
\end{aligned}
\label{eq:DFT2-lambda1}
\end{equation}
\end{lemma}
Now, by utilizing expression (\ref{eq:DFT2-lambda1}) above, proofs of (\ref{eq:lim-DFT-IRS}) and (\ref{eq:lim-DFT2-IRS}) are almost identical to those in the proof of \cite{p:dsy-24}, Theorem 4.1 (we omit the details). 
\end{proof}

Under second-order stationarity, an expectation of the periodogram converges to the spectral density function. Bearing this in mind, along with the limiting behavior in (\ref{eq:lim-DFT2-IRS}), we define the intensity reweighted pseudo-spectral density functions SOIRS processes.

\begin{definition}[Intensity reweighted pseudo-spectral density function]
Let $X_{D_n}$ ($n \in \N$) be a sequence of SOIRS spatial point processes that satisfy Assumption \ref{assum:infill}. Suppose $\ell_2$ in (\ref{eq:gn}) belongs to $L^1(\R^d)$. Then, the intensity reweighted pseudo-spectral density function (IR-PSD) of $X_{D_n}$ corresponding to the data taper $h$ is defined as
\begin{equation} \label{eq:f-IRS}
f_{h}^{(IR)}(\ob) =  (2\pi)^{-d} \frac{H_{h^2/\lambda,1}}{ H_{h,2}} + 
\mathcal{F}^{-1}(\ell_2)(\ob)
, \quad \ob \in \R^d.
\end{equation}
\end{definition}
It follows from (\ref{eq:lim-DFT2-IRS}) that $f_{h}^{(IR)}$ is an even and non-negative function on $\R^d$. However, unlike the classical spectral density function, the IR-PSD $f_h^{(IR)}$ depends on the specify data taper function $h$. Under stationarity, $f_h^{(IR)}$ equals $\lambda^{-2} f$, where $f$ is the spectral density.

In the theorem below, we derive the asymptotic joint distribution of the theoretical IR-DFTs and IR-periodograms. 
The proof is almost identical to that of Theorem \ref{thm:asymp-DFT}, so we omit the details.

\begin{theorem}[Asymptotic joint distribution of the IR-DFTs and IR-periodograms] \label{thm:asymp-DFT-IRS}
Let $X_{D_n}$ ($n \in \N$) be a sequence of SOIRS spatial point processes that satisfy Assumption \ref{assum:infill}. Suppose that Assumptions \ref{assum:A}, \ref{assum:C} (for $\ell=4$), \ref{assum:D}(i), and \ref{assum:E}(i) hold.
Furthermore, $\lambda(\cdot)$ from (\ref{eq:lambda-infill}) is strictly positive and continuous on $[-1/2,1/2]^d$.
For a fixed $r\in \N$, $\{\ob_{1,n}\}$, ..., $\{\ob_{r,n}\}$ denote $r$ sequences on $\R^d$ that satisfy conditions (1) and (3) in the statement of Theorem \ref{thm:asymp-DFT}. Then,
\begin{equation*}
\left( \frac{J_{h,n}^{(IR)}(\ob_{1,n};\lambda)}{ (\frac{1}{2}f_h^{(IR)}(\ob_1))^{1/2}}, \dots, \frac{J_{h,n}^{(IR)}(\ob_{r,n};\lambda)}{( \frac{1}{2}f_h^{(IR)}(\ob_r))^{1/2}} \right)
 \Dcon (Z_1, \dots, Z_r), \quad n \rightarrow \infty,
\end{equation*} where $\{Z_k\}_{k=1}^{r}$ are independent standard normal random variables on $\C$. By using continuous mapping theorem, we conclude 
\begin{equation*}
\left( \frac{I_{h,n}^{(IR)}(\ob_{1,n};\lambda)}{\frac{1}{2}f_h^{(IR)}(\ob_1)}, \dots,
\frac{I_{h,n}^{(IR)}(\ob_{r,n};\lambda)}{\frac{1}{2}f_h^{(IR)}(\ob_r)}
 \right)
 \Dcon (\chi^2_1, \dots, \chi^2_r), \quad n \rightarrow \infty,
\end{equation*} where $\{\chi^2_k\}_{k=1}^{r}$ are independent chi-squared random variables with degrees of freedom two.
\end{theorem}

\end{document}